\documentclass[a4paper,11pt]{article}
\usepackage{jheppub} 
\usepackage{amssymb,amsmath,bm}
\usepackage{amsthm}
\usepackage{physics}
\usepackage{tikz}
\usepackage{mathrsfs}
\usetikzlibrary{positioning,decorations,decorations.pathmorphing,decorations.shapes,shapes.symbols,shapes.geometric,arrows.meta,calc,calligraphy,fit,tikzmark,overlay-beamer-styles}
\usepackage{subcaption}
\usepackage{graphicx}  
\usepackage{multicol,multirow}
\usepackage{booktabs}
\usepackage{pifont}
\usepackage{comment}
\usepackage{cleveref}
\crefname{equation}{Eq.}{Eqs.}
\crefname{figure}{Fig.}{Figs.}
\crefname{section}{Sec.}{Secs.}
\crefname{appendix}{App.}{Apps.}
\crefname{lemma}{Lem.}{Lems.}
\crefname{proposition}{Prop.}{Props.}
\crefname{table}{Tbl.}{Tbls.}
\crefname{theorem}{Thm.}{Thms.}
\crefname{conjecture}{Conj.}{Conjs.}
\crefname{definition}{Def.}{Defs.}
\usepackage{rotating}
\usepackage{float}
\usepackage{adjustbox}
    
\newcommand{\cS}{\mathcal S}

\DeclareMathOperator{\dlog}{Li}
\newcommand{\ord}{\mathcal O}
\newcommand{\half}{\frac{1}{2}}
\newcommand{\textoverline}[1]{$\overline{\mbox{#1}}$}

\DeclareMathOperator{\Gr}{Gr}
\DeclareMathOperator{\Fl}{Fl}

\newtheorem*{theorem*}{Theorem}
\newtheorem{theorem}{Theorem}[section]
\newtheorem{lemma}[theorem]{Lemma}
\newtheorem{proposition}[theorem]{Proposition}

\newtheorem{conjecture}[theorem]{Conjecture}
\theoremstyle{definition}
\newtheorem{definition}[theorem]{Definition}

\newtheorem{remark}[theorem]{Remark}

\title{\boldmath Landau and cluster structures of one-loop amplitudes in ${\mathcal N}=4$ SYM in dimensional regularization}

\usepackage{orcidlink}
\author[\!a \orcidlink{0009-0002-5095-8308}]{Athanasia-Konstantina Angelopoulou,}

\author[\!a \orcidlink{0000-0003-2462-6481}]{Ruth Britto,}

\author[b \orcidlink{0000-0001-9556-9555}]{Matteo Parisi}

\affiliation{$^a$School of Mathematics and Hamilton Mathematics Institute, Trinity College, Dublin 2, Ireland}

\affiliation{$^b$Okinawa Institute of Science and Technology (OIST), Onna, Okinawa, Japan}

\emailAdd{angelopa@tcd.ie}
\emailAdd{britto@maths.tcd.ie}
\emailAdd{matteo.parisi@oist.jp}

\abstract{Generalized unitarity, Landau analysis, and cluster adjacency encode complementary aspects of scattering amplitudes. We use one-loop planar $\mathcal N=4$ SYM amplitudes in dimensional regularization, at arbitrary multiplicity and helicity, to make their interface explicit. The weight-two symbol decomposes into an \emph{LS part}, in which maximal-cut leading singularities furnish the coefficients and Landau loci associated with nested cuts organize the ordered symbol entries, an algebraic four-mass sector, and residual terms. Cancellations of certain letters contributed by individual box integrals, as well as further simplifications, are explained by the \emph{two-mass triangle relations} among box coefficients. We prove these relations using a BCFW-like application of the global residue theorem and show that they can be understood geometrically as different dissections of the same region in the tree amplituhedron obtained by projecting the loop geometry of a triple cut. We then prove that the full rational symbol, including its infrared-divergent part, obeys cluster adjacency in the flag cluster algebra $\mathrm{Fl}_{2,4;n}$ for all multiplicities and helicities. Within the sector depending only on momentum-twistor four-brackets, we conjecture a stronger cluster property in $\operatorname{Gr}(4,n)$ for all helicities and prove it for NMHV amplitudes: the amplitude admits a representation in which every pole of each coefficient is compatible with both symbol entries. Finally, we observe that the algebraic four-mass letters, although non-rational, exhibit a suggestive Sklyanin-bracket pattern.}

\begin{document}
\maketitle
\flushbottom

\section{Introduction}
\label{sec:intro}

Scattering amplitudes admit several complementary descriptions of their analytic structure: generalized unitarity organizes them in terms of master integrals
and leading singularities;
Landau analysis studies candidate singularities through the geometry of
cut equations;
while the cluster bootstrap constrains the integrated answer through
symbol alphabets, Steinmann relations, and cluster adjacency.

In this paper we use one-loop amplitudes in planar $\mathcal N=4$
super Yang--Mills theory as a setting in which these three perspectives
can be compared directly. 
The master integrals are simply scalar boxes \cite{Bern:1994zx}, and their coefficients are determined from tree-level amplitudes by quadruple cuts in generalized unitarity \cite{Britto:2004nc}.  
Moreover, infrared consistency gives linear
relations among them \cite{Roiban:2004ix}, while anomalous
dual conformal symmetry imposes further constraints
\cite{Brandhuber:2009rel,
Brandhuber:2009proof}. As a result, 
the symbol of the amplitude is considerably simpler than the symbols of the separate box functions would suggest.
We would like to determine which features
of the integrated amplitude are directly controlled by maximal cuts and
leading singularities, how relations among different cuts lead to
cancellations, what residual structure remains in the full symbol, and what cluster structure it exhibits.

A complementary description of the amplitude starts from the singularities themselves \cite{Hannesdottir:2024hke}.
The Landau equations give necessary conditions for singularities of
Feynman integrals, and their solutions can be organized geometrically as
discriminant loci in the space of external kinematics
\cite{Landau1959,Pham1967,MizeraTelen2022}. In planar $\mathcal N=4$ SYM,
Landau singularities have been compared directly with symbol alphabets
\cite{DennenSpradlinVolovich2016} and studied through the geometry of
the amplituhedron
\cite{Dennen:2016mdk}. The amplituhedron
analysis was subsequently extended to arbitrary helicity, recovering
the known branch points of all one-loop amplitudes
\cite{Prlina:2017azl}.

The relation between Landau geometry and symbols is subtler than a
simple identification of Landau loci with symbol letters. The local
behavior of an integral near a Landau singularity constrains the
positions at which the corresponding vanishing loci occur in its
symbol \cite{Hannesdottir:2021kpd}, while the geometry
of on-shell spaces places further constraints on allowed sequences of
discontinuities
\cite{Bourjaily:2020wvq,HannesdottirMcLeodSchwartzVergu:2022,Hannesdottir:2024cnn}, including constraints valid to all orders in dimensional regularization \cite{Bargiela:2026lje}.
Geometric
Landau analysis can be combined with information about positive
geometries to distinguish physical from spurious Landau solutions and
to determine symbol alphabets which can subsequently be used as input
for a bootstrap \cite{ChicherinEtAl:2025}. The work by one of the authors with Hollering, Mazzucchelli and Sturmfels develops Landau analysis in
Grassmannian geometry
\cite{HolleringMazzucchelliParisiSturmfels:2026Landau,
HolleringMazzucchelliParisiSturmfels:2026Positivity}. In this framework,
the degree of a maximal cut counts its leading singularities, while
\emph{leading-singularity discriminants}, or \emph{LS discriminants},
detect degenerations in which leading singularities collide. For a
large class of rational Landau problems, these discriminants exhibit
positivity and factorization into cluster variables. 

Cluster coordinates on $\operatorname{Gr}(4,n)$ have been observed to organize much
of the singularity and polylogarithmic structure of planar
$\mathcal N=4$ SYM amplitudes
\cite{GoldenGoncharovSpradlinVerguVolovich:2013}, and cluster adjacency
was subsequently proposed as a restriction on consecutive symbol
letters \cite{Drummond:2017ssj}. These ideas have become important ingredients of a bootstrap program that largely fixes the particle multiplicity and pushes to increasingly high loop order, most notably for six- and seven-particle amplitudes \cite{boot_review2020}. There is, however, a complementary
all-multiplicity story at low loop order. Suitably regulated one- and
two-loop MHV amplitudes satisfy cluster adjacency at arbitrary
multiplicity \cite{Golden:2019kks}, while the
all-multiplicity one-loop NMHV ratio function satisfies Steinmann cluster
adjacency \cite{Mago:2020eua}. The latter work
also established adjacency between Yangian-invariant coefficients and
final symbol entries for the BDS-like normalized one-loop NMHV amplitude
through nine points.

Relatedly, in \cite{Gurdogan:2020tip} a direct cluster relation was
proposed between the Landau singularities of a maximal cut and the poles
of each Yangian invariant appearing in the corresponding leading
singularity. We refer to this relation as \emph{LL-adjacency}. This
suggests that cluster structure may organize simultaneously the
kinematic singularities appearing as symbol letters and the rational
singularities carried by their leading-singularity coefficients.

In the present paper we take a different point of view from
much of the bootstrap literature. We work with the \emph{full
dimensionally regulated one-loop amplitude}, without passing to a ratio
function or BDS-like normalized quantity. This
keeps the connection with generalized unitarity and the scalar-box
expansion completely explicit, and retains rather than removes the
infrared-divergent sector. Our results are therefore complementary to
the usual high-loop bootstrap program: we keep the loop order fixed and
derive statements valid at arbitrary multiplicity and in all helicity sectors.

Keeping the full amplitude also changes the appropriate cluster
structure. Dual-conformal quantities are naturally expressed in terms of
momentum-twistor four-brackets and the cluster algebra of
$\operatorname{Gr}(4,n)$. The full dimensionally regulated amplitude
additionally contains spinor-helicity two-brackets, or equivalently
momentum-twistor brackets involving the infinity twistor. The natural
space containing both types of coordinates is therefore the
momentum-twistor partial flag variety
$
\operatorname{Fl}_{2,4;n}$.
Cluster structures on spinor-helicity and momentum-twistor partial flag
varieties have recently been constructed
\cite{BossingerLi:2024}, and related flag cluster algebras have begun to
organize symbol alphabets for massless scattering beyond the
dual-conformal setting
\cite{PokrakaSpradlinVolovichWeng:2025}.

Our first goal is to make the relation between cuts, Landau
singularities and symbol words precise. We introduce the
\emph{leading-singularity part}, or \emph{LS part}, of a symbol: the terms whose rational coefficients are leading singularities of maximal cuts, and whose ordered
sequence of letters can be associated to Landau singularities of a
nested sequence of cuts terminating on that maximal cut. In this way,
the ordering of the cut hierarchy is reflected in the ordering of the
symbol entries, in a way closely related in spirit to the
constraints on ordered discontinuities of
\cite{Abreu:2017ptx, Hannesdottir:2021kpd,
HannesdottirMcLeodSchwartzVergu:2022,Berghoff:2022mqu}, but applied here directly
to the generalized-unitarity representation of the full amplitude. For the scalar-box representation, this construction gives an explicit
all-multiplicity and helicity-independent decomposition 
\begin{equation}
\label{eq:intro-symbol-decomposition}
\mathcal S[A_{n,k}^{(1)}]
=
\mathcal S^{\rm LS}[A_{n,k}^{(1)}]
+
\mathcal S^{\rm 4m}
+
\sum_{i,j,l}
d_{ijl}\,
x_{ij}^{2}\otimes x_{il}^{2}.
\end{equation}
This decomposition is most meaningful at weight two, since lower-weight terms are tightly controlled by universal infrared behaviour.
The first term here has final entries which are the Landau discriminants associated with the corresponding maximal cuts, while earlier entries
arise from Landau loci of less restrictive cuts in the nested sequence. The second term in the above decomposition is
qualitatively different: for a generic four-mass box the maximal-cut
solutions are algebraic, 
rather than rational. Thus,
this sector must be isolated throughout. The third term is what remains beyond the leading-singularity
picture, and part of it can be simplified using relations among box coefficients.

A particularly important family of such identities is given by what we
call the \emph{two-mass triangle relations}. These relations arose previously in the study of the one-loop
dual conformal anomaly
\cite{Brandhuber:2009rel,
Brandhuber:2009proof}. Here we give a new derivation
which reveals a stronger statement. After cutting the three propagators of a two-mass
triangle, the remaining cut is a one-dimensional on-shell variety with two
components. 
The global residue theorem can be applied separately on each of the two branches. Summing the two
branchwise identities gives a two-mass triangle relation among box
coefficients. At the level of the symbol, these relations explain the
cancellation of the letters involving differences of Mandelstam invariants
and simplify part, but not all, of the residual sector. The branchwise form
also suggests a geometric interpretation: intersecting the one-loop
amplituhedron with a non-maximal cut and projecting to the tree
amplituhedron, different completions of the cut to maximal cuts give
collections of regions whose canonical forms are the corresponding leading
singularities. The two sides of the residue identity can then be read as two
dissections of the same projected cut region. We verify this picture in
cases where the relevant projected geometries can be controlled explicitly.

Our second main set of results concerns cluster structure. After
rewriting dual coordinates in momentum-twistor variables and expanding
symbol letters multiplicatively, we prove ordinary cluster adjacency for
the \emph{rational sector of the full dimensionally regulated one-loop
amplitude}, for arbitrary multiplicity and all helicity sectors, in the
flag cluster algebra $\operatorname{Fl}_{2,4;n}$. In particular, this
statement retains the infinity-twistor letters associated with the
infrared-divergent part of the amplitude.

The four-mass algebraic sector requires separate treatment. From the viewpoint of the geometric Landau analysis of \cite{HolleringMazzucchelliParisiSturmfels:2026Landau,HolleringMazzucchelliParisiSturmfels:2026Positivity}, its discriminant is an LS discriminant associated with a non-rational maximal-cut fiber. Although positivity of the LS discriminant continues to hold in this case, the discriminant is not itself a cluster variable. In the integrated amplitude it  enters instead through algebraic letters involving its square root. We nevertheless find that the \emph{Sklyanin bracket test} \cite{Golden:2019kks, HeLiYang2022Constraints}, usually employed as a test of compatibility between rational cluster coordinates, exhibits a remarkably simple behavior when applied directly to the algebraic four-mass letters. To our knowledge, this is the first such application of the Sklyanin bracket to algebraic symbol letters. We regard this as evidence that the Poisson structure underlying cluster adjacency may admit a natural extension beyond the rational cluster alphabet.

Finally, we refine ordinary symbol adjacency by retaining information
about the rational coefficients of the words. After expressing the
amplitude in momentum-twistor variables and multiplicatively expanding
the symbol, we restrict to the sector in which both symbol entries are
momentum-twistor four-brackets, rather than infinity-twistor
two-brackets.
For the one-loop NMHV amplitude, at arbitrary multiplicity, we show that
this sector admits a representation in which the poles of each
$R$-invariant coefficient are compatible with both entries of the
corresponding symbol word. For LS terms, compatibility with the final
entry refines the relation between Landau and leading singularities
suggested by LL-adjacency, while compatibility with the first entry can
be analyzed directly from the box structure. For the non-LS sector,
identities among $R$-invariants give a representation in which the same
compatibility becomes manifest.

The paper is organized as follows. In \cref{sec:backgr} we review the
one-loop box expansion, cuts, leading and Landau singularities, symbols,
infrared behavior, and dual conformal symmetry. In \cref{sec:res} we
define the LS part of the symbol and derive the decomposition of
\cref{eq:intro-symbol-decomposition}, together with the coefficients of
all rational symbol words. In \cref{sec:2mtri} we study the two-mass
triangle relations, their role in simplifying the one-loop symbol,
their branchwise proof from on-shell forms, their amplituhedron
interpretation, and further relations among box coefficients. In
\cref{sec:clust} we study the flag-cluster structure of the full
one-loop amplitude, the algebraic four-mass sector, and cluster
adjacency involving rational coefficients. We conclude in
\cref{sec:discussion} with open questions concerning the passage from
Landau geometry to integrated amplitudes and its extension to higher
loops.

\section{One-loop amplitudes, singularities, and symbols}
\label{sec:backgr}
We consider one-loop superamplitudes in planar $\mathcal{N} = 4 \text{ SYM}$ with $n>4$ external legs, taken in dimensional regularization with $D=4-2\epsilon$. Throughout, all particle labels are understood $\bmod \;n$. For the external momenta, we use dual coordinates $x_i$, defined by
\begin{equation}\label{eq:dualcoord}
    p_i = x_i - x_{i+1},
    \qquad
    x_{ij}:= x_i - x_j =p_i+p_{i+1}+\cdots + p_{j-1},
\end{equation}
where $p_k$ is the incoming momentum of leg $k$.

The sixteen on-shell states of the theory are packaged into a single on-shell superfield $\Phi(p_i,\eta_i)$ using Grassmann variables $\eta_{i}^{\alpha}$, labeled by the $\mathrm{SU}(4)$ R-symmetry index $\alpha=1,\ldots,4$. Each state appears as the coefficient of a distinct monomial in $\eta_{i}^\alpha$, with the positive-helicity gluon at degree zero and the negative-helicity gluon at degree four. A superamplitude $A_n(p_i,\eta_i)$ thus encodes all component amplitudes at once, and its expansion in overall Grassmann degree $4(k+2)$ defines the N$^k$MHV sectors \cite{bible}, 
\begin{equation}
    A_n = A_{n,\mathrm{MHV}} + A_{n,\mathrm{NMHV}} + \cdots +A_{n,\overline{\mathrm{MHV}}}.
\end{equation}
After separating color factors, the superamplitude up to order $\epsilon^0$ can be expressed as a linear combination of scalar box functions  $F_{ijkl}$ \cite{Bern:1994zx},
\begin{equation}\label{eq:boxdecomp}
    A^{(1)}_n=\sum_{i,j,k,l} c_{ijkl} F_{ijkl} + \mathcal{O}(\epsilon),
\end{equation}
where the superscript $(1)$ on the amplitude refers to one-loop order. The sum runs over all partitions of the $n$ cyclically ordered external legs into four non-empty sets of adjacent legs. Each partition is labeled by the four indices $(i,j,k,l)$, denoting the first leg of each set: the box $F_{ijkl}$ carries the legs
\begin{equation}\label{eq:corners}
    \{i,\ldots,j-1\},\quad
    \{j,\ldots,k-1\},\quad
    \{k,\ldots,l-1\},\quad
    \{l,\ldots,i-1\}
\end{equation}
at its four corners, so that the corner momenta are
$x_{ij}$, $x_{jk}$, $x_{kl}$ and $x_{li}$, respectively. A corner
carrying a single leg is massless. The box functions $F_{ijkl}$ are obtained from the scalar box integrals $I_{ijkl}$, through the relations
\begin{align}
F_{ijkl}&= -\frac{\Gamma(1-2\epsilon)}{2\Gamma(1+\epsilon)\Gamma^2(1-\epsilon)}
\sqrt{\rho_{ijkl}}I_{ijkl},\label{eq:boxfcns}\\ I_{ijkl}&=-i(4\pi)^{2-\epsilon}\int  \frac{\mathrm{d}^{4-2\epsilon}\ell}{(2\pi)^{4-2\epsilon}}\frac{1}{\ell^2(\ell+x_{ij})^2(\ell+x_{ik})^2(\ell+x_{il})^2},
\end{align}
where 
\begin{equation}\label{eq:jacobrho}
\rho_{ijkl}
:=
\left(
  x_{ik}^2x_{jl}^2
  -x_{ij}^2x_{kl}^2
  -x_{jk}^2x_{il}^2
\right)^2
-4x_{ij}^2x_{jk}^2x_{kl}^2x_{il}^2 .
\end{equation}
The box coefficients $c_{ijkl}$ can be found by performing a maximal cut, giving a leading singularity of the amplitude.
\subsection{Leading and Landau singularities}
We first state our conventions for cuts, leading singularities, and
Landau singularities.  Let $G$ be a Feynman
diagram, and let $E$ be a collection of internal edges of $G$.  For
fixed external kinematics, the cut associated to $E$ is the subvariety
\begin{equation}
  \mathcal C_E
  :=
  \left\{ \ell \;:\; q_e(\ell)^2=0 \ \text{for all } e\in E \right\}
\end{equation}
of complexified loop-momentum space $\mathcal{L}$.  Thus, on $\mathcal C_E$, the
momenta flowing through all edges in $E$ are on shell.  We
will say that $\mathcal C_E$ is a maximal cut if, for generic external
kinematics, it is zero-dimensional when $\epsilon=0$.  In that case the loop momenta are
localized to a finite set of points
\begin{equation}
  \mathcal C_E=\left\{
  \ell_*^{(1)},\ldots,\ell_*^{(r_E)}
  \right\},
\end{equation}
where $r_E$ denotes the number of solutions, counted with the
appropriate multiplicities. In the planar case, in addition to the dual variables associated to the external faces, we introduce the dual variables $y_1,\ldots,y_L$ for the internal faces. If an edge $e$ separates the internal faces $y_a$ and $y_b$, then its inverse propagator is $q_e^2=(y_a-y_b)^2$.
If instead $e$ separates the external face $x_i$ from the internal face $y_a$, then $q_e^2=(x_i-y_a)^2$.

In planar $\mathcal N=4$ SYM we may apply the same language directly to
a fixed representation of the full integrand $\mathcal I[A^{(L)}_{n,k}]$ \cite{ArkaniHamedBourjailyCachazoCaronHuotTrnka2011}:
a propagator cut is the vanishing locus of a collection of equations of
the form
\begin{equation}
  (x_i-y_a)^2=0,
  \qquad
  (y_a-y_b)^2=0.
\end{equation}

The leading singularities of a fixed integrand representation are the
multivariate residues of the corresponding rational top form on maximal
cuts \cite{Cachazo2008,Arkani-Hamed:2009ljj}.  More precisely, if $\mathcal C_E$ is maximal and
$\ell_*^{{\sigma}}\in \mathcal C_E$, we write
\begin{equation}
  \mathrm{LeS}^{{\sigma}}_{\mathcal C_E}
  \bigl[\mathcal A^{(L)}_{n,k}\bigr]
  :=
  \operatorname*{Res}_{\ell=\ell_*^{{\sigma}}}
  \mathcal I[A^{(L)}_{n,k}] .
\end{equation}
For example, suppose that the cut is locally a transverse complete
intersection defined by $4L$ independent inverse propagators
$f_a(\ell)$ and that
\begin{equation}
\mathcal I[A^{(L)}_{n,k}]
=
\frac{N(\ell)\,\mathrm d^{4L}\ell}
     {f_1(\ell)\cdots f_{4L}(\ell)}
+\cdots .
\end{equation}
Then, up to the orientation of the residue contour,
\begin{equation}\label{eq:resjacobian}
\mathrm{LeS}^{{\sigma}}_{\mathcal C_E}
\bigl[A^{(L)}_{n,k}\bigr]
=
\left.
\frac{N(\ell)}
     {\det\!\left(
        \partial f_a/\partial\ell_A
      \right)}
\right|_{\ell=\ell_*^{{\sigma}}} ,
\end{equation}
where $\ell_A$ runs over the $4L$ loop-momentum components. 

\paragraph{One-loop maximal cuts and box coefficients.}

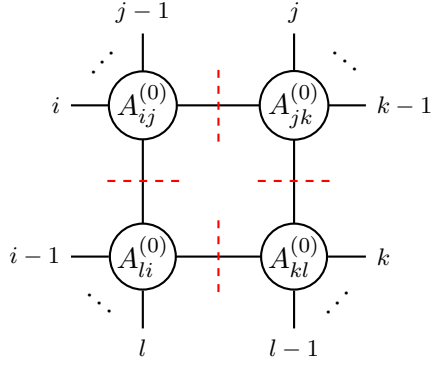
\begin{figure}
\centering
\begin{tikzpicture}
\def\rblob{0.45}
\def\rext{0.5}
\def\l{0.48}
\coordinate (A) at (-1,-1);
\coordinate (B) at (1,-1);
\coordinate (C) at (1,1);
\coordinate (D) at (-1,1);
\draw[thick] (A) -- (B) -- (C) -- (D) -- cycle;
\node[circle, draw, fill=white, thick, minimum size=0.5cm,inner sep=0pt] (a) at (A) {${A}_{li}^{(0)}$};
\node[circle, draw, fill=white, thick, minimum size=0.5cm,inner sep=0pt] (b) at (B) {${A}_{kl}^{(0)}$};
\node[circle, draw, fill=white, thick, minimum size=0.5cm,inner sep=0pt] (c) at (C) {${A}_{jk}^{(0)}$};
\node[circle, draw, fill=white, thick, minimum size=0.5cm,inner sep=0pt] (d) at (D) {${A}_{ij}^{(0)}$};

\draw[thick] ($ (A) + (-\rblob,0) $) -- ++(-\rext,0) node(g)[left]{\footnotesize $i-1$}; 
\draw[thick] ($ (A) + (0,-\rblob) $) -- ++(0,-\rext) node(h)[below]{\footnotesize $l$}; 
\path (g) -- node[sloped,xshift=2pt,yshift=2pt]{$\cdots$} (h);

\draw[thick] ($ (B) + (\rblob,0) $) -- ++(\rext,0) node(e)[right]{\footnotesize $k$}; 
\draw[thick] ($ (B) + (0,-\rblob) $) -- ++(0,-\rext) node(f)[below]{\footnotesize $l-1$}; 
\path (e) -- node[sloped]{$\cdots$} (f); 

\draw[thick] ($ (C) + (\rblob,0) $) -- ++(\rext,0) node(d)[right]{\footnotesize $k-1$}; 
\draw[thick] ($ (C) + (0,\rblob) $) -- ++(0,\rext) node(c)[above]{\footnotesize $j$}; 
\path (c) -- node[sloped,xshift=2pt,yshift=-2pt]{$\cdots$} (d);

\draw[thick] ($ (D) + (-\rblob,0) $) -- ++(-\rext,0) node(b)[left]{\footnotesize $i$}; 
\draw[thick] ($ (D) + (0,\rblob) $) -- ++(0,\rext) node(a)[above]{\footnotesize $j-1$}; 
\path (a) -- node[sloped,xshift=2pt,yshift=-2pt]{$\cdots$}  (b);

\coordinate (M1) at ($ (A)!0.5!(B) $);
\draw[dashed,red,thick] ($ (M1) + (0,\l) $) -- ($ (M1) + (0,-\l) $);
\coordinate (M2) at ($ (B)!0.5!(C) $);
\draw[dashed,red,thick] ($ (M2) + (-\l,0) $) -- ($ (M2) + (\l,0) $);
\coordinate (M3) at ($ (C)!0.5!(D) $);
\draw[dashed,red,thick] ($ (M3) + (0,-\l) $) -- ($ (M3) + (0,\l) $);
\coordinate (M4) at ($ (D)!0.5!(A) $);
\draw[dashed,red,thick] ($ (M4) + (\l,0) $) -- ($ (M4) + (-\l,0) $);
\end{tikzpicture}
\caption{Maximal cut of the one-loop amplitude used to extract the box coefficient $c_{ijkl}$ of \cref{eq:boxcoef}. The dashed red lines indicate that the propagators are on-shell, which localizes the loop momentum onto the two solutions $\ell_*^{\pm}$. On the cut, the integrand factorizes into the four tree-level superamplitudes of \cref{eq:maxcutsusy}. All indices are understood $\bmod \;n$.}
\label{fig:maxcutbox}
\end{figure}
Consider the box labeled by the cyclically ordered indices
$(i,j,k,l)$ and let $\mathcal C_{ijkl}$ denote its quadruple cut.
Writing the loop dual point as $y$ and setting $\ell_i=y-x_i$, its
cut equations are
\begin{equation}\label{eq:quadcutconditions}
\ell^2=0,\qquad
(\ell+x_{ij})^2=0,\qquad
(\ell+x_{ik})^2=0,\qquad
(\ell+x_{il})^2=0 .
\end{equation}
For generic external kinematics, these equations have two solutions,
denoted by $\ell_*^\sigma$ with $\sigma=\pm$.
On either solution, generalized unitarity factorizes the cut
integrand into the four tree-level superamplitudes located at the
corners of the box
\cite{Britto:2004nc,Drummond:2008bq}, as shown in \cref{fig:maxcutbox}. Define
\begin{equation}\label{eq:maxcutsusy}
\mathcal U^\sigma_{ijkl}
:=
\left.
\int\prod_{a=i,j,k,l}\mathrm d^4\eta_{\ell_a}\,
A_{ij}^{(0)}
A_{jk}^{(0)}
A_{kl}^{(0)}
A_{li}^{(0)}
\right|_{\ell=\ell_*^\sigma},
\end{equation}
where $\ell_i,\ell_j,\ell_k,\ell_l,$ are the four cut momenta.  Each tree
amplitude depends on the two adjacent cut legs and on the external
legs belonging to the corresponding corner.  We use
$\mathrm d^4\eta_{\ell_a}
:=\prod_{\alpha=1}^4\mathrm d\eta_{\ell_a}^{\alpha}$, where $\alpha$ is the R-symmetry index. 
The Grassmann integrations implement the sum over the complete
$\mathcal N=4$ supermultiplet propagating across each cut line.

Including the bosonic cut Jacobian and using the
normalizations of \cref{eq:boxfcns}, one obtains
\begin{equation} \label{eq:maxcutsusyL}
\mathrm{LeS}_{ijkl}^{\sigma}
\bigl[A_n^{(1)}\bigr]
=
\frac{2}{\sqrt{\rho_{ijkl}}}\,
\mathcal U^\sigma_{ijkl}.
\end{equation}
The branch of $\sqrt{\rho_{ijkl}}$ and the orientations of the two
residue contours are chosen consistently.  With the same
conventions, the factor $\sqrt{\rho_{ijkl}}$ in
\cref{eq:boxfcns} normalizes the four-dimensional integrand of the
box function to have unit residue on either branch,
\begin{equation}
\mathrm{LeS}_{ijkl}^{\sigma}[F_{ijkl}]=1.
\end{equation}
Here and below, the notation on the left denotes the residue of the
four-dimensional integrand associated with $F_{ijkl}$.
  The
scalar-box coefficient is determined by the contour enclosing both
quadruple-cut solutions, or equivalently by the branch-even
combination of their residues \cite{Cachazo2008}:
\begin{equation}\label{eq:boxcoef}
c_{ijkl}
=
\frac{
  \displaystyle\sum_{\sigma=\pm}
  \mathrm{LeS}_{ijkl}^{\sigma}[A_n^{(1)}]
}{
  \displaystyle\sum_{\sigma=\pm}
  \mathrm{LeS}_{ijkl}^{\sigma}[F_{ijkl}]
}
=
\frac12
\sum_{\sigma=\pm}
\mathrm{LeS}_{ijkl}^{\sigma}[A_n^{(1)}]
=
\frac{1}{\sqrt{\rho_{ijkl}}}
\sum_{\sigma=\pm}\mathcal U^\sigma_{ijkl}.
\end{equation}
Thus, the individual leading singularities retain branchwise
information about the four-dimensional integrand, while their
average determines the coefficient of the parity-even scalar box
function.

\paragraph{Landau Singularities.}
We will also need the Landau singularity associated with a cut.  Let
$\mathcal K$ denote the complexified space of external kinematics
and consider the incidence variety
\begin{equation}
\mathfrak C_E
=
\left\{
  (\kappa,\ell)\in\mathcal K\times\mathcal L:
  q_e(\ell;\kappa)^2=0
  \ \text{for all }e\in E
\right\},
\end{equation}
together with the projection
$\pi_E:\mathfrak C_E\rightarrow\mathcal K$.
The Landau
locus of the cut is the codimension-one locus in external kinematics where the fiber of this projection becomes non-generic,
for instance because it develops a singularity or
points collide. We denote
by $\mathrm{LaS}[\mathcal C_E]$ a reduced polynomial cutting out this locus,
understood up to multiplication by a nonzero constant and, for a
reducible locus, as the product of its irreducible hypersurface
components \cite{Landau1959,Pham1967,MizeraTelen2022,
BourjailyVerguvonHippel2023}.

If $\mathcal C_E$ is maximal, the map $\pi_E$ is generically finite
of degree $r_E$, and
$\mathrm{LaS}[\mathcal C_E]=0$ detects kinematics for which its
fiber ceases to consist of $r_E$ distinct reduced points.  At
a generic point of this discriminant, two solutions of the maximal
cut coincide.  For the one-loop box cut $\mathcal C_{ijkl}$, this
finite branch locus is $\rho_{ijkl}=0$ \cite{Britto:2004nc,Abreu:2017ptx}.  If $\mathcal C_E$ is
next-to-maximal and its generic fiber is a smooth curve, then
$\mathrm{LaS}[\mathcal C_E]=0$ detects kinematics for which this
curve becomes singular
\cite{HolleringMazzucchelliParisiSturmfels:2026Landau}.
\subsection{Symbols of box functions}
\label{sec:symbols-boxes}
One-loop amplitudes can be expressed in terms of multiple polylogarithms (MPLs), a class of iterated integrals graded by the transcendental weight $w$. They contain the logarithm $\log z$ at weight 1 and the classical polylogarithms $\operatorname{Li}_n z$ at weight $n$. The analytic structure of these functions is encoded in the \emph{symbol} \cite{Chen:1977oja,Goncharov:2005sla,Brown:2009qja,Goncharov:2010jf,Duhr:2011zq,Duhr:2012fh}, defined recursively in the weight as follows. A general transcendental function $F_w$ of weight $w$ has a total differential that can be written as a finite sum
\begin{equation}
 \mathrm{d} F_w=\sum_{i}F^i_{w-1}\,\mathrm{d}\log R_i,
\end{equation}
where $R_i$ are algebraic functions of the kinematic variables and $F^i_{w-1}$ are transcendental functions of weight $w-1$. The symbol is then defined as
\begin{equation}
\mathcal{S}[F_w]=\sum_{i}\mathcal{S}[F^i_{w-1}]\otimes R_i,
\end{equation}
with the recursion starting at weight 1. The symbol annihilates transcendental constants. Iterating, the symbol of a weight $w$ function is a linear combination of tensor products made up of $w$ entries, which are called \emph{letters}.

If a function $f$  can be expressed in terms of MPLs, then its symbol can be written as
\begin{equation}
 \mathcal{S}[f]=\sum_{i} c_i w_i, \quad w_i=b_i^{(1)} \otimes \ldots \otimes b^{(w)}_i,   
\end{equation}
where $b_i^{(j)}$ are letters of the symbol $\mathcal{S}$ of $f$, and the $c_i$ are algebraic functions. The set of possible letters is called the \emph{alphabet}. We also call $c_i$ the \emph{coefficient} of the word $w_i$ and denote it by ${\rm Coeff}[w_i]$. 
The functions we encounter at one-loop have at most weight two. The relevant symbols are, at weight one,
\begin{equation}
    \mathcal{S}\left[\log (a)\right]=\otimes a
\end{equation}
and at weight two,
\begin{align}
\mathcal{S}\left[\log(a)\cdot\log(b)\right]&=a \otimes b + b\otimes a,\\
\mathcal{S}\left[\dlog_2(a)\right]&=-(1-a)\otimes a.
\end{align}

Before considering the symbol of the amplitude, let us first look at the symbol of the box functions of \cref{eq:boxfcns}, whose expression in terms of polylogarithms was given in \cite{Bern:1994zx}. The boxes can be classified according to the number of nonzero external masses, with the appropriate indexing, as shown in \cref{fig:boxtypes}.  
The symbol for each type of box, including divergent terms and denoting weight two terms by the subscript 2, is given below. 
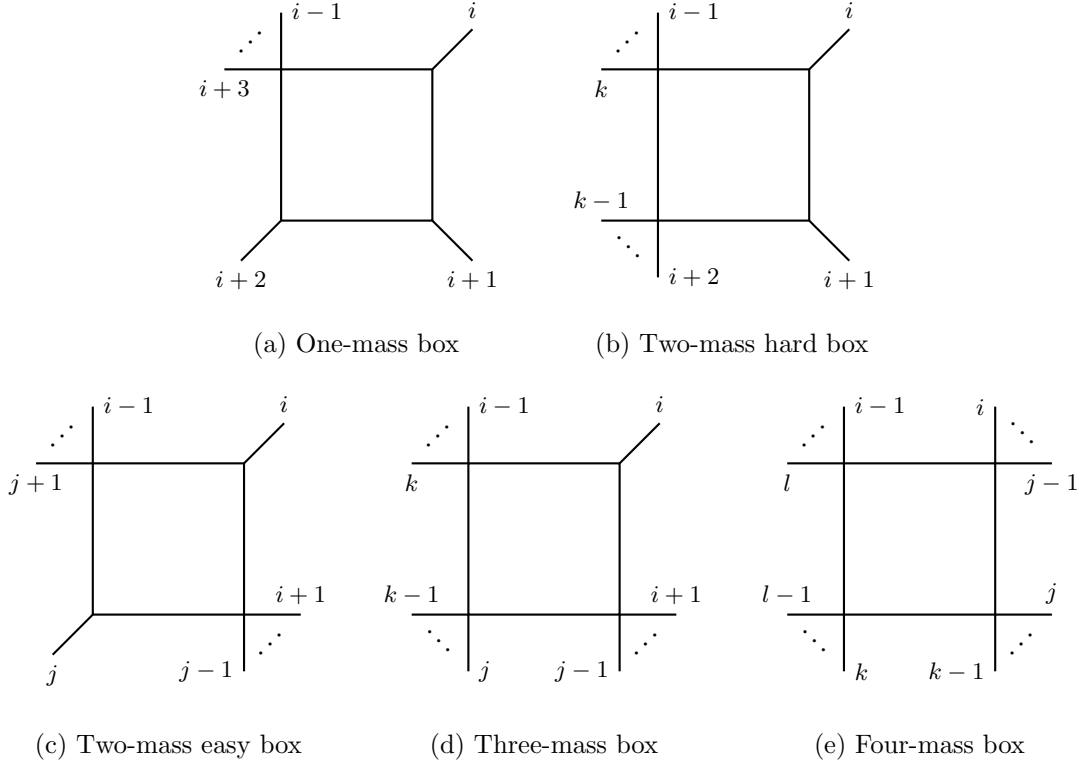
\begin{figure}[ht]
\centering
\begin{subfigure}[h]{\textwidth}
\centering
\begin{subfigure}[t]{0.32\textwidth}
\centering
\begin{tikzpicture}
	\coordinate (1) at (-1,-1); 
    \coordinate (2) at (-1, 1); 
    \coordinate (3) at (1,1); 
    \coordinate (4) at (1,-1); 
    \path[draw,thick] (1) -- (4) -- (3) -- (2) -- cycle;
	\draw[thick](1) -- ++(-135:0.75) node[below]{\footnotesize $i+2$};
    \draw[thick](2)--(-1.75,1) node(a){} node[below]{\footnotesize $i+3$} ;
    \draw[thick](2)--(-1,1.75) node(b){} node[right]{\footnotesize $i-1$} ;
    \path (a) -- node[sloped]{$\cdots$} node[near start,above,xshift=-2pt, yshift=2pt]{} (b);
    \draw[thick](3) -- ++(45:0.75) node[above]{\footnotesize $i$};
    \draw[thick](4) -- ++(-45:0.75) node[below]{\footnotesize $i+1$};
    \path[use as bounding box] (-2.2,-2.2) rectangle (2.2,2.2);
\end{tikzpicture}
\caption{One-mass box}
\end{subfigure}
\begin{subfigure}[t]{0.32\textwidth}
\centering
\begin{tikzpicture}
	\coordinate (1) at (-1,-1); 
    \coordinate (2) at (-1, 1); 
    \coordinate (3) at (1,1); 
    \coordinate (4) at (1,-1); 
    \path[draw,thick] (1) -- (4) -- (3) -- (2) -- cycle;
    
    \draw[thick](1)--(-1.75,-1) node(g){} node[above]{\footnotesize $k-1$};
    \draw[thick](1)--(-1,-1.75) node(h){} node[right]{\footnotesize $i+2$};
    \path (g) -- node[sloped]{$\cdots$} node[left,xshift=-2pt, yshift=-4pt]{} (h);
    
    \draw[thick](3) -- ++(45:0.75) node[above]{\footnotesize $i$};
    
    \draw[thick](2)--(-1.75,1) node(a){} node[below]{\footnotesize $k$};
    \draw[thick](2)--(-1,1.75) node(b){} node[right]{\footnotesize $i-1$};
    \path (a) -- node[sloped]{$\cdots$} node[near start,above,xshift=-2pt, yshift=2pt]{} (b);
   
    \draw[thick](4) -- ++(-45:0.75) node[below]{\footnotesize $i+1$};
    \path[use as bounding box] (-2.2,-2.2) rectangle (2.2,2.2);
\end{tikzpicture}
\caption{Two-mass hard box}
\end{subfigure}
\end{subfigure}
\begin{subfigure}[h]{\textwidth}
\centering
\begin{subfigure}[t]{0.32\textwidth}
\centering
\begin{tikzpicture}
	\coordinate (1) at (-1,-1); 
    \coordinate (2) at (-1, 1); 
    \coordinate (3) at (1,1); 
    \coordinate (4) at (1,-1); 
    \path[draw,thick] (1) -- (4) -- (3) -- (2) -- cycle;
	\draw[thick](1) -- ++(-135:0.75) node[below]{\footnotesize $j$} ;
    \draw[thick](2)--(-1.75,1) node(a){} node[below]{\footnotesize $j+1$};
    \draw[thick](2)--(-1,1.75) node(b){} node[right]{\footnotesize $i-1$} ;
    \path (a) -- node[sloped]{$\cdots$} node[near start,above,xshift=-2pt, yshift=2pt]{} (b);
    \draw[thick](3) -- ++(45:0.75) node[above]{\footnotesize $i$};
    \draw[thick](4)--(1.75,-1) node(e){} node[above]{\footnotesize $i+1$};
    \draw[thick](4)--(1,-1.75) node(f){} node[left]{\footnotesize $j-1$};
    \path (e) -- node[sloped]{$\cdots$} node[right,xshift=-2pt,yshift=-4pt]{}(f); 
    \path[use as bounding box] (-2.3,-2.3) rectangle (2.3,2.3);
\end{tikzpicture}
\caption{Two-mass easy box}
\end{subfigure}
\begin{subfigure}[t]{0.32\textwidth}
\centering
\begin{tikzpicture}
	\coordinate (1) at (-1,-1); 
    \coordinate (2) at (-1, 1); 
    \coordinate (3) at (1,1); 
    \coordinate (4) at (1,-1); 
    \path[draw,thick] (1) -- (4) -- (3) -- (2) -- cycle;
    \draw[thick](3) -- ++(45:0.75) node[above]{\footnotesize $i$};
    
    \draw[thick](1)--(-1.75,-1) node(g){} node[above]{\footnotesize $k-1$};
    \draw[thick](1)--(-1,-1.75) node(h){} node[right]{\footnotesize $j$};
    \path (g) -- node[sloped]{$\cdots$} node[left,xshift=-2pt, yshift=-4pt]{} (h);
    
    \draw[thick](2)--(-1.75,1) node(a){} node[below]{\footnotesize $k$} ;
    \draw[thick](2)--(-1,1.75) node(b){} node[right]{\footnotesize $i-1$};
    \path (a) -- node[sloped]{$\cdots$} node[near start,above,xshift=-2pt, yshift=2pt]{} (b);
   
    \draw[thick](4)--(1.75,-1) node(e){} node[above]{\footnotesize $i+1$};
    \draw[thick](4)--(1,-1.75) node(f){} node[left]{\footnotesize $j-1$};
    \path (e) -- node[sloped]{$\cdots$} node[right,xshift=-2pt,yshift=-4pt]{}(f);
    \path[use as bounding box] (-2.3,-2.3) rectangle (2.3,2.3);
\end{tikzpicture}
\caption{Three-mass box}
\end{subfigure}
\begin{subfigure}[t]{0.32\textwidth}
\centering
\begin{tikzpicture}
	\coordinate (1) at (-1,-1); 
    \coordinate (2) at (-1, 1); 
    \coordinate (3) at (1,1); 
    \coordinate (4) at (1,-1); 
    \path[draw,thick] (1) -- (4) -- (3) -- (2) -- cycle;
	\draw[thick](1)--(-1.75,-1) node(g){}  node[above]{\footnotesize $l-1$};
    \draw[thick](1)--(-1,-1.75) node(h){}  node[right]{\footnotesize $k$};
    \path (g) -- node[sloped]{$\cdots$} node[left,xshift=-2pt, yshift=-4pt]{} (h);
    \draw[thick](2)--(-1.75,1) node(a){}  node[below]{\footnotesize $l$};
    \draw[thick](2)--(-1,1.75) node(b){}  node[right]{\footnotesize $i-1$};
    \path (a) -- node[sloped]{$\cdots$} node[near start,above,xshift=-2pt, yshift=2pt]{} (b);
    \draw[thick](3)--(1,1.75) node(c){} node[left]{\footnotesize $i$};
    \draw[thick](3)--(1.75,1) node(d){} node[below]{\footnotesize $j-1$};
    \path (c) -- node[sloped]{$\cdots$} node[above,near end,xshift=3pt,yshift=2pt]{}(d);
    \draw[thick](4)--(1.75,-1) node(e){} node[above]{\footnotesize $j$} ;
    \draw[thick](4)--(1,-1.75) node(f){}  node[left]{\footnotesize $k-1$};
    \path (e) -- node[sloped]{$\cdots$} node[right,xshift=-2pt,yshift=-4pt]{}(f);
    \path[use as bounding box] (-2.3,-2.3) rectangle (2.3,2.3);
\end{tikzpicture}
\caption{Four-mass box}
\end{subfigure}
\end{subfigure}
\caption{Different types of scalar boxes appearing at one-loop, classified by the placement of massive (off-shell) legs. 
}
\label{fig:boxtypes}
\end{figure}
\begin{align}
\begin{split}
   \cS[F^{\rm 1m}_{i,i+1,i+2,i+3}] = -\frac{1}{\epsilon
   ^2}+ \frac{\otimes x^2_{i,{i+2}}-\otimes
   x^2_{i,{i+3}}+\otimes
   x^2_{{i+1},{i+3}}}{\epsilon}
   +\cS_2[F^{1m}_{i,i+1,i+2,i+3}] + \ord(\epsilon) 
\end{split}\\[1ex]
\begin{split}
    \cS[F^{\rm 2me}_{i,i+1,j,j+1}] = \frac{\otimes x^2_{ij}-\otimes x^2_{i,j+1}-\otimes x^2_{{i+1},j}+\otimes   x^2_{{i+1},{j+1}}}{\epsilon } + \cS_2[F^{\rm 2me}_{i,i+1,j,j+1}]  + \ord(\epsilon)
\end{split}\\[1ex]
\begin{split}
   \cS[F^{\rm 2mh}_{i,i+1,i+2,k}] 
= -\frac{1}{2 \epsilon ^2}+\frac{\otimes
   x^2_{i,{i+2}}-\otimes x^2_{ik}+2
   \left(\otimes x^2_{{i+1},k}\right)-\otimes
   x^2_{{i+2},k}}{2 \epsilon }+ \cS_2[F^{\rm 2mh}_{i,i+1,i+2,k}]  + \ord(\epsilon)  
\end{split}\\[1ex]
\begin{split}
    \cS[F^{\rm 3m}_{i,i+1,j,k}] = 
\frac{\otimes x^2_{ij}-\otimes x^2_{ik}-\otimes
   x^2_{{i+1},j}+\otimes x^2_{{i+1},k}}{2
   \epsilon } + \cS_2[F^{\rm 3m}_{i,i+1,j,k}]  + \ord(\epsilon) 
\end{split}\\[1ex]
\begin{split}
    \cS[F^{\rm 4m}_{ijkl}] = 
 \cS_2[F^{\rm 4m}_{ijkl}]  + \ord(\epsilon)
\end{split}
\end{align}
Focusing on weight two terms,
\begin{align}
\begin{split}\label{eq:1mbox-symbol}
    \cS_2[F^{\rm 1m}_{i,i+1,i+2,i+3}] &=
\frac{x^2_{i,{i+3}}}{x^2_{i,{i+2}}x^2_{{i+1},{i+3}}}\otimes x^2_{i,{i+2}}x^2_{{i+1},{i+3}}+x^2_{i,{i+3}}\otimes  x^2_{i,{i+3}}\\
&
+\frac{x^2_{{i+1},{i+3}}}{x^2_{i,{i+3}}} \otimes \left(x^2_{i,{i+3}}-x^2_{{i+1},{i+3}}\right) 
+ \frac{x^2_{i,{i+2}}}{x^2_{i,{i+3}}} \otimes \left(x^2_{i,{i+3}}-x^2_{i,{i+2}}\right)    
\end{split}\\[1ex]
\begin{split}\label{eq:2mhbox-symbol}
\cS_2[F^{\rm 2mh}_{i,i+1,i+2,k}] 
&= \frac{x^2_{{i+2},k}x^2_{ik}}{x^2_{i,{i+2}} x^2_{{i+1},k}}\otimes x^2_{i,{i+2}} x^2_{{i+1},k}\\
&
+\frac{x^2_{{i+1},k}}{x^2_{ik}}\otimes   \left(x^2_{ik}-x^2_{{i+1},k}\right)
   +  \frac{x^2_{{i+1},k}}{x^2_{{i+2},k}}\otimes \left(x^2_{{i+2},k} -x^2_{{i+1},k}\right ) \\
   &
   +\frac{1}{2} \frac{x^2_{i,{i+2}}x^2_{ik}}{x^2_{{i+2},k}}\otimes  x^2_{ik}
   +\frac{1}{2} \frac{x^2_{i,{i+2}}x^2_{{i+2},k}}{x^2_{ik}}\otimes  x^2_{{i+2},k} \\
   &
   -\frac{1}{2} \frac{x^2_{{i+2},k}x^2_{ik}}{x^2_{i,{i+2}}}\otimes  x^2_{i,{i+2}}
   -x^2_{{i+1},k}\otimes x^2_{{i+1},k}
\end{split}
\end{align}
\begin{align}
\begin{split}\label{eq:2mebox-symbol}
\cS_2[F^{\rm 2me}_{i,i+1,j,j+1}]  &= 
\frac{x^2_{{i+1},j}x^2_{i,j+1}}{x^2_{ij} x^2_{{i+1},{j+1}}} \otimes
   \left(x^2_{i,j+1} x^2_{{i+1},j}-x^2_{ij} x^2_{{i+1},{j+1}}\right) \\
   &
   +\frac{x^2_{ij}}{x^2_{{i+1},j}}\otimes \left(x^2_{{i+1},j}-x^2_{ij}\right)   
   +\frac{x^2_{{i+1},{j+1}}}{x^2_{i,j+1}}\otimes \left(x^2_{i,j+1}-x^2_{{i+1},{j+1}}\right) \\
   &
   +\frac{x^2_{ij}}{x^2_{i,j+1}}\otimes \left(x^2_{i,j+1}-x^2_{ij}\right)
   +\frac{x^2_{{i+1},{j+1}}}{x^2_{{i+1},j}}\otimes  \left(x^2_{{i+1},j}-x^2_{{i+1},{j+1}}\right) \\
   &
   -x^2_{ij}\otimes x^2_{ij}
   +x^2_{i,j+1}\otimes x^2_{i,j+1}\\&
   +x^2_{{i+1},j}\otimes x^2_{{i+1},j}
   -x^2_{{i+1},{j+1}}\otimes x^2_{{i+1},{j+1}} 
\end{split}\\[1ex]
\begin{split}\label{eq:3mbox-symbol}
\cS_2[F^{\rm 3m}_{i,i+1,j,k}] 
&= 
 \frac{ x^2_{ik} x^2_{{i+1},j}}{x^2_{ij} x^2_{{i+1},k}}\otimes \left(x^2_{ik} x^2_{{i+1},j}-x^2_{ij} x^2_{{i+1},k}\right)\\
    &
   + \frac{x^2_{ij}}{x^2_{{i+1},j}}\otimes \left(x^2_{{i+1},j}-x^2_{ij}\right)   
   + \frac{x^2_{{i+1},k}}{ x^2_{ik}}\otimes \left(x^2_{ik}-x^2_{{i+1},k}\right) \\
   &
   +\frac{1}{2} \frac{x^2_{ij}x^2_{{i+1},k}}{ x^2_{{i+1},j}x^2_{ik}} \otimes x^2_{jk}
   +\frac{1}{2} \frac{x^2_{ij}x^2_{ik}}{x^2_{jk}} \otimes \frac{ x^2_{ik}}{x^2_{ij}}
   +\frac{1}{2} \frac{x^2_{{i+1},k}x^2_{{i+1},j}}{x^2_{jk}} \otimes \frac{x^2_{{i+1},j}}{x^2_{{i+1},k}}
\end{split}\\[1ex]
\begin{split}\label{eq:4mbox-symbol}
\cS_2[F^{\rm 4m}_{ijkl}] 
&= \frac{x^2_{ij} x^2_{kl}}{x^2_{ik} x^2_{jl}} \otimes
\frac{x^2_{ik} x^2_{jl}+x^2_{jk} x^2_{il}-x^2_{ij} x^2_{kl} - \sqrt{\rho_{ijkl}}}{x^2_{ik} x^2_{jl}+x^2_{jk} x^2_{il}-x^2_{ij} x^2_{kl} + \sqrt{\rho_{ijkl}}} \\
& + \frac{x^2_{jk} x^2_{il}}{x^2_{ik} x^2_{jl}} \otimes 
\frac{x^2_{ik} x^2_{jl}- x^2_{jk} x^2_{il} + x^2_{ij} x^2_{kl} -\sqrt{\rho_{ijkl}}}{x^2_{ik} x^2_{jl} - x^2_{jk} x^2_{il} + x^2_{ij} x^2_{kl} + \sqrt{\rho_{ijkl}}}
\end{split}
\end{align}
\subsection{Infrared behaviour and dual conformal symmetry}
From what we have discussed so far, it is straightforward to obtain the symbol of the one-loop amplitude using the expression for the coefficients given in \cref{eq:boxcoef} and the symbols of the box functions. However, the result can be simplified by taking into account that the box coefficients are not independent but satisfy linear relations. For example, the IR behaviour of $\mathcal{N} = 4 \text{ SYM}$ imposes specific constraints. One important simplification is the cancellation of certain symbol letters. In the symbols of the boxes, we find four types of letters, of the form $x_{ij}^2$, $x_{ij}^2-x_{kl}^2$, $x_{ij}^2 x_{kl}^2 - x_{il}^2 x_{jk}^2$, as well as the letters containing square roots associated to four-mass boxes. However, all terms with letters of the form $x_{ij}^2-x_{kl}^2$ cancel in the full amplitude. This property is a consequence of a certain class of relations that we refer to as \textit{two-mass triangle relations}.

The two-mass triangle relations can be understood as a consequence of the form of the one-loop dual conformal anomaly, as shown in \cite{Brandhuber:2009rel,Brandhuber:2009proof}. Dual conformal symmetry is exact at tree level, as conjectured in \cite{Drummond:2008vq} and proven in \cite{Brandhuber:2008pf}. At loop level the symmetry is broken in a well-defined way due to the presence of infrared divergences. These are treated in the framework of dimensional regularization, where they appear as poles in $\epsilon$. On general grounds, infrared divergences of $\mathcal{N}=4\text{ SYM}$ amplitudes are captured by the corresponding tree amplitude and have a known helicity independent universal form,  \cite{Giele:1991vf,Catani:1998bh,Kunszt:1994np} 
\begin{align}\label{eq:generalir}
A_n^{(1)\mathrm{IR}}&=-\frac{1}{\epsilon^2}\sum_{i=1}^n\left(-x^2_{i,i+2}\right)^{-\epsilon}A^{(0)}_n\\&=\sum_{i=1}^n\left[-\frac{1}{\epsilon^2}+\frac{1}{\epsilon} \log (-x^2_{i,i+2})-\frac{1}{2} \log^2 (-x^2_{i,i+2})+\mathcal{O}(\epsilon)\right]A^{(0)}_n,
\label{eq:N4div}
\end{align}
where in the second equality we have expanded the series to highlight that the IR divergences will contribute to the finite part of the amplitude through the term $A_n^{(1)\mathrm{IR}}\big|_{\epsilon^0}=\sum_i-\frac{1}{2} \log^2 (-x^2_{i,i+2})A^{(0)}_n$. This gives us a complete understanding of the terms that break dual conformal invariance and their coefficients, which is relevant for the analysis of \cref{sec:clust}. One also expects the letters $x_{ij}^2-x_{kl}^2$ to cancel and would not include them in a possible alphabet used for bootstrapping the amplitude, as they are not dual conformal invariant and cannot be produced from the divergent part.

There are various ways to isolate a part of the amplitude that is IR finite and dual conformal invariant, for example using ratio functions \cite{Drummond:2008vq}, since the MHV amplitude captures the anomaly completely. The ratio function can then be expressed in terms of dual conformal cross-ratios. Here, we do not use any normalization but work directly with dimensional regularization, meaning that our results include the IR divergent part. The consequence is that the full amplitude cannot be written in cross-ratios but the part which spoils this property has the specific structure of \cref{eq:generalir}. In the literature, one also finds BDS or BDS-like normalized amplitudes. These differ only in terms involving two-particle invariants compared to the unnormalised amplitudes, since they are the only invariants involved in the IR divergent terms of the scalar integrals. This is important when comparing with results in \cite{Gurdogan:2020tip}, since it means that the coefficients of two-particle invariants can be different.

\section{Leading-singularity decomposition of the one-loop symbol}
\label{sec:res}
Having extracted the divergences in \cref{eq:N4div}, we focus in this section on the finite part of the one-loop amplitude, which is the part of transcendental weight two. We show that the symbol of the one-loop amplitude at weight two\footnote{From now on, we drop the subscript 2 from $\mathcal{S}_2$.} can be written as 
\begin{equation}\label{eq:symbolresult}
\mathcal{S}[A^{(1)}_{n,k}]=\mathcal{S}^{\mathrm{LS}}[A^{(1)}_{n,k}]+\mathcal{S}^{\mathrm{4m}}+\sum_{i,j,l} d_{ijl}\,x^2_{ij} \otimes x^2_{il},
\end{equation}
where the first term is the \emph{leading-singularity part} of the symbol introduced in \cref{subsec:LS-part}. It is the part directly controlled by maximal-cut leading
singularities. The second term is the contribution from all possible 4-mass boxes,
\begin{equation}
    \mathcal{S}^{\mathrm{4m}} = \sum_{i,j,k,l} c^{\mathrm{4m}}_{ijkl}\mathcal{S}[F^{\rm 4m}_{ijkl}],
\end{equation}
where the sum range is: $1\leq i<j-1 < k-2 < l-3 < \min(n-2,n+i-4)$.
Its contribution is qualitatively different, since a generic
four-mass box has maximal-cut solutions which are algebraic and the
corresponding discriminant enters the symbol through square-root letters. It gives the simplest example in our analysis in which a
Landau discriminant associated with a non-rational cut fiber naturally
leads to algebraic symbol letters and requires an extension of ordinary
cluster adjacency.

Finally, regarding the third term of \cref{eq:symbolresult}, the terms which contribute come from cases of words $x^2_{ij} \otimes x^2_{il}$ with exactly three distinct indices, and the case of a repeated letter $x^2_{i,i+2} \otimes x^2_{i,i+2}$, which is a two-particle invariant. In the first case, we get 
\begin{equation}\label{eq:dijl}
    d_{ijl}=\begin{cases}
         \frac{1}{2}(C^+_{ij,l}- C^+_{jl,i}-C^-_{l i,j}),  \, &\text{if $l=i-2$}.\\
          \frac{1}{2}(C^+_{ij,l}- C^+_{jl,i}+C^-_{l i,j}), \, &\text{otherwise,}
          \end{cases}
\end{equation}
where 
\begin{equation} \label{def:Cpm}
    C^\pm_{ij,l}:= \begin{cases}
c^{\mathrm{3m}}_{i,i+1,j,l}\pm c^{\mathrm{3m}}_{i,j-1,j,l}, \, &\text{if $|i-j|\geq 3$,}
        \\
         c^{\mathrm{2mh}}_{i,i+1,j,l},  \, &\text{if $j=i+2$}.
        \end{cases}
\end{equation}
In the case of a repeated letter, we have
\begin{equation}\label{eq:dijj}
    d_{ijj}=\begin{cases}
        \frac{1}{2} \left(c^{\mathrm{1m}}_{i,i+1,i+2,i+3}+ c^{\mathrm{1m}}_{i-1,i,i+1,i+2}+\sum_{k=i+4}^{i-2}c^{\mathrm{2mh}}_{i,i+1,i+2,k} \right),\, &\text{if $|i-j|=2$} \\
        0,\, &\text{otherwise}.
    \end{cases}
\end{equation}
The above terms demonstrate
that the integrated amplitude is not exhausted by the
leading-singularity picture. We determine the coefficient of every
rational symbol word in  \cref{sec:wordscoef} and distinguish three phenomena. Some words cancel
completely in the full amplitude; some have coefficients already
captured by their LS contribution; and others retain a genuine residual,
non-LS coefficient. These residual terms provide a concrete measure of
the gap between the structure inferred directly from maximal cuts and
the symbol of the integrated, dimensionally regulated amplitude.
Importantly, they do not all have the same origin, and only part of the
residual sector can be simplified using relations among box
coefficients.

\subsection{The LS part of a symbol}
\label{subsec:LS-part}

We now define the leading-singularity (LS) part of a symbol.  Since rational
prefactors may satisfy linear relations, we first fix a representation
of the symbol
\begin{equation}
  \mathcal S[A]
  =
  \sum_\alpha
  f_\alpha \,
  a_\alpha^{(1)}\otimes \cdots \otimes a_\alpha^{(2L)}
\end{equation} 
where $\alpha$ indexes the terms of the symbol.
Here $A$ is an amplitude of uniform weight $2L$.
The following definition is therefore relative to this chosen
representation. 
\begin{definition} \label{def:LS_part}
A term
$
  f\, a^{(1)}\otimes\cdots\otimes a^{(2L)}$
belongs to the \emph{leading-singularity (LS) part}
$\mathcal S^{\mathrm{LS}}[A]$ of $\mathcal S[A]$ if there exists a
nested sequence of cuts
$
  \mathcal C^{[1]}\supseteq
  \mathcal C^{[2]}\supseteq
  \cdots\supseteq
  \mathcal C^{[2L]}$
such that:
\begin{enumerate}
  \item The final cut $\mathcal C^{[2L]}$ is maximal.  Thus, for generic
  external kinematics,
  \begin{equation}
    \mathcal C^{[2L]}
    =
    \left\{\ell_*^{(1)},\ldots,\ell_*^{(\gamma)}\right\}.
  \end{equation}

  \item There is a point $\ell_*^{{\sigma}}\in \mathcal C^{[2L]}$ and a
  nonzero constant $\lambda\in\mathbb C^\times$ such that
  \begin{equation}
    f
    =
    \lambda\,
    \mathrm{LeS}^{{\sigma}}_{\mathcal C^{[2L]}}
    \bigl[A\bigr].
  \end{equation}

  \item  The symbol letter $a^{(j)}$ is an
  irreducible factor of the Landau singularity $\mathrm{LaS}[\mathcal C^{[j]}]$, for each $j=1,\ldots,2L$.
\end{enumerate}
\end{definition}

The superscript $[j]$ on $\mathcal C^{[j]}$ labels the position in the
symbol word.

Given a word
$w=a^{(1)}\otimes\cdots\otimes a^{(2L)}$,
we define its \emph{leading-singularity (LS) coefficient} $ {\rm Coeff}^{\rm LS} [w]$ to be 
\begin{equation}
    {\rm Coeff}^{\rm LS} [w] = \sum_{\alpha} f_{\alpha},
\end{equation}
where the sum is over all terms $f_\alpha$ such that $f_\alpha w$ belongs to the LS part as in \cref{def:LS_part}.

\begin{proposition}
The  \emph{LS part} of the one-loop amplitude in planar $\mathcal{N}=4$ SYM using the box expansion of \cref{eq:boxdecomp} is:
\begin{equation}
\label{eq:ll}
\begin{aligned}
   \mathcal{S}^{\mathrm{LS}}[A^{(1)}_n] &= \sum_{i=1}^n c^{\mathrm{1m}}_{i,i+1,i+2,i+3} \frac{x^2_{i,i+3}}{x^2_{i,i+2} x^2_{i+1,i+3}} \otimes x^2_{i,i+2} x^2_{i+1,i+3} \\
    & + \sum_{i=1}^{n} \sum_{j=i+3}^{t(i,n)} c^{\mathrm{2me}}_{i,i+1,j,j+1} \frac{x^2_{i,j+1} x^2_{i+1,j}}{x^2_{ij} x^2_{i+1,j+1}} \otimes \left(x^2_{ij} x^2_{i+1,j+1}-x^2_{i,j+1} x^2_{i+1,j} \right) \\
    & + \sum_{i=1}^n\sum_{k=i+4}^{i-2} c^{\mathrm{2mh}}_{i,i+1,i+2,k} \frac{x^2_{ik} x^2_{i+2,k}}{x^2_{i+1,k} x^2_{i,i+2}} \otimes x^2_{i,i+2} x^2_{i+1,k}  \\
    & + \sum_{i=1}^n \sum_{j=i+3}^{i-4} \sum_{k=j+2}^{i-2}c^{\mathrm{3m}}_{i,i+1,j,k} \frac{x^2_{i+1,j} x^2_{ik}}{x^2_{i+1,k} x^2_{ij}} \otimes \left(x^2_{ij}x^2_{i+1,k}-x^2_{ik} x^2_{i+1,j}  \right), 
\end{aligned}
\end{equation}
where all indices are understood $\bmod \;n$ and for the two-mass easy box, the upper limit of the inner sum is $t(i,n)=\begin{cases}i+\frac{n}{2}, &\text{if }n\text{ is even and } 1\leq i \leq \frac{n}{2} \\i+\frac{n}{2}-1,&\text{if }n\text{ is even and } \frac{n}{2}< i \leq n\\i+\frac{n-1}{2},&\text{if }n\text{ is odd} \end{cases}$.
\end{proposition}
For one-loop amplitudes, the Landau singularity for a maximal cut is at a quadruple cut and is given by $\rho_{ijkl}$ of \cref{eq:jacobrho}. 
For any box with at least one massless leg, $\rho_{ijkl}$ is a perfect square, and we make conventional and consistent branch choices of the square root appearing in Eqs.~(\ref{eq:boxfcns}),(\ref{eq:maxcutsusyL}) such that the term with the $s$ and $t$ invariants comes without a minus sign, as follows for each such type of box: \begin{align}\label{eq:sqrt-rho}\begin{split}
     \sqrt{\rho^{\rm 1m}_{i,i+1,i+2,i+3}} &=  x^2_{i,i+2} x^2_{i+1,i+3}, \\
     \sqrt{\rho^{\rm 2me}_{i,i+1,j,j+1}} &=  x^2_{i j} x^2_{i+1,j+1}- x^2_{i,j+1} x^2_{i+1,j}, \\
     \sqrt{\rho^{\rm 2mh}_{i,i+1,i+2,k}} &=  x^2_{i,i+2} x^2_{i+1,k}, \\
     \sqrt{\rho^{\rm 3m}_{i,i+1,j,k}} &=  x^2_{i j} x^2_{i+1,k}- x^2_{ik} x^2_{i+1,j}, \\
     \end{split}
\end{align} 
These expressions can be recognised
as the second-entry letters of the symbols of boxes with massless legs, for which we take the corresponding first entries and note that their letters are the Landau singularities of their contractions to bubbles, shown in \cref{fig:box-bubbles}.
\begin{figure}[htbp]
    \centering
    \includegraphics[width=1.0\textwidth]{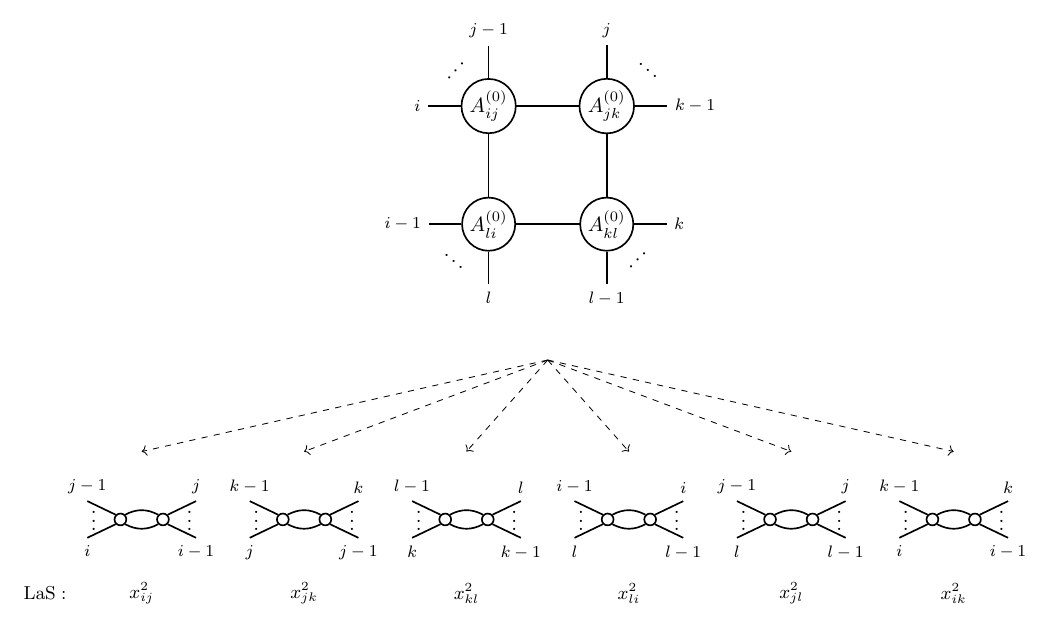}
    \caption{General box and its contractions to all possible different bubbles. The last line shows the Landau singularity of each bubble.}
    \label{fig:box-bubbles}
\end{figure}
Each box coefficient above is the average of two leading singularities associated to the same maximal cut of \cref{eq:boxcoef}. 
The four-mass box is a special case: as seen in \cref{eq:4mbox-symbol}, $\rho_{ijkl}$ does not appear itself as a letter of the final entry, but only within a more complicated algebraic letter. Therefore there is no LS part coming from four-mass boxes in \cref{eq:ll}. 
\subsection{Coefficients of words}
\label{sec:wordscoef}
In this section, we summarise explicitly the coefficients and the LS coefficients of each word for one-loop planar $\mathcal{N}=4$ SYM. 
\begin{proposition}\label{prop:cancel}
    Words of the form $x^2_{ij} \otimes (x^2_{kl}-x^2_{mn})$ vanish in the amplitude, meaning $\mathrm{Coeff}[x^2_{ij} \otimes (x^2_{kl}-x^2_{mn})]=0$.
\end{proposition}
A proof is given in \cref{sec:simplifications}.

\subsubsection*{Words of the form $x^2_{ij} \otimes x^2_{kl}$}
We now focus on words of the form $x^2_{ij} \otimes x^2_{kl}$.
We distinguish the different cases according to the number of distinct indices among $i,j,k,l$. 

\paragraph{Four indices distinct.}
\begin{proposition}\label{prop:4distinct}
    In the case where all indices $i,j,k,l$ are distinct, we have:
    \begin{enumerate}
        \item $\mathrm{Coeff}[x^2_{ij} \otimes x^2_{kl}]=0$, 
        \begin{itemize}
            \item if $i,j,k,l$ are cyclically ordered, or
            \item if $i,j,k,l$ are not cyclically ordered and\footnote{This corresponds to the Steinmann relations \cite{Steinmann:1960a,Steinmann:1960b}.} $|i-j|> 2$, $|k-l| >2$.
        \end{itemize}
        \item $\mathrm{Coeff}[x^2_{ij} \otimes x^2_{kl}]=\mathrm{Coeff}^{\mathrm{LS}}[x^2_{ij} \otimes x^2_{kl}]=-c_{\pi(ijkl)}$,\\ if $i,j,k,l$ are not cyclically ordered and $|i-j|= 2$ or $|k-l| =2$, where $\pi$ is the permutation of $i,j,k,l$ that restores the cyclic order. The coefficient $c_{\pi(ijkl)}$ belongs to a box of type $2mh$ or $1m$.
    \end{enumerate}   
\end{proposition}
In the first case, when the indices are cyclically ordered, the vanishing of the coefficient follows by inspection of the possible words appearing in  \cref{eq:1mbox-symbol,eq:2mhbox-symbol,eq:2mebox-symbol,eq:3mbox-symbol,eq:4mbox-symbol}. 
\paragraph{Three indices distinct.}
The repeated index, denoted by $i$, must appear once in each letter, so without loss of generality we consider a word of the form $x^2_{ij} \otimes x^2_{il}$. We denote by $d_{ijl}$ the difference between the total coefficient of this word and the LS part, 
\begin{equation}
    d_{ijl}:=\mathrm{Coeff}[x^2_{ij} \otimes x^2_{il}]-\mathrm{Coeff}^{\mathrm{LS}}[x^2_{ij} \otimes x^2_{il}].
\end{equation}
Moreover, we will assume the cyclic ordering $i<j<l$. The other ordering $i<l<j$ can be obtained by reflection.
\begin{proposition} \label{prop:3distinct}
In the case where  $i<j<l$ are cyclically ordered:
\begin{enumerate} 
    \item If any two of $i,j,l$ differ by 1, then 
    \begin{equation}\mathrm{Coeff}[x^2_{ij} \otimes x^2_{il}]=\mathrm{Coeff}^{\mathrm{LS}}[x^2_{ij} \otimes x^2_{il}],
    \end{equation}
    and specifically, the only nonvanishing instances are of the form
    \begin{align}
        \mathrm{Coeff}[x^2_{i,i+3} \otimes x^2_{i,i+2}] &= c^{\rm 1m}_{i,i+1,i+2,i+3} \\
        \mathrm{Coeff}[x^2_{i,i-3} \otimes x^2_{i,i-2}] &= c^{\rm 1m}_{i-3,i-2,i-1,i} \\
      \mathrm{Coeff}[x^2_{ij} \otimes x^2_{i,j-1}] &= c^{\rm 2mh}_{j-2,j-1,j,i} \\
        \mathrm{Coeff}[x^2_{ij} \otimes x^2_{i,j+1}] &= c^{\rm 2mh}_{j,j+1,j+2,i} 
    \end{align}
    where $|j-i|>3$.
   \item Otherwise we have the following:
\begin{align}
  \mathrm{Coeff}[x^2_{ij} \otimes x^2_{il}]&= \frac{1}{2}(C^+_{ij,l}- C^+_{jl,i}+C^-_{l i,j})\\
\mathrm{Coeff}^{\mathrm{LS}}[x^2_{ij} \otimes x^2_{il}]&=\begin{cases}
         C^-_{l i,j},  \, \text{if $l=i-2$}.\\
         0, \, \text{otherwise.}
        \end{cases}\\
        d_{ijl}&=\begin{cases}
         \frac{1}{2}(C^+_{ij,l}- C^+_{jl,i}-C^-_{l i,j}),  \, \text{if $l=i-2$}.\\
          \mathrm{Coeff}[x^2_{ij} \otimes x^2_{il}], \, \text{otherwise.}
        \end{cases}
\end{align}
\end{enumerate}    
\end{proposition}
The above follow by inspection of the symbols of the different boxes, with the $C^\pm_{ij,l}$ as defined in \cref{def:Cpm}. A table where all the cases are spelled out explicitly is found in \cref{app:coeff}. We observe that the latter is the case where most of the coefficients are not captured by the LS terms.
\paragraph{Two indices distinct.}
Due to cyclic symmetry, we can consider the word $x^2_{1j} \otimes x^2_{1j}$ without loss of generality.
\begin{proposition}\label{prop:repeated}
In the case of repeated letters, we have the following:    
\begin{enumerate}
    \item $\mathrm{Coeff}[x^2_{1j} \otimes x^2_{1j}]=\mathrm{Coeff}^{\mathrm{LS}}[x^2_{1j} \otimes x^2_{1j}]=-c^{\mathrm{2mh}}_{1,2,j,n}
- c^{\mathrm{2mh}}_{1,j-1,j,j+1}$, for $n \geq 6$ and $4\leq j\leq n-2$. 
\item 
$\mathrm{Coeff}[x^2_{13} \otimes x^2_{13}]= \frac{1}{2}\mathrm{Coeff}^{\mathrm{LS}}[x^2_{13} \otimes x^2_{13}]=\frac{1}{2} \left(-c^{\mathrm{1m}}_{1,2,3,n}- c^{\mathrm{1m}}_{1,2,3,4}-\sum_{j=5}^{n-1}c^{\mathrm{2mh}}_{1,2,3,j} \right)$, for $n \geq 6$.
\end{enumerate}
\end{proposition}
A proof is given in \cref{sec:simplifications}.

\section{Two-mass triangle relations}
\label{sec:2mtri}
In this section, we present the two-mass triangle relations $t_{i,i+1,j}$, which were re-derived in our work, in the process of understanding the cancellation of letters of the form $x_{ij}^2-x_{kl}^2$. They hold for any multiplicity and helicity at one-loop and can be stated as
\begin{equation}\label{eq:2mrelation}
    t_{i,i+1,j}=0,
\end{equation}
where for $j \notin \{i-2,i-1,i,i+1,i+2,i+3\}$,
\begin{equation}\label{eq:2mrelationRHS}
    t_{i,i+1,j}:= c^{\mathrm{2mh}}_{i,i+1,i+2,j}+c^{\mathrm{2me}}_{i,i+1,j-1,j}-c^{\mathrm{2me}}_{i,i+1,j,j+1}-c^{\mathrm{2mh}}_{i,i+1,j,i-1}+\sum_{k=i+3}^{j-2} c^{\mathrm{3m}}_{i,i+1,k,j}-\sum_{k=j+2}^{i-2} c^{\mathrm{3m}}_{i,i+1,j,k},
\end{equation} 
and in the edge cases $j=i-2$ and $j=i+3$, we account for a degeneration in the terms,
\begin{align}
     t_{i,i+1,i+3} &:= c^{\mathrm{1m}}_{i,i+1,i+2,i+3}-c^{\mathrm{2me}}_{i,i+1,i+3,i+4}-c^{\mathrm{2mh}}_{i,i+1,i+3,i-1}-\sum_{k=i+5}^{i-2} c^{\mathrm{3m}}_{i,i+1,i+3,k}, \\
      t_{i,i+1,i-2} &:= c^{\mathrm{2mh}}_{i,i+1,i+2,i-2}+c^{\mathrm{2me}}_{i,i+1,i-3,i-2}-c^{\mathrm{1m}}_{i,i+1,i-2,i-1}+\sum_{k=i+3}^{i-4} c^{\mathrm{3m}}_{i,i+1,k,i-2}.
\end{align}
The relations can be understood in a diagrammatic way by considering a two-mass triangle $T_{i,i+1,j}$ as in \cref{fig:2mtriangle} and partitioning its momenta into four sets in all possible ways compatible with the original triple partition. 
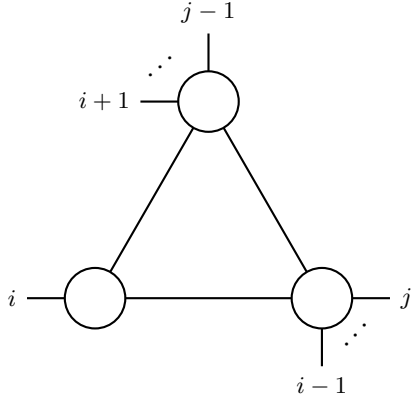
\begin{figure}[ht]
\centering

\begin{tikzpicture}

\def\rblob{0.4}
\def\rext{0.5}

\coordinate (A) at (0,0);
\coordinate (B) at (3,0);
\coordinate (C) at (1.5,2.6); 

\draw[thick] (A) -- (B) -- (C) -- cycle;

\node[circle, draw, fill=white, thick, minimum size=0.8cm, inner sep=0pt] at (A){};
\node[circle, draw, fill=white, thick, minimum size=0.8cm, inner sep=0pt] at (B) {};
\node[circle, draw, fill=white, thick, minimum size=0.8cm, inner sep=0pt] at (C) {};

\draw[thick] ($ (A) + (-\rblob,0) $) -- ++(-\rext,0) node[left]{\footnotesize $i$};
\draw[thick] ($ (C) + (-\rblob,0) $) -- ++(-\rext,0) node(a)[left]{\footnotesize $i+1$}; 
\draw[thick] ($ (C) + (0,\rblob) $) -- ++(0,\rext) node(b)[above]{\footnotesize $j-1$}; 
\path (a) -- node[sloped,yshift=-4pt]{$\cdots$}
(b);
\draw[thick] ($ (B) + (\rblob,0) $) -- ++(\rext,0) node(c)[right]{\footnotesize $j$}; 
\draw[thick] ($ (B) + (0,-\rblob) $) -- ++(0,-\rext) node(d)[below]{\footnotesize $i-1$}; 
\path (c) -- node[sloped,yshift=4pt]{$\cdots$} 
(d);

\end{tikzpicture}
\caption{Two-mass triangle $T_{i,i+1,j}$ corresponding to \cref{eq:2mrelation}.}
\label{fig:2mtriangle}
\end{figure}
\subsection{Simplifications of the one-loop symbol}\label{sec:simplifications}
\par The two-mass triangle relations significantly simplify the one-loop symbol. We can use them to prove two of the propositions of the previous section. 
\paragraph{Proof of Proposition \ref{prop:cancel}.}
Letters of the form $x^2_{ij} - x^2_{kl}$ only appear in the box functions for precisely  those boxes that can be contracted to a two-mass triangle. Boxes contractible to the two-mass triangle $T_{i,i+1,j}$ can include the letter $x^2_{i+1,j} - x^2_{j,i}$. It is enough to prove the absence of this letter for a specific choice of $i$. We will take $i=n$ and then relabel $j$ as $i$. 
By inspection of \cref{eq:1mbox-symbol,eq:2mhbox-symbol,eq:2mebox-symbol,eq:3mbox-symbol,eq:4mbox-symbol},
the letter in question moreover appears only within the word 
$w_{n1i}=\frac{x^2_{1 i}}{x^2_{i n}} \otimes (x^2_{1 i}-x^2_{i n})$.
The coefficient of this word  in the full amplitude is equal to
\begin{equation}
\mathrm{Coeff}[w_{n1i}] = c^{\mathrm{2me}}_{n,1,i-1,i} - c^{\mathrm{2me}}_{n,1,i,i+1} + c^{\mathrm{2mh}}_{n,1,2,i} - c^{\mathrm{2mh}}_{n-1,n,1,i} + \sum_{k=3}^{i-2} c^{\mathrm{3m}}_{n,1,k,i} - \sum_{j=i+2}^{n-2} c^{\mathrm{3m}}_{n,1,i,j}.
\end{equation}
Using \cref{eq:2mrelation} we can generate the two-mass triangle relation $t_{n,1,i}$, which is equivalent to $\mathrm{Coeff}[w_{n1i}]=0$. \qed
\paragraph{Proof of Proposition \ref{prop:repeated}.}
We consider the word $x^2_{ij} \otimes x^2_{ij}$. Without loss of generality, we can set one of the indices to 1 and consider $x^2_{1j} \otimes x^2_{1j}$. In the case where $n \geq 7$ and $4<j<n-2$, its coefficient is
\begin{align}
\mathrm{Coeff}[x^2_{1j} \otimes x^2_{1j}]=& 
-2 c^{\mathrm{2mh}}_{1,2,j,n}
-2 c^{\mathrm{2mh}}_{1,j-1,j,j+1}
+\half c^{\mathrm{2mh}}_{1,j,j+1,j+2}
+\half c^{\mathrm{2mh}}_{1,j,n-1,n}
+\half c^{\mathrm{2mh}}_{1,2,3,j} \nonumber\\ &
 +\half c^{\mathrm{2mh}}_{1,j-2,j-1,j}
+
c^{\mathrm{2me}}_{1,2,j-1,j} + c^{\mathrm{2me}}_{1,j,j+1,n} 
- c^{\mathrm{2me}}_{1,2,j,j+1} - c^{\mathrm{2me}}_{1,j-1,j,n}
 \nonumber \\ &
+ \sum_{k=j+3}^{n-1} \half c^{\mathrm{3m}}_{1,j,j+1,k}
+ \sum_{k=j+2}^{n-2} \half c^{\mathrm{3m}}_{1,j,k,n}
  +\sum_{k=4}^{j-2} \half c^{\mathrm{3m}}_{1,2,k,j}+\sum_{k=3}^{j-3} \half c^{\mathrm{3m}}_{1,k,j-1,j}
\nonumber \\ &
- \sum_{k=j+2}^{n-1} \half c^{\mathrm{3m}}_{1,j-1,j,k}
- \sum_{k=j+2}^{n-1} \half c^{\mathrm{3m}}_{1,2,j,k}
- \sum_{k=3}^{j-2}\half c^{\mathrm{3m}}_{1,k,j,j+1}
- \sum_{k=3}^{j-2} \half c^{\mathrm{3m}}_{1,k,j,n} \nonumber\\
=&\frac{1}{2}\left(t_{1,2,j}-t_{j-1,j,1}+t_{j,j+1,1}-t_{n,1,j}\right)-c^{\mathrm{2mh}}_{1,2,j,n}
- c^{\mathrm{2mh}}_{1,j-1,j,j+1}.
\end{align}
Therefore, the terms appearing in the left-hand sides of various two-mass triangle relations of \cref{eq:2mrelation} vanish. The coefficient is then reduced to
\begin{equation}
\mathrm{Coeff}[x^2_{1j} \otimes x^2_{1j}]=-c^{\mathrm{2mh}}_{1,2,j,n}
- c^{\mathrm{2mh}}_{1,j-1,j,j+1}=\mathrm{Coeff}^{\mathrm{LS}}[x^2_{1j} \otimes x^2_{1j}].
\end{equation} 
We note that this way, $4n-18$ terms reduce to only two.

For the case where the repeated letter is a two-particle invariant, the coefficient of $x^2_{13} \otimes x^2_{13}$, for $n \geq 7$ is
\begin{align}
\mathrm{Coeff}[x^2_{13} \otimes x^2_{13}]=&- c^{\mathrm{1m}}_{1,2,3,n}
- c^{\mathrm{1m}}_{1,2,3,4}
+\half c^{\mathrm{2mh}}_{1,3,4,5}
+\half c^{\mathrm{2mh}}_{1,3,n-1,n}
+ c^{\mathrm{2me}}_{1,3,4,n} 
 \nonumber \\ &
+ \half\sum_{j=6}^{n-1} c^{\mathrm{3m}}_{1,3,4,j}
+ \half\sum_{j=5}^{n-2} c^{\mathrm{3m}}_{1,3,j,n}
- \half\sum_{j=5}^{n-1} c^{\mathrm{2mh}}_{1,2,3,j}\\
=&\frac{1}{2}\left(t_{3,4,1}-t_{n,1,3}- c^{\mathrm{1m}}_{1,2,3,n}- c^{\mathrm{1m}}_{1,2,3,4}-\sum_{j=5}^{n-1}c^{\mathrm{2mh}}_{1,2,3,j}\right).
\end{align}
The expression reduces to
\begin{eqnarray}
    \mathrm{Coeff}[x^2_{13} \otimes x^2_{13}]=\frac{1}{2}\left(- c^{\mathrm{1m}}_{1,2,3,n}- c^{\mathrm{1m}}_{1,2,3,4}-\sum_{j=5}^{n-1}c^{\mathrm{2mh}}_{1,2,3,j}\right)=\frac{1}{2} \mathrm{Coeff}^{\mathrm{LS}}[x^2_{13} \otimes x^2_{13}].
\end{eqnarray}
The special cases of $j=4$ and $n=6$ are proven in \cref{ref:app_special}.  \qed
\subsection{Proof of relations from on-shell forms}\label{sec:sketchproof}

The two-mass triangle relations of \cref{eq:2mrelation} can be proved as a corollary of the following proposition. Without loss of generality, we can take the massless leg to correspond to the external momentum $p_n$.
\begin{proposition}\label{prop:2mtrls}
For any $i$ such that $3\leq i\leq n-2$,
    \begin{eqnarray}
    \sum_{j=2}^{i-1} {\rm LeS}_{1,j,i,n}^{\pm}[A^{(1)}_n]= \sum_{k=i+1}^{n-1}{\rm LeS}^{\pm}_{1,i,k,n}[A^{(1)}_n].
\end{eqnarray}
\end{proposition}
We note that the above relation is a refinement of the relations between box coefficients that have appeared in \cite{Brandhuber:2009rel}, since they equate the leading singularities on each branch of the quadruple cut solution, separately. We present here a summary of the proof, which can be found in detail in \cref{ref:app_2mt_relations}.
\begin{proof}
\textbf{Parametrizing the triple cut of the two-mass triangle.}  
We use dual notation, in which for example $x_{1 j} = p_1+\dots +p_{j-1}$.
The triple cut of the two-mass triangle involves 
putting on shell the three propagators indicated in \cref{fig:triplecut}.
\begin{figure}[ht]
\centering

\begin{tikzpicture}

\def\rblob{0.4}
\def\rext{0.5}

\coordinate (A) at (0,0);
\coordinate (B) at (3,0);
\coordinate (C) at (1.5,2.6); 

\draw[thick] (A) -- (B) -- (C) -- cycle;

\def\l{0.26}
\coordinate (M1) at ($ (A)!0.5!(B) $);
\coordinate (M2) at ($ (B)!0.5!(C) $);
\coordinate (M3) at ($ (C)!0.5!(A) $);

\draw[dashed, thick, red] ($ (M1) + (0,\l) $) -- ($ (M1) + (0,-\l) $);
\draw[dashed, thick, red] ($ (M2) + (30:\l) $) -- ($ (M2) + (210:\l) $);
\draw[dashed, thick, red] ($ (M3) + (150:\l) $) -- ($ (M3) + (330:\l) $); 

\node[circle, draw, fill=white, thick, minimum size=0.8cm, inner sep=0pt] at (A) {\footnotesize $A_{n,1}$};
\node[circle, draw, fill=white, thick, minimum size=0.8cm, inner sep=0pt] at (B) {\footnotesize ${A}_{i,n}$};
\node[circle, draw, fill=white, thick, minimum size=0.8cm, inner sep=0pt] at (C) {\footnotesize ${A}_{1,i}$};

\draw[thick] ($ (A) + (-\rblob,0) $) -- ++(-\rext,0) node[left]{\footnotesize $n$};
\draw[thick] ($ (C) + (-\rblob,0) $) -- ++(-\rext,0) node(a)[left]{\footnotesize $1$}; 
\draw[thick] ($ (C) + (0,\rblob) $) -- ++(0,\rext) node(b)[above]{\footnotesize $i-1$}; 
\path (a) -- node[sloped,yshift=-4pt]{$\cdots$}
(b);
\draw[thick] ($ (B) + (\rblob,0) $) -- ++(\rext,0) node(c)[right]{\footnotesize $i$}; 
\draw[thick] ($ (B) + (0,-\rblob) $) -- ++(0,-\rext) node(d)[below]{\footnotesize $n-1$}; 
\path (c) -- node[sloped,yshift=4pt]{$\cdots$} 
(d);

\def\off{0.5}
\def\Len{0.3}
\draw[<-, thick] ($ (M1) + (0,-\off) - \Len*(1,0) $) -- ($ (M1) + (0,-\off) + \Len*(1,0) $) node[midway, below]{\footnotesize $\hat\ell_n$};
\draw[<-, thick] ($ (M2) + (30:\off) + (300:\Len) $) -- ($ (M2) + (30:\off) + (120:\Len) $) node[midway, right, yshift=3pt]{\footnotesize $\hat\ell_i$};
\draw[<-, thick] ($ (M3) + (150:\off) + (60:\Len) $) -- ($ (M3) + (150:\off) + (240:\Len) $) node[midway, left, yshift=3pt]{\footnotesize $\hat\ell_1$};
\end{tikzpicture}
\caption{Triple cut of triangle generating two-mass triangle relations. The external momenta are all taken to be incoming. The helicity of the three-point amplitude  $A_{n,1}$ at the vertex with external momentum $p_n$ distinguishes the different irreducible components of cut solutions.}
\label{fig:triplecut}
\end{figure}
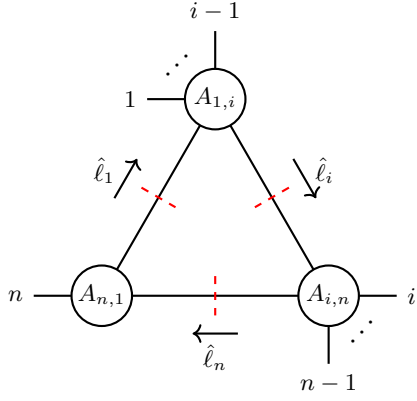
This results in a one-dimensional variety in loop-momentum space $\mathbb{C}^{1,3}$ that has two irreducible components that can be parametrized as
\begin{equation}\label{eq:cutpar}
 \hat{\ell}_1(z) = \frac{x_{1 i}^2}{x_{1 i}^2-x_{i n}^2} p_n  + z q,
\end{equation}
where $z \in \mathbb{C}$ and
\begin{equation} \label{eq:twocompon}
    q=\lambda_q \tilde\lambda_q, \qquad \mbox{and}~
    \begin{cases}
        \lambda_q = \lambda_n, \quad \tilde\lambda_q=x_{1i}\cdot\lambda_n, \quad \mbox{or} \\
        \lambda_q = x_{1i}\cdot\tilde\lambda_n, \quad \tilde\lambda_q=\tilde\lambda_n.
    \end{cases}
\end{equation}

\textbf{On-shell form of the triple cut of the two-mass triangle.}
We can associate an on-shell form to the triple cut of the two-mass triangle by taking the residue of the full integrand of the amplitude on the cut. This is obtained by multiplying tree amplitudes in the corners and summing over intermediate states. If this were a maximal cut, we would obtain a ``leading singularity'' function, whereas in this case we get a ``next-to-leading singularity''. In our parametrisation of the cut in \cref{eq:cutpar}, the on-shell $1$-form $\omega$ in $z$ is:

\begin{align}\label{eq:triplemain}
\omega(z)=&J \, \frac{\mathrm{d}z}{z} \prod_{r=1,i,n} \mathrm{d}^4 \eta_{\ell_r}A_{i,n}(\hat{\ell}_i(z),p_{i},\dots,p_{n-1},-\hat{\ell}_n(z)) \times \nonumber
\\
& \qquad \qquad A_{n,1}(\hat{\ell}_n(z),p_n,-\hat{\ell}_1(z))A_{1,i}(\hat{\ell}_1(z),p_1,\dots,p_{i-1},-\hat{\ell}_i(z)),
\end{align}
where 
$J=(x_{i n}^2-x_{1 i}^2)^2$ is a Jacobian factor. We omit the dependence of the amplitudes on the Grassmann variables for brevity. Below, we will focus on the second component of \cref{eq:twocompon}. The treatment for the first component is analogous.

\textbf{Singularities of the on-shell form.}
We now study the poles of $\omega(z)$, which are located where some fourth propagator between $j-1$ and $j$ goes on shell, i.e.~where $z=z_j$ such that $\hat\ell_j(z_j)^2=0$.
The solution to this on-shell condition is
\begin{align}
    z_j 
     &= \frac{x^2_{1 j} x_{n i}^2-x_{1 i}^2 x_{n j}^2}{(x^2_{n i}-x^2_{1 i}) \lbrack n | x_{1 i} x_{1 j}|n \rbrack}.
\end{align}
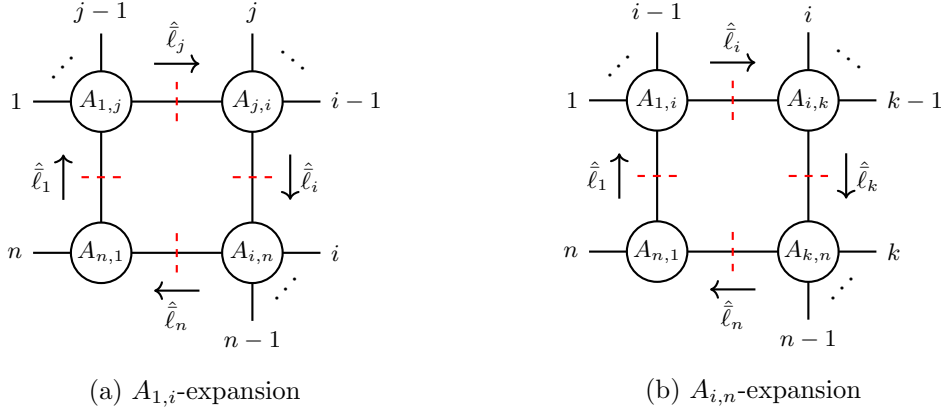
\begin{figure}[ht]
\centering
\begin{subfigure}[h]{0.48\textwidth}
\centering
\begin{tikzpicture}

\def\rblob{0.4}
\def\rext{0.5}

\coordinate (A) at (-1,-1);
\coordinate (B) at (1,-1);
\coordinate (C) at (1,1);
\coordinate (D) at (-1,1);

\draw[thick] (A) -- (B) -- (C) -- (D) -- cycle;

\def\l{0.26}
\coordinate (M1) at ($ (A)!0.5!(B) $);
\draw[dashed, thick, red] ($ (M1) + (0,\l) $) -- ($ (M1) + (0,-\l) $);
\coordinate (M2) at ($ (B)!0.5!(C) $);
\draw[dashed, thick, red] ($ (M2) + (-\l,0) $) -- ($ (M2) + (\l,0) $);
\coordinate (M3) at ($ (C)!0.5!(D) $);
\draw[dashed, thick, red] ($ (M3) + (0,-\l) $) -- ($ (M3) + (0,\l) $);
\coordinate (M4) at ($ (D)!0.5!(A) $);
\draw[dashed, thick, red] ($ (M4) + (\l,0) $) -- ($ (M4) + (-\l,0) $);

\node[circle, draw, fill=white, thick, minimum size=0.8cm,inner sep=0pt] (a) at (A) {\footnotesize $A_{n,1}$};
\node[circle, draw, fill=white, thick, minimum size=0.8cm,inner sep=0pt] (b) at (B) {\footnotesize ${A}_{i,n}$};
\node[circle, draw, fill=white, thick, minimum size=0.8cm,inner sep=0pt] (c) at (C) {\footnotesize $A_{j,i}$};
\node[circle, draw, fill=white, thick, minimum size=0.8cm,inner sep=0pt] (d) at (D) {\footnotesize $A_{1,j}$};


\draw[thick] ($ (A) + (-\rblob,0) $) -- ++(-\rext,0) node[left]{\footnotesize $n$}; 

\draw[thick] ($ (B) + (\rblob,0) $) -- ++(\rext,0) node(c)[right]{\footnotesize $i$}; 
\draw[thick] ($ (B) + (0,-\rblob) $) -- ++(0,-\rext) node(d)[below]{\footnotesize $n-1$}; 
\path (c) -- node[sloped,yshift=4pt]{$\cdots$} (d);

\draw[thick] ($ (C) + (\rblob,0) $) -- ++(\rext,0) node(e)[right]{\footnotesize $i-1$}; 
\draw[thick] ($ (C) + (0,\rblob) $) -- ++(0,\rext) node(f)[above]{\footnotesize $j$}; 
\path (e) -- node[sloped,yshift=-4pt]{$\cdots$} (f);

\draw[thick] ($ (D) + (-\rblob,0) $) -- ++(-\rext,0) node(a)[left]{\footnotesize $1$}; 
\draw[thick] ($ (D) + (0,\rblob) $) -- ++(0,\rext) node(b)[above]{\footnotesize $j-1$}; 
\path (a) -- node[sloped,yshift=-4pt]{$\cdots$} (b);


\def\off{0.5}

\def\Len{0.3}

\draw[<-, thick]($ (A)!0.5!(B) + (0,-\off) - \Len*(1,0) $) -- ($ (A)!0.5!(B) + (0,-\off) + \Len*(1,0) $) node[midway, below]{\footnotesize $\hat{\bar{\ell}}_n$};

\draw[<-, thick]($ (B)!0.5!(C) + (\off,0) - \Len*(0,1) $) -- ($ (B)!0.5!(C) + (\off,0) + \Len*(0,1) $) node[midway, right]{\footnotesize $\hat{\bar{\ell}}_i$};

\draw[->, thick]($ (A)!0.5!(D) + (-\off,0) - \Len*(0,1) $) -- ($ (A)!0.5!(D) + (-\off,0) + \Len*(0,1) $) node[midway, left]{\footnotesize $\hat{\bar{\ell}}_1$};

\draw[->, thick]($ (C)!0.5!(D) + (0,\off) - \Len*(1,0) $) -- ($ (C)!0.5!(D) + (0,\off) + \Len*(1,0) $) node[midway, above]{\footnotesize $\hat{\bar{\ell}}_j$};

\end{tikzpicture}
\caption{$A_{1,i}$-expansion}
\label{fig:Pexp}
\end{subfigure}
\begin{subfigure}[h]{0.48\textwidth}
\centering
\begin{tikzpicture}

\def\rblob{0.4}
\def\rext{0.5}

\coordinate (A) at (-1,-1);
\coordinate (B) at (1,-1);
\coordinate (C) at (1,1);
\coordinate (D) at (-1,1);

\draw[thick] (A) -- (B) -- (C) -- (D) -- cycle;

\def\l{0.26}
\coordinate (M1) at ($ (A)!0.5!(B) $);
\draw[dashed, thick, red] ($ (M1) + (0,\l) $) -- ($ (M1) + (0,-\l) $);
\coordinate (M2) at ($ (B)!0.5!(C) $);
\draw[dashed, thick, red] ($ (M2) + (-\l,0) $) -- ($ (M2) + (\l,0) $);
\coordinate (M3) at ($ (C)!0.5!(D) $);
\draw[dashed, thick, red] ($ (M3) + (0,-\l) $) -- ($ (M3) + (0,\l) $);
\coordinate (M4) at ($ (D)!0.5!(A) $);
\draw[dashed, thick, red] ($ (M4) + (\l,0) $) -- ($ (M4) + (-\l,0) $);

\node[circle, draw, fill=white, thick, minimum size=0.8cm,inner sep=0pt] (a) at (A) {\footnotesize $A_{n,1}$};
\node[circle, draw, fill=white, thick, minimum size=0.8cm,inner sep=0pt] (b) at (B) {\footnotesize $A_{k,n}$};
\node[circle, draw, fill=white, thick, minimum size=0.8cm,inner sep=0pt] (c) at (C) {\footnotesize $A_{i,k}$};
\node[circle, draw, fill=white, thick, minimum size=0.8cm,inner sep=0pt] (d) at (D) {\footnotesize ${A}_{1,i}$};


\draw[thick] ($ (A) + (-\rblob,0) $) -- ++(-\rext,0) node[left]{\footnotesize $n$}; 

\draw[thick] ($ (B) + (\rblob,0) $) -- ++(\rext,0) node(c)[right]{\footnotesize $k$}; 
\draw[thick] ($ (B) + (0,-\rblob) $) -- ++(0,-\rext) node(d)[below]{\footnotesize $n-1$}; 
\path (c) -- node[sloped,yshift=4pt]{$\cdots$} (d);

\draw[thick] ($ (C) + (\rblob,0) $) -- ++(\rext,0) node(e)[right]{\footnotesize $k-1$}; 
\draw[thick] ($ (C) + (0,\rblob) $) -- ++(0,\rext) node(f)[above]{\footnotesize $i$}; 
\path (e) -- node[sloped,yshift=-4pt]{$\cdots$} (f);

\draw[thick] ($ (D) + (-\rblob,0) $) -- ++(-\rext,0) node(a)[left]{\footnotesize $1$}; 
\draw[thick] ($ (D) + (0,\rblob) $) -- ++(0,\rext) node(b)[above]{\footnotesize $i-1$}; 
\path (a) -- node[sloped,yshift=-4pt]{$\cdots$} (b);


\def\off{0.5}

\def\Len{0.3}

\draw[<-, thick]($ (A)!0.5!(B) + (0,-\off) - \Len*(1,0) $) -- ($ (A)!0.5!(B) + (0,-\off) + \Len*(1,0) $) node[midway, below]{\footnotesize $\hat{\bar{\ell}}_n$};

\draw[<-, thick]($ (B)!0.5!(C) + (\off,0) - \Len*(0,1) $) -- ($ (B)!0.5!(C) + (\off,0) + \Len*(0,1) $) node[midway, right]{\footnotesize $\hat{\bar{\ell}}_k$};

\draw[->, thick]($ (A)!0.5!(D) + (-\off,0) - \Len*(0,1) $) -- ($ (A)!0.5!(D) + (-\off,0) + \Len*(0,1) $) node[midway, left]{\footnotesize $\hat{\bar{\ell}}_1$};

\draw[->, thick]($ (C)!0.5!(D) + (0,\off) - \Len*(1,0) $) -- ($ (C)!0.5!(D) + (0,\off) + \Len*(1,0) $) node[midway, above]{\footnotesize $\hat{\bar{\ell}}_i$};

\end{tikzpicture}
\caption{$A_{i,n}$-expansion}
\label{fig:Qexp}
\end{subfigure}
\caption{Quadruple cuts. In case (a) we consider partitions of the external legs in $A_{1,i}$, so that $\hat{\bar{\ell}}_j^2=0$ where $1<j<i$, while in case (b) we consider partitions of the external legs in $A_{i,n}$, so that $\hat{\bar{\ell}}_k^2=0$ where $i<k<n$.}
\label{fig:quadruplecuts}
\end{figure}
\textbf{Residue theorem for the on-shell form.}
By arguments given in \cref{ref:app_2mt_relations},
the apparent pole at $z=0$ of the on-shell form $\omega$ is removable, and the residue at infinity vanishes.   Therefore the sum of the residues at all $z=z_j$ must be zero.\footnote{The triple cut of a two-mass triangle configuration must vanish on purely kinematic grounds, based on the vanishing of the minor Cayley determinant \cite{Abreu:2017ptx}.}
We apply the residue theorem to $\omega(z)$ and distinguish the cases where the cut propagator is located within $A_{1,i}$ and $A_{i,n}$, rewriting $j \to k $ in the latter case. See  \cref{fig:quadruplecuts}, where the combination of a hat and a bar is used for momenta fully constrained by the quadruple cut. We obtain
\begin{align}
0 =&    \left(\sum_{j=2}^{i-1} {\rm Res}_{z=z_j}+\sum_{k=i+1}^{n-1} {\rm Res}_{z=z_k} \right) \omega(z) \\
=& \sum_{j=2}^{i-1}
\frac{1}{x^2_{1 j} x_{n i}^2-x_{1 i}^2 x_{n j}^2} \int  \prod_{r=1,j,i,n}\mathrm{d}^4 \eta_{\hat{\ell}_r} A_{1,j}(z_j) A_{j,i}(z_j) A_{i,n}(z_j)A_{n,1}(z_j) \nonumber\\&+\sum_{k=i+1}^{n-1} \frac{1}{x^2_{1 k} x_{n i}^2-x_{1 i}^2 x_{n k}^2} \int  \prod_{r=1,i,k,n}\mathrm{d}^4 \eta_{\hat{\ell}_r}  A_{1,i}(z_k) A_{i,k}(z_k)A_{k,n}(z_k) A_{n,1}(z_k) 
\end{align}
The individual terms are the leading singularities of \cref{eq:maxcutsusyL} associated to the quadruple cuts on one of the components $\ell_*^+$. This is the solution of the maximal cut of propagators and the $+$ indicates we are on one of the two components. Noting the sign conventions of \cref{eq:sqrt-rho}, we obtain the result 
\begin{eqnarray}
    \sum_{j=2}^{i-1} {\rm LeS}_{1,j,i,n}^+[A^{(1)}_n]= \sum_{k=i+1}^{n-1}{\rm LeS}_{1,i,k,n}^+[A^{(1)}_n]
\end{eqnarray}
By considering the other component of the two-mass triangle cut, we obtain an analogous statement for the other solution $\ell^*_{-}$.
\end{proof}
Adding the relations for the two solutions gives the two-mass triangle relations of \cref{eq:2mrelation}.
\subsection{Relations from amplituhedron geometry}
\label{sec:geomint}

We now give a geometric interpretation of the two-mass triangle
relations of \cref{prop:2mtrls}. The residue-theoretic proof given in
\cref{sec:sketchproof} is independent of the amplituhedron; the purpose
of this subsection is instead to explain why the two sides of the
relation should arise geometrically as two different dissections of the
same projected cut region.

\paragraph{Projected cut geometries.}
Let $\mathcal A^{(L)}_{n,k}$ denote the $L$-loop amplituhedron
\cite{ArkaniHamedTrnka:2013,ArkaniHamedHodgesTrnka:2014}. Its canonical
form encodes the planar $\mathcal N=4$ SYM loop integrand; after the
standard extraction of the $Y$-measure, one obtains the usual
$4L$-form in the loop variables. A point of
$\mathcal A^{(L)}_{n,k}$ is written as
$
  (Y,\mathcal L_1,\ldots,\mathcal L_L),
$
with $Y\in\Gr(k,k+4)$ and
$\mathcal L_\ell\in\Gr(2,k+4)$ obtained from positive data
$(C,D_1,\ldots,D_L;Z)$. There is a natural projection to the tree
amplituhedron,
\begin{equation}
  \pi:\mathcal A^{(L)}_{n,k}\longrightarrow
  \mathcal A^{(0)}_{n,k},
  \qquad
  \pi(Y,\mathcal L_1,\ldots,\mathcal L_L)=Y .
\end{equation}

Let $\mathcal C$ be a cut variety defined by equations of the form
\begin{equation}
  \langle Y\,\mathcal L_\ell\,i{-}1\,i\rangle=0,
  \qquad
  \langle
  Y\,\mathcal L_{\ell_1}\,\mathcal L_{\ell_2}
  \rangle=0 .
\end{equation}
For physical propagator cuts these equations describe boundary
conditions of the loop amplituhedron
\cite{ArkaniHamedTrnka:2013Into,LangerSrikant:2019}. We denote the
corresponding cut geometry by
$
  \mathcal A_{\mathcal C}
  :=
  \mathcal A^{(L)}_{n,k}\cap\mathcal C .
$
If the cut is reducible, we will work with an individual irreducible
component $\mathcal B\subset\mathcal C$ and write
\begin{equation}
  \mathcal A_{\mathcal B}
  :=
  \mathcal A^{(L)}_{n,k}\cap\mathcal B,
  \qquad
  \mathcal T_{\mathcal B}
  :=
  \pi(\mathcal A_{\mathcal B})
  \subseteq \mathcal A^{(0)}_{n,k}.
\end{equation}

It is useful to formulate the relation to leading singularities at the
level of differential forms. Let
\begin{equation}
  \omega_{\mathcal B}
  :=
  \underset{\mathcal B}{\operatorname{Res}}\,
  \Omega(\mathcal A^{(L)}_{n,k})
\end{equation}
denote the residue of the loop-amplituhedron form on the chosen cut
component. If $\mathcal B_m$ is a maximal-cut component, there are no
remaining loop differentials and the pushforward to $Y$ gives the
corresponding leading singularity,
\begin{equation}
  \pi_*\omega_{\mathcal B_m}
  =
  {\rm LeS}_{\mathcal B_m}.
\end{equation}
Equivalently, after resolving the tree amplitudes at the corners into
on-shell diagrams, the leading singularity can be written as a sum of
canonical forms associated with positroid cells
\cite{ArkaniHamedBourjailyCachazoGoncharovPostnikovTrnka:2012}.
Whenever the relevant projected image is itself a positive geometry and
the map has the appropriate degree, this pushforward may be identified
with its canonical form.

For a non-maximal cut, such an identification is not automatic. The cut
still has positive-dimensional fibers in loop space, and the
set-theoretic projection $\mathcal T_{\mathcal B}$ need not
\emph{a priori} come equipped with a positive-geometry structure.
The geometric proposal in this section is that, for the two-mass
triangle cuts considered below, each projected branch
$\mathcal T_{\mathcal B}$ is a positive geometry which admits
dissections by projected maximal-cut pieces.

By a \emph{dissection} we mean a decomposition into pieces whose
interiors are pairwise disjoint and whose union is dense in the full geometry. Canonical forms are
additive under such decompositions
\cite{ArkaniHamedBaiLam:2017,KarpWilliamsZhang:2017}. Thus, if a
collection of maximal-cut components $\{\mathcal B_m\}_{m\in I}$
dissects $\mathcal T_{\mathcal B}$,
\begin{equation}
  \mathcal T_{\mathcal B}
  =
  \bigcup_{m\in I}
  \pi(\mathcal A_{\mathcal B_m}),
\end{equation}
then
\begin{equation}
  \Omega(\mathcal T_{\mathcal B})
  =
  \sum_{m\in I}
  \Omega\!\left(\pi(\mathcal A_{\mathcal B_m})\right)
  =
  \sum_{m\in I}
  {\rm LeS}_{\mathcal B_m}.
\end{equation}
Consequently, two dissections of the same projected cut geometry give
an identity among leading singularities.

This picture is structurally related to the fiber and
Jeffrey--Kirwan-residue descriptions of amplituhedron canonical forms
developed in
\cite{Ferro:2018vpf,Mohammadi:2020plf}. In those constructions a
canonical form is recovered from a distinguished sum of residues of a
form on the fibers of the amplituhedron map. Here the projection is a
different one---from a loop cut geometry to the tree amplituhedron---but
the underlying mechanism is similar: different residue decompositions
encode different geometric decompositions of the same projected object.

\paragraph{The two-mass triangle cut.}
We now specialize to $L=1$ and consider the triple cut associated with a
two-mass triangle whose massless corner carries momentum $p_n$. In
momentum-twistor notation,
\begin{equation}
\label{eq:triple-cut-amplituhedron}
  \langle Y\,\mathcal L\,n\,1\rangle
  =
  \langle Y\,\mathcal L\,i{-}1\,i\rangle
  =
  \langle Y\,\mathcal L\,n{-}1\,n\rangle
  =
  0 .
\end{equation}
We denote the corresponding cut by $\mathcal C_i$. At fixed external
kinematics, the triple cut leaves one complex parameter unfixed. The cut
has two components, corresponding to the two three-point solutions at
the massless corner, or equivalently to the two black/white choices for
the trivalent on-shell vertex. We denote them by
$\mathcal C_i^\pm$ and define
\begin{equation}
\label{eq:def-T-pm}
  \mathcal T^{i,\pm}_{n,k}
  :=
  \pi\!\left(
  \mathcal A^{(1)}_{n,k}\cap\mathcal C_i^\pm
  \right)
  \subseteq
  \mathcal A^{(0)}_{n,k}.
\end{equation}

A maximal cut compatible with the triple cut of \cref{eq:triple-cut-amplituhedron} is
obtained by cutting one additional propagator. There are two natural
families, depending on which massive corner is factorized. Splitting the
corner containing the particles in the cyclic interval $[1,i{-}1]$ gives
$\mathcal C^{\pm}_{1,j,i,n}$, with $1< j<i$, whereas splitting the opposite massive corner gives $ \mathcal C^{\pm}_{1,i,r,n}$, with $i< r < n$.
These are precisely the two families of poles of the one-dimensional
cut form in \cref{sec:sketchproof}.

The geometric statement suggested by the residue theorem is therefore
the following.

\begin{conjecture}
\label{conj:T-dissection}
For each branch $\sigma=\pm$, the projected cut region
$\mathcal T^{i,\sigma}_{n,k}$ is a positive geometry and admits the two
dissections
\begin{equation}
\label{eq:T-two-dissections}
  \mathcal T^{i,\sigma}_{n,k}
  =
  \bigcup_{j=2}^{i-1}
  \pi\!\left(
  \mathcal A^{(1)}_{n,k}\cap
  \mathcal C^{\sigma}_{1,j,i,n}
  \right)
 =
  \bigcup_{r=i+1}^{n-1}
  \pi\!\left(
  \mathcal A^{(1)}_{n,k}\cap
  \mathcal C^{\sigma}_{1,i,r,n}
  \right),
\end{equation}
where in each union the interiors of the projected pieces are pairwise
disjoint.
Equivalently, the two families of plabic graphs shown in
\cref{fig:General_T} describe two dissections of the same projected
positive geometry.
\end{conjecture}

The combinatorial T-duality operation on plabic graphs used in
\cref{fig:General_T} is the one described in
\cite{ParisiShermanBennettWilliams:2021}. The relation between positroid
cells, on-shell diagrams, and amplituhedron dissections is part of the
general positive-Grassmannian/amplituhedron framework
\cite{ArkaniHamedBourjailyCachazoGoncharovPostnikovTrnka:2012,
KarpWilliamsZhang:2017,EvenZoharLakrecTessler:2021}.

Taking canonical forms in \cref{eq:T-two-dissections} gives
\begin{equation}
\label{eq:T-canonical-relation}
  \sum_{j=2}^{i-1}
  {\rm LeS}^{\sigma}_{1,j,i,n}[A^{(1)}_n]
  =
  \sum_{r=i+1}^{n-1}
  {\rm LeS}^{\sigma}_{1,i,r,n}[A^{(1)}_n],
\end{equation}
which is exactly the componentwise two-mass triangle relation of
\cref{prop:2mtrls}.

The residue theorem of \cref{sec:sketchproof} gives an algebraic proof
of \cref{eq:T-canonical-relation} without assuming
\cref{conj:T-dissection}. Indeed, on either component
$\mathcal C_i^\sigma$, the triple-cut form is a meromorphic one-form in
the residual parameter $z$. Its poles are divided into the two families
above: factorization poles of the first massive corner produce the
leading singularities on the left-hand side of
\cref{eq:T-canonical-relation}, while factorization poles of the second
massive corner produce those on the right-hand side. Since there are no
additional residues at $z=0$ or $z=\infty$, the global residue theorem
equates the two sums. The geometric conjecture asserts that this
partition of the poles is the canonical-form shadow of two genuine
dissections of $\mathcal T^{i,\sigma}_{n,k}$.

It is important here to work branch by branch. The two components
$\mathcal C_i^+$ and $\mathcal C_i^-$ give different cut forms and, in
general, different projected regions
$\mathcal T^{i,+}_{n,k}$ and $\mathcal T^{i,-}_{n,k}$. Thus
\cref{prop:2mtrls} should be interpreted geometrically as an equality of
canonical forms associated with two dissections of each branch
separately, rather than as a statement only after the two cut solutions
have been summed.

\begin{figure}[htbp]
  \centering


\begin{tikzpicture}

\def\rblob{0.32}
\def\rsmall{0.10}
\def\rext{0.5}
\def\l{0.26}

\node at (-3.4,0) {$\displaystyle\sum_{k_1,k_2,k_3}\sum_{j=2}^{i-1}$};
\node at (2,0) {\Huge \(=\)};
\node at (3.3,0) {$\displaystyle\sum_{k_1,k_2,k_3}\sum_{r=i+1}^{n-1}$};

\begin{scope}[shift={(5.7,0)}]

\coordinate (A) at (-1,-1);
\coordinate (B) at ( 1,-1);
\coordinate (C) at ( 1, 1);
\coordinate (D) at (-1, 1);

\draw[thick] (A)--(B)--(C)--(D)--cycle;

\node[circle,draw,fill=white,thick,minimum size=0.64cm,inner sep=0pt] at (A) {\footnotesize $k_3$};
\node[circle,draw,fill=white,thick,minimum size=0.64cm,inner sep=0pt] at (B) {\footnotesize $k_2$};
\node[circle,draw,fill=white,thick,minimum size=0.64cm,inner sep=0pt] at (C) {\footnotesize $k_1$};
\node[circle,draw,fill=white,thick,minimum size=0.20cm,inner sep=0pt] at (D) {};

\draw[thick] ($(A)+(-\rblob,0)$)--++(-\rext,0) node(a1)[left] {\footnotesize $n-1$};
\draw[thick] ($(A)+(0,-\rblob)$)--++(0,-\rext) node(a2)[below] {\footnotesize $r$};
\path (a1)--node[sloped,yshift=4pt]{$\cdots$}(a2);

\draw[thick] ($(B)+(\rblob,0)$)--++(\rext,0) node(b1)[right] {\footnotesize $i$};
\draw[thick] ($(B)+(0,-\rblob)$)--++(0,-\rext) node(b2)[below] {\footnotesize $r-1$};
\path (b1)--node[sloped,yshift=4pt]{$\cdots$}(b2);

\draw[thick] ($(C)+(\rblob,0)$)--++(\rext,0) node(c1)[right] {\footnotesize $i-1$};
\draw[thick] ($(C)+(0,\rblob)$)--++(0,\rext) node(c2)[above] {\footnotesize $1$};
\path (c1)--node[sloped,yshift=-4pt]{$\cdots$}(c2);

\draw[thick] ($(D)+(135:\rsmall)$)--++(135:\rext) node[above left] {\footnotesize $n$};

\end{scope}

\begin{scope}[shift={(-1.0,0)}]

\coordinate (A) at (-1,-1);
\coordinate (B) at ( 1,-1);
\coordinate (C) at ( 1, 1);
\coordinate (D) at (-1, 1);

\draw[thick] (A)--(B)--(C)--(D)--cycle;

\node[circle,draw,fill=white,thick,minimum size=0.64cm,inner sep=0pt] at (A) {\footnotesize $k_3$};
\node[circle,draw,fill=white,thick,minimum size=0.64cm,inner sep=0pt] at (B) {\footnotesize $k_2$};
\node[circle,draw,fill=white,thick,minimum size=0.64cm,inner sep=0pt] at (C) {\footnotesize $k_1$};
\node[circle,draw,fill=white,thick,minimum size=0.20cm,inner sep=0pt] at (D) {};

\draw[thick] ($(A)+(-\rblob,0)$)--++(-\rext,0) node(a1)[left] {\footnotesize $n-1$};
\draw[thick] ($(A)+(0,-\rblob)$)--++(0,-\rext) node(a2)[below] {\footnotesize $i$};
\path (a1)--node[sloped,yshift=4pt]{$\cdots$}(a2);

\draw[thick] ($(B)+(\rblob,0)$)--++(\rext,0) node(b1)[right] {\footnotesize $j$};
\draw[thick] ($(B)+(0,-\rblob)$)--++(0,-\rext) node(b2)[below] {\footnotesize $i-1$};
\path (b1)--node[sloped,yshift=4pt]{$\cdots$}(b2);

\draw[thick] ($(C)+(\rblob,0)$)--++(\rext,0) node(c1)[right] {\footnotesize $j-1$};
\draw[thick] ($(C)+(0,\rblob)$)--++(0,\rext) node(c2)[above] {\footnotesize $1$};
\path (c1)--node[sloped,yshift=-4pt]{$\cdots$}(c2);

\draw[thick] ($(D)+(135:\rsmall)$)--++(135:\rext) node[above left] {\footnotesize $n$};

\end{scope}

\end{tikzpicture}

  \caption{T-dual plabic graphs labeling the two conjectural
  dissections of $\mathcal T^{i,-}_{n,k}$. The helicity degrees satisfy
  $k_1+k_2+k_3=k-1$. The indices $j$ and $r$ range over the values for which each
  vertex $v$ is either trivalent or satisfies
  $0\leq k_v\leq n_v-4$, where $n_v$ and $k_v$ denote respectively the
  degree and helicity degree of the vertex. For a trivalent white vertex
  we use the convention $k_v=-1$.}
  \label{fig:General_T}
\end{figure}
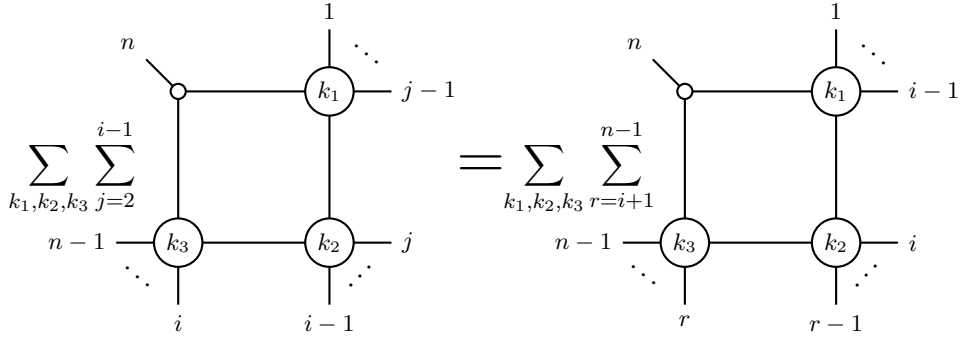

\begin{remark}[Helicity refinement]
\label{rem:helicity-refinement}
The projected geometry of the white branch admits a further
decomposition according to the helicity degrees carried by its two
massive blobs. Let
$\mathcal T^{i,-}_{n;k_a,k_b}$ denote the projected region represented
by the unresolved triangle graph $(w;k_a,k_b)$, where $w$ is the
trivalent white vertex and $k_a,k_b$ are the helicity degrees of the
two massive blobs. Since $k_w=-1$, the total helicity constraint is
$k_a+k_b=k$, and
\begin{equation}
\label{eq:T-helicity-decomposition}
  \mathcal T^{i,-}_{n,k}
  =
  \bigcup_{k_a+k_b=k}
  \mathcal T^{i,-}_{n;k_a,k_b}.
\end{equation}

For each fixed pair $(k_a,k_b)$, the two proposed dissections restrict
to the same subregion $\mathcal T^{i,-}_{n;k_a,k_b}$. Let
$\mathcal G^{(a)}_j(w;k_{a,1},k_{a,2},k_b)$ denote the projected
geometry represented by the box graph obtained by splitting the
$a$-blob at $j$, and define
$\mathcal G^{(b)}_r(w;k_a,k_{b,1},k_{b,2})$ analogously. The refined
geometric statement is
\begin{equation}
\label{eq:T-fixed-helicity-dissections}
  \mathcal T^{i,-}_{n;k_a,k_b}
  =
  \bigcup_{\substack{2\leq j\leq i-1\\
                     k_{a,1}+k_{a,2}=k_a-1}}
  \mathcal G^{(a)}_j
  (w;k_{a,1},k_{a,2},k_b)
  =
  \bigcup_{\substack{i+1\leq r\leq n-1\\
                     k_{b,1}+k_{b,2}=k_b-1}}
  \mathcal G^{(b)}_r
  (w;k_a,k_{b,1},k_{b,2}).
\end{equation}
The first dissection is obtained by triangulating the $a$-blob,
whereas the second is obtained by triangulating the $b$-blob. Taking
canonical forms therefore gives, for every admissible pair
$(k_a,k_b)$,
\begin{equation}
\label{eq:T-fixed-helicity-relation}
\sum_{\substack{2\leq j\leq i-1\\
                 k_{a,1}+k_{a,2}=k_a-1}}
\Omega\!\left(
  \mathcal G^{(a)}_j
  (w;k_{a,1},k_{a,2},k_b)
\right)
 =
\sum_{\substack{i+1\leq r\leq n-1\\
                 k_{b,1}+k_{b,2}=k_b-1}}
\Omega\!\left(
  \mathcal G^{(b)}_r
  (w;k_a,k_{b,1},k_{b,2})
\right),
\end{equation}
where only admissible helicity assignments are included. Summing these
identities over $k_a+k_b=k$ recovers the full two-mass triangle
relation on the $-$ branch. At the endpoints, we retain the conventions
$k_v=-1$ for a trivalent white vertex and $k_v=0$ for a trivalent
black vertex.
\end{remark}

\paragraph{The NMHV case.}
For $k=1$, the tree amplituhedron is the cyclic $4$-polytope
$C(n,4)$
\cite{ArkaniHamedTrnka:2013,ArkaniHamedHodgesTrnka:2014}.
For a cyclic interval $[a,b]$, write $P[a,b]
  :=
  \operatorname{conv}
  \{Z_a,Z_{a+1},\ldots,Z_b\}$,
where the convex hull is taken in a positive affine chart, and set $\Delta_{abcde}
  :=
  \operatorname{conv}
  \{Z_a,Z_b,Z_c,Z_d,Z_e\}$.

After extracting the $Y$-measure, their canonical forms are
\begin{equation}
\label{eq:NMHV-polytope-canonical-forms}
  \Omega(P[a,b])
  =
  A_{\rm NMHV}^{(0)}([a,b]),
  \qquad
  \Omega(\Delta_{abcde})
  =
  [a,b,c,d,e].
\end{equation}

Recall from \cref{eq:boxcoef} that
\begin{equation}
  c_{abcd}
  =
  \frac{1}{2}
  \left(
    {\rm LeS}^{+}_{abcd}
    +
    {\rm LeS}^{-}_{abcd}
  \right).
\end{equation}
We therefore use the single-branch contributions
\begin{equation}
\label{eq:def-branch-contribution-NMHV}
  \mathscr L^\sigma_{abcd}
  :=
  \frac{1}{2}\,
  {\rm LeS}^\sigma_{abcd},
  \qquad
  \sigma=\pm,
\end{equation}
so that
$c_{abcd}=\mathscr L^+_{abcd}+\mathscr L^-_{abcd}$.
We choose the branch labels so that the generic three-mass
contributions lie on the $-$ branch.

Following the standard momentum-twistor convention
\cite[Eq.~(1.1)]{ArkaniHamedTrnka:2013}, amplitudes, leading
singularities, and on-shell functions written in momentum-twistor
variables are understood with the universal tree-level MHV
superamplitude stripped off. Thus
$A^{(0)}_{n,\mathrm{MHV}}=1$ in momentum-twistor space, and the same
normalization is understood for box coefficients written in terms of
Yangian invariants.

\begin{theorem}[NMHV projected-cut geometry]
\label{thm:NMHV-projected-cut}
Let $Z$ be generic positive external data and suppose
$4\leq i\leq n-3$. Then the two dissection identities of
\cref{conj:T-dissection} hold for $k=1$. More precisely,
\begin{equation}
\label{eq:NMHV-projected-regions}
  \mathcal T^{i,+}_{n,1}
  =
  \Delta_{n-1,n,1,i-1,i},
  \qquad
  \mathcal T^{i,-}_{n,1}
  =
  P[n,i]\cup P[i-1,n],
\end{equation}
where the two polytopes in the second union have disjoint interiors.

Furthermore, the $-$ region has the two dissections
\begin{align}
\label{eq:NMHV-polytope-dissections}
  \mathcal T^{i,-}_{n,1}
  &=
  P[n,i-1]
  \cup
  \bigcup_{j=2}^{i-2}
  \Delta_{n,j-1,j,i-1,i}\cup P[i-1,n]
  \nonumber\\
  &=
  P[n,i]\cup P[i,n]
  \cup
  \bigcup_{r=i+2}^{n-1}
  \Delta_{n,i-1,i,r-1,r},
\end{align}
where the first line is the dissection obtained by factorizing the first massive
corner corresponding to dissecting $P[n,i]$, while the second is obtained by factorizing the opposite
massive corner corresponding to dissecting $P[i-1,n]$.

On the $+$ branch, the two families reduce respectively to the cuts
with $j=2$ and $r=n-1$, and both project to the same simplex
$\Delta_{n-1,n,1,i-1,i}$.
\end{theorem}

\begin{proof}
For $k=1$, a point of the nonnegative Grassmannian is represented by a
nonnegative row $C=(c_1,\ldots,c_n)$, modulo overall rescaling, and the
amplituhedron map is
\begin{equation}
  Y=CZ=\sum_{a=1}^{n}c_aZ_a.
\end{equation}
Consequently, the closure of the image of a rank-one positroid cell
with support $J\subseteq[n]$ is the cyclic subpolytope
$\operatorname{conv}\{Z_a:a\in J\}$.

We first identify the image of the triple cut. By the helicity
counting described above, for $k=1$ the $-$ branch has precisely the
two assignments $(k_a,k_b)=(1,0)$ and $(0,1)$. The
boundary-measurement maps of the corresponding diagrams give arbitrary
positive coefficients $c_a$ on the supports $[n,i]$ and $[i-1,n]$,
respectively, and hence parametrize the full rank-one positroid cells
with these supports. On the $+$ branch,
the only admissible configuration has both massive blobs MHV, and the
corresponding cell has support $\{n-1,n,1,i-1,i\}$. Including the
nonnegative boundary limits, the map above therefore sends these cells
onto $P[n,i]$, $P[i-1,n]$, and
$\Delta_{n-1,n,1,i-1,i}$, respectively. This establishes
\cref{eq:NMHV-projected-regions}.

We now resolve either massive corner into maximal cuts. Reading the
supports from the resulting NMHV on-shell diagrams gives the following
projected pieces. On the $-$ branch, the cuts in the first family give
\begin{equation}
  \Delta_{n,j-1,j,i-1,i},
  \qquad
  2\leq j\leq i-2,
\end{equation}
while the two helicity assignments at $j=i-1$ give
$P[n,i-1]$ and $P[i-1,n]$. In the second family, the two assignments
at $r=i+1$ give $P[n,i]$ and $P[i,n]$, while the remaining cuts give
\begin{equation}
  \Delta_{n,i-1,i,r-1,r},
  \qquad
  i+2\leq r\leq n-1.
\end{equation}
After resolving either massive corner on the $+$ branch, the only
nonempty NMHV contributions are the two endpoint cells, and both have
support $\{n-1,n,1,i-1,i\}$. Each of the two families therefore
consists of the single simplex
$\Delta_{n-1,n,1,i-1,i}$.

Having identified the full projected cut, it remains to verify that
resolving either massive corner gives a dissection on the $-$ branch.
The polytope $P[n,i]$ is obtained from $P[n,i-1]$ by adding the vertex
$Z_i$. By Gale's evenness criterion, the facets of $P[n,i-1]$ visible
from $Z_i$ are precisely \cite{Ziegler:1995}
\begin{equation}
  \operatorname{conv}
  \{Z_n,Z_{j-1},Z_j,Z_{i-1}\},
  \qquad
  2\leq j\leq i-2.
\end{equation}
Coning these facets to $Z_i$ gives
\begin{equation}
\label{eq:NMHV-left-polytope-triangulation}
  P[n,i]
  =
  P[n,i-1]
  \cup
  \bigcup_{j=2}^{i-2}
  \Delta_{n,j-1,j,i-1,i},
\end{equation}
with pairwise disjoint interiors. This is precisely the BCFW
triangulation of the $k=1$, $m=4$ tree amplituhedron on the cyclic
interval $[n,i]$, with the BCFW bridge between $i$ and $n$.

Similarly, adding $Z_{i-1}$ to $P[i,n]$ gives
\begin{equation}
\label{eq:NMHV-right-polytope-triangulation}
  P[i-1,n]
  =
  P[i,n]
  \cup
  \bigcup_{r=i+2}^{n-1}
  \Delta_{n,i-1,i,r-1,r},
\end{equation}
again with pairwise disjoint interiors.

Finally, the covector form of Gale's evenness criterion gives a
hyperplane containing
$\operatorname{conv}\{Z_n,Z_{i-1},Z_i\}$ such that
$Z_1,\ldots,Z_{i-2}$ lie strictly on one side and
$Z_{i+1},\ldots,Z_{n-1}$ lie strictly on the other. Therefore
\begin{equation}
\label{eq:NMHV-subpolytope-intersection}
  P[n,i]\cap P[i-1,n]
  =
  \operatorname{conv}\{Z_n,Z_{i-1},Z_i\}.
\end{equation}
In particular, the two polytopes have disjoint interiors. Combining
\cref{eq:NMHV-left-polytope-triangulation,eq:NMHV-right-polytope-triangulation,eq:NMHV-subpolytope-intersection}
proves \cref{eq:NMHV-polytope-dissections}.
\end{proof}

Taking canonical forms in the two cyclic-polytope dissections gives:
\begin{align}
\label{eq:NMHV-cyclic-triangulations}
  A_{\rm NMHV}^{(0)}([n,i])
  &=
  A_{\rm NMHV}^{(0)}([n,i-1])
  +
  \sum_{j=2}^{i-2}
  [n,j-1,j,i-1,i],
  \\
  A_{\rm NMHV}^{(0)}([i-1,n])
  &=
  A_{\rm NMHV}^{(0)}([i,n])
  +
  \sum_{r=i+2}^{n-1}
  [n,i-1,i,r-1,r].
\end{align}
The $-$ branch of the two-mass triangle relation is therefore
\begin{align}
\label{eq:NMHV-minus-relation}
  \sum_{j=2}^{i-1}
  \mathscr L^-_{1,j,i,n}
  &=
  A_{\rm NMHV}^{(0)}([n,i-1])
  +
  \sum_{j=2}^{i-2}
  [n,j-1,j,i-1,i]
 + A_{\rm NMHV}^{(0)}([i-1,n])
  \nonumber\\
  &=
  A_{\rm NMHV}^{(0)}([n,i])
  +
  A_{\rm NMHV}^{(0)}([i-1,n])
  \nonumber\\
  &=
  A_{\rm NMHV}^{(0)}([n,i])
  +
  A_{\rm NMHV}^{(0)}([i,n])
  +
  \sum_{r=i+2}^{n-1}
  [n,i-1,i,r-1,r]
  \nonumber\\
  &=
  \sum_{r=i+1}^{n-1}
  \mathscr L^-_{1,i,r,n}.
\end{align}
On the $+$ branch, both dissections consist of the same simplex, so
\begin{equation}
\label{eq:NMHV-plus-relation}
  \sum_{j=2}^{i-1}
  \mathscr L^+_{1,j,i,n}
  =
  [n-1,n,1,i-1,i]
  =
  \sum_{r=i+1}^{n-1}
  \mathscr L^+_{1,i,r,n}.
\end{equation}
Thus the branchwise two-mass triangle identities are precisely the
canonical-form identities associated with the two dissections in
\cref{eq:NMHV-polytope-dissections}.

The endpoint cases $i=3$ and $i=n-2$ follow by degeneration. In these
cases the coincident endpoint contribution should be treated as a
one-mass box, rather than simultaneously as two different two-mass
endpoint types.

\paragraph{Beyond NMHV.}
The preceding theorem establishes \cref{conj:T-dissection} at NMHV
level. For $k>1$, the tree amplituhedron is no longer an ordinary
convex polytope. In the first case beyond NMHV, namely $k=2$, we checked for $n\leq 9$
and every admissible value of $i$ that the projected pieces within each
proposed dissection in \cref{conj:T-dissection} have pairwise disjoint interiors in the tree
amplituhedron. These computations provide evidence for the geometric
conjecture but do not constitute a proof of the covering property. A
general proof would also require showing that the
two resulting collections cover the same region.

\subsection{Other relations} \label{sec:otherrel}
Let us now discuss the different types of relations among box coefficients. 
Only one-mass, two-mass, and three-mass box coefficients participate in linear relations, and there are a total of $n(n^2-8n+17)/2$ such boxes. In particular, there are 35 such boxes for $n=7$ and 68 for $n=8$.

First of all, the form of the IR divergences given in \cref{eq:generalir} implies linear relations. By equating the coefficients of functions associated to the  different Mandelstam invariants, we derive ${n(n-3)}/{2}$ relations between box coefficients. At order $\epsilon^{-2}$, we conclude in particular that for $n \geq 5$,\footnote{This is stated in \cite{Roiban:2004ix}, while two types of proofs can also be found in the appendix of \cite{Arkani-Hamed:2008owk}. A global residue calculation is presented in \cite{Brodel:2010pma}.}
\begin{align}\label{eq:onemassir}
\sum_{j=0}^{n-4}c^{\text{2mh}}_{i-2,i-1,i,i+j+1}=2A^{(0)}_n.
\end{align}
This also includes one-mass boxes as the corresponding limiting cases. 

Taking into account the $n(n-4)$ two-mass triangle relations, we have a larger set of linear relations between box coefficients. 
The two-mass triangle relations satisfy the extra condition that their sum is zero. This property can be checked explicitly from their definition in \cref{eq:2mrelationRHS}. We also notice that all the IR relations without the tree-level amplitude can be derived from the two-mass triangle ones. The only independent relation comes from \cref{eq:onemassir}. Thus, for $n \geq 7$, we arrive at a set of $n(n-4)$ helicity-independent linear relations among box coefficients. We have checked up to $n=12$ that there are no additional relations among two-mass triangle relations.
This enumeration is fully consistent with the one given in \cite{Brandhuber:2009rel}.

If we refine by helicity, we expect to find many more relations between the box coefficients. By numerically checking the space of all homogeneous linear relations, we find extra ones. Specifically, for $n=8$ we find four additional  independent relations at NMHV, and five at NNMHV. Two of them are common to both helicity cases. 

A conjectured number of independent on-shell superfunctions appearing at arbitrary loop order is given in \cite{Bourjaily:2023uln}, Eq.~K.11, and is equal to $\prod_{j=1}^{k}\binom{n-j}{4}\big/\binom{k-j+4}{4}$. In general, this is not directly comparable to the counting of box coefficients, since they are products of R-invariants. However, in the case of NMHV, the relation between coefficients and R-invariants is linear, so the two spaces can be compared directly. We note that the box coefficients span only the subspace of Yangian invariants realized as residues of one-loop maximal cuts, whereas  Eq.~K.11 of \cite{Bourjaily:2023uln} counts all those appearing at any loop order. This means that the formula of \cite{Bourjaily:2023uln} is an upper bound for the number of independent box coefficients, since the $R$-invariants arising from box maximal cuts have a constrained label structure, so a generic five-bracket is not produced at one loop.
For NMHV, the formula of \cite{Bourjaily:2023uln} simplifies to $\binom{n-1}{4}$. For $n=7$ we find no relations beyond the $21-1=20$ homogeneous ones. Note that the remaining relation is the one involving the tree amplitude, whose right-hand side is nonzero. It therefore fixes the value of a combination of coefficients rather than forcing it to vanish, and so does not reduce the dimension of their span. This leaves $35-20=15$ independent box coefficients, in exact agreement with $\binom{6}{4}=15$. For $n=8$, the four extra NMHV relations give $31+4=35$ homogeneous relations, hence $68-35=33$ independent coefficients, so that the box coefficients span a subspace of codimension two inside the $\binom{7}{4}=35$-dimensional space of R-invariants. 

We note that our analysis also included using the integrability of the symbol \cite{Duhr:2011zq} as a complementary approach to finding relations between box coefficients. Imposing the integrability condition produces a set of linear relations among the coefficients of the various words. However, after using the results of \cref{prop:4distinct,prop:3distinct,prop:repeated}, we find that the relations arising from integrability are not independent of  the IR relations of \cref{sec:otherrel} and the two-mass triangle relations.

\section{Cluster structure of the full one-loop amplitude}
\label{sec:clust}
In this section, we study cluster adjacency for the one-loop amplitudes
considered in this paper. 
The relevant cluster algebra is not only the usual
momentum-twistor cluster algebra of $\Gr(4,n)$
\cite{Scott:2003tq,Drummond:2017ssj}, but rather the
cluster algebra associated to the momentum-twistor flag variety
$\Fl_{2,4;n}$
\cite{GeissLeclercSchroer2008,BossingerLi:2024}.
 This is because, once infrared-divergent terms and
spinor-helicity brackets are included, the amplitude is no longer
expressed solely in terms of dual-conformal four-brackets; it also
contains brackets involving the infinity twistor.

Let $Z_1,\ldots,Z_n\in\mathbb P^3$ be momentum twistors
\cite{Hodges:2009hk,Mason:2009qx}; see also \cref{app:kinvar}. We choose the infinity twistor as $I=\langle A,B\rangle\subset \mathbb P^3$. Equivalently, the data $(I;Z_1,\ldots,Z_n)$ may be embedded into
$\Gr(4,n+2)$ by adjoining two extra columns $A,B$ to the
momentum-twistor matrix.  The cluster algebra of $\Fl_{2,4;n}$ is
obtained from the Grassmannian cluster algebra in a way compatible with
this embedding: the four-brackets $\langle a b c d\rangle$ and the
two-brackets $\langle a b\rangle_I=\langle A B a b\rangle$ are both
cluster coordinates of the flag variety \cite[Eqs.~(1.2), (1.5), Thm.~1.2 and Sec.~5.1]{BossingerLi:2024}.  We will use this embedding only
as a device for checking compatibility.  In particular, the final
statement will be independent of where the two columns spanning $I$ are
placed in the cyclic ordering.

The relation between dual variables and momentum twistors is given in \cref{eq:dualtotwist} as
\begin{equation}
    x_{ij}^2 = \frac{\langle i-1\,i\,j-1\,j\rangle}{\langle i-1,i\rangle_I\langle j-1,j\rangle_I}.
\end{equation}
Thus, the numerator is a dual-conformal four-bracket, while the denominator
contains the infinity twistor.  

We will use the following convention throughout this section.  Whenever
a symbol letter is a product or ratio of irreducible factors, we expand
it multiplicatively; for instance, $(ab)\otimes c=a\otimes c+b\otimes c$
and $(a/b)\otimes c=a\otimes c-b\otimes c$.  After applying
\cref{eq:dualtotwist}, each symbol word is therefore expanded
into words whose letters are irreducible four-brackets, two-brackets, or
other irreducible flag cluster coordinates.

We denote by $\operatorname{Exp}_{4B}\mathcal S[A]$ the part of the
expanded symbol whose letters are all independent of the infinity
twistor, i.e.~the part written only in terms of four-brackets.  
We also restrict to \emph{rational letters}, i.e. letters that are polynomials in the flag coordinates and, in particular, do not involve algebraic
quantities such as square roots.
Similarly,
$\operatorname{Exp}_{I}\mathcal S[A]$ denotes the complementary part,
consisting of symbol words with at least one letter involving $I$.  Thus
$\mathcal S[A]=\operatorname{Exp}_{4B}\mathcal S[A]+
\operatorname{Exp}_{I}\mathcal S[A]+\mathcal{S}^{\mathrm{4m}}$.

For the one-loop amplitudes, we use the standard decomposition
\cite{Drummond:2008vq,Elvang:2009ya}
\begin{equation}
  A^{(1)}_{n,k}
  =
  R^{(1)}_{n,k}\,A^{(0)}_{n,\mathrm{MHV}}
  +
  R^{(0)}_{n,k}\,A^{(1)}_{n,\mathrm{MHV}},
  \label{eq:one-loop-ampl-decomp}
\end{equation}
where $R^{(0)}_{n,k}$ and $R^{(1)}_{n,k}$ are the tree-level and
one-loop ratio functions.  The first term in
\cref{eq:one-loop-ampl-decomp} is dual conformal and contributes only to
the four-bracket part of the expanded symbol.  Therefore the
$I$-dependent part of the full one-loop amplitude comes entirely from the
universal MHV factor:
\begin{equation}
  \operatorname{Exp}_{I}\mathcal S[A^{(1)}_{n,k}]
  =
  R^{(0)}_{n,k}\,
  \operatorname{Exp}_{I}\mathcal S[A^{(1)}_{n,\mathrm{MHV}}].
  \label{eq:exp-I-from-MHV}
\end{equation}

\subsection{Cluster adjacency for rational letters}

We now recall the notion of cluster adjacency used below.  
A weight-two
symbol $\sum_\alpha f_\alpha\,a_\alpha\otimes b_\alpha$ satisfies
\emph{cluster adjacency} in $\Fl_{2,4;n}$ if, after expanding all rational letters
into irreducible flag cluster coordinates, every pair
$\{a_\alpha,b_\alpha\}$ appears in a common cluster of $\Fl_{2,4;n}$.
Frozen variables are compatible with every cluster variable.

We briefly recall the \emph{weak-separation} criterion for Pl\"ucker coordinates. 
Let $S,T\subset [n]$.  We say that $S$ and $T$ are \emph{weakly
separated} if after placing the labels $1,\ldots,n$ on a
circle, the two sets $S\setminus T$ and $T\setminus S$ do not cross.
When $|S|=|T|=k$, weak separation gives the standard compatibility
criterion for Pl\"ucker cluster variables in $\Gr(k,n)$: if $S$ and $T$
are weakly separated, then the Pl\"ucker coordinates $\langle S\rangle$
and $\langle T\rangle$ appear together in a cluster \cite[Thm.~1.6]{OhPostnikovSpeyer2015}.  

We will repeatedly use the following elementary consequences of weak
separation.

\begin{lemma}
\label{lem:basic-compatibility}
The following pairs of letters are compatible cluster variables in
$\Fl_{2,4;n}$.
\begin{enumerate}
  \item Any two identical letters.
  \item Two four-brackets $\langle S\rangle$ and $\langle T\rangle$ in
  $\Gr(4,n)$ whenever the corresponding four-subsets $S,T$ are weakly
  separated.  In particular, this holds if $S$ and $T$ share three
  indices.
  \item A two-bracket $\langle a b\rangle_I$ and a four-bracket
  $\langle S\rangle$ whenever $\{a,b\}\subset S$.
  \item Two consecutive two-brackets $\langle i\,i{+}1\rangle_I$ and
  $\langle j\,j{+}1\rangle_I$.
\end{enumerate}
\end{lemma}

\begin{proof}
After adjoining two columns $A,B$ spanning $I$, the two-bracket
$\langle a b\rangle_I$ becomes the Pl\"ucker coordinate
$\langle ABa b\rangle$ of $\Gr(4,n+2)$.  The statements then follow
from weak separation for Pl\"ucker coordinates in the Grassmannian cluster
algebra.  For example, if $\{a,b\}\subset S$, then the four-subsets
$\{A,B,a,b\}$ and $S=\{a,b,c,d\}$ are weakly separated. The same argument applies to consecutive two-brackets, which
correspond to non-crossing boundary edges.
\end{proof}

\begin{theorem}
\label{th:clust_adj_full_amplitude}
The one-loop amplitude $A^{(1)}_{n,k}(Z,I)$ satisfies cluster adjacency in
$\Fl_{2,4;n}$.  Equivalently, every rational word $a\otimes b$ appearing in the
expanded symbol has the property that $a$ and $b$ are compatible cluster
variables of $\Fl_{2,4;n}$.
\end{theorem}
This statement is independent of the position chosen for the
two consecutive columns spanning the infinity twistor $I$ in the auxiliary
Grassmannian embedding into $\Gr(4,n+2)$.

We prove the theorem by separating the four-bracket part from the
$I$-dependent part.  The latter is controlled by the MHV amplitude.

\begin{proposition}
\label{prop:MHV_adj}
The one-loop MHV amplitude $A^{(1)}_{n,\mathrm{MHV}}(Z,I)$ satisfies
cluster adjacency in $\Fl_{2,4;n}$.
\end{proposition}

\begin{proof}
The one-loop MHV amplitude is built from one-mass and two-mass-easy box functions and can be written as
$A^{(1)}_{n,\mathrm{MHV}}=A^{(0)}_{n,\mathrm{MHV}}\,V^{(1)}_{n,\,\mathrm{MHV}}$, where the symbol of $V^{(1)}_{n,\,\mathrm{MHV}}$ is given in  \cref{eq:amplmhvVn}. Replacing with momentum twistors we find
\begin{equation}
\label{eq:symbol-Vf-flag}
\begin{split}
\mathcal{S}\left[V^{(1)}_{n,\,\mathrm{MHV}}\right]
=&\frac{1}{2}
\sum_{i=1}^{n}
\sum_{j=i+3}^{n+i-3}
\frac{
\langle i{-}1\,i\,j\,j{+}1\rangle
\langle i\,i{+}1\,j{-}1\,j\rangle}
{
\langle i\,i{+}1\,j\,j{+}1\rangle
\langle i{-}1\,i\,j{-}1\,j\rangle}
\\
&\hspace{7em}\otimes
\frac{\langle i\,\bar j\rangle
      \langle j\,\bar i\rangle}
     {\langle i{-}1\,i\rangle_I
      \langle i\,i{+}1\rangle_I
      \langle j{-}1\,j\rangle_I
      \langle j\,j{+}1\rangle_I}
\\
&+
\sum_{i=1}^n
\frac{
\langle i\,i{+}1\rangle_I
\langle i{+}1\,i{+}2\rangle_I
\langle i{-}1\,i\,i{+}2\,i{+}3\rangle}
{
\langle i\,i{+}1\,i{+}2\,i{+}3\rangle
\langle i{-}1\,i\,i{+}1\,i{+}2\rangle}
\\
&\hspace{7em}\otimes
\frac{\langle i\,\overline{i{+}2}\rangle
      \langle \bar i\,i{+}2\rangle}
     {\langle i\,i{+}1\rangle_I
      \langle i{+}2\,i{+}3\rangle_I
      \langle i{-}1\,i\rangle_I
      \langle i{+}1\,i{+}2\rangle_I}
\\
&+
\sum_{i=1}^n
\frac{
\langle i{-}1\,i\,i{+}1\,i{+}2\rangle}
{\langle i{-}1\,i\rangle_I
 \langle i{+}1\,i{+}2\rangle_I}
\otimes
\frac{
\langle i{-}1\,i\,i{+}1\,i{+}2\rangle}
{\langle i{-}1\,i\rangle_I
 \langle i{+}1\,i{+}2\rangle_I}.
\end{split}
\end{equation}

Here $\bar i$ denotes the triple $(i{-}1,i,i{+}1)$, so for instance
$\langle j\,\bar i\rangle=\langle j\,i{-}1\,i\,i{+}1\rangle$. After  multiplicative expansion, we get identical pairs, pairs of consecutive
two-brackets, and pairs consisting of
$\langle i{-}1\,i\,i{+}1\,i{+}2\rangle$ together with one of
$\langle i{-}1\,i\rangle_I$ or $\langle i{+}1\,i{+}2\rangle_I$.  These
are all compatible by \cref{lem:basic-compatibility}.
We also possibly get pairs
of the form $\langle a b c d\rangle\otimes
\langle r\,r{+}1\rangle_I$, where the consecutive pair $\{r,r{+}1\}$ is
not contained in $\{a,b,c,d\}$.  These non-manifest pairs are artifacts
of the chosen representation of the amplitude.  In the full MHV symbol the
potentially non-adjacent contributions cancel.  More precisely, the only
surviving words with a four-bracket in the first entry and a consecutive
two-bracket in the second entry are
\begin{equation}
  \mathcal S\left[V^{(1)}_{n,\,\mathrm{MHV}}\right]\Big|_{(4B)\otimes
  \langle i\,i{+}1\rangle_I}
  =
  2\,
  \langle i{-}1\,i\,i{+}1\,i{+}2\rangle
  \otimes
  \langle i\,i{+}1\rangle_I .
  \label{eq:MHV-only-surviving-4B-2B}
\end{equation}
These words are compatible by \cref{lem:basic-compatibility}, since
$\{i,i{+}1\}\subset\{i{-}1,i,i{+}1,i{+}2\}$.  Therefore every word in
the expanded MHV symbol is cluster adjacent in $\Fl_{2,4;n}$.
\end{proof}

\begin{proposition}
\label{prop:exp_4B_adj}
The four-bracket part
$\operatorname{Exp}_{4B}\mathcal S[A^{(1)}_{n,k}]$ satisfies cluster
adjacency in $\Gr(4,n)$, and hence also in $\Fl_{2,4;n}$.
\end{proposition}

\begin{proof}
We use the scalar-box representation of the one-loop amplitude, \cref{eq:boxdecomp}, with its symbol as described in \cref{sec:symbols-boxes}.  Since
we are considering the four-bracket part, all $x_{ij}^2$ are replaced by
their momentum-twistor numerators, namely
$x_{ij}^2\mapsto\langle i{-}1\,i\,j{-}1\,j\rangle$.  Thus compatibility
is checked in the Grassmannian cluster algebra $\Gr(4,n)$.

\smallskip
\noindent
\emph{One-mass boxes.}
The relevant second entries are $x^2_{i,i+2}$ or $x^2_{i+1,i+3}$,  
becoming $\langle i{-}1\,i\,i{+}1\,i{+}2\rangle$ and
$\langle i\,i{+}1\,i{+}2\,i{+}3\rangle$, which are frozen Pl\"ucker
coordinates.  Hence they are compatible with all cluster variables.

\smallskip
\noindent
\emph{Two-mass-easy boxes.}
The nontrivial second entry is
$x^2_{i,j+1}x^2_{i+1,j}
  -
  x^2_{i,j}x^2_{i+1,j+1}$,
which becomes, up to frozen factors, the product
$\langle i{-}1\,i\,i{+}1\,j\rangle
 \langle i\,j{-}1\,j\,j{+}1\rangle$.  The possible first entries are
$\langle i{-}1\,i\,j\,j{+}1\rangle$,
$\langle i\,i{+}1\,j{-}1\,j\rangle$,
$\langle i{-}1\,i\,j{-}1\,j\rangle$, and
$\langle i\,i{+}1\,j\,j{+}1\rangle$.  Each of these shares three indices
with each of the two four-brackets in the second entry.  Therefore all
required pairs are compatible by \cref{lem:basic-compatibility}.

\smallskip
\noindent
\emph{Two-mass-hard boxes.}
The non-frozen second entry is $x^2_{i+1,k}$, which becomes
$\langle i\,i{+}1\,k{-}1\,k\rangle$.  The relevant first entries are
$\langle i{-}1\,i\,k{-}1\,k\rangle$,
$\langle i{+}1\,i{+}2\,k{-}1\,k\rangle$, and
$\langle i\,i{+}1\,k{-}1\,k\rangle$, together with the frozen coordinate
$\langle i{-}1\,i\,i{+}1\,i{+}2\rangle$.  The first two share three
indices with the non-frozen second entry, the third is identical to it,
and the last is frozen.  Hence all pairs are compatible.

\smallskip
\noindent
\emph{Three-mass boxes.}
The nontrivial second entry is the quadratic cluster variable
$\langle i\,j{-}1\,j \mid k{-}1\,k
\mid i{-}1\,i\,i{+}1\rangle$.  The possible first entries are
$\langle i\,i{+}1\,j{-}1\,j\rangle$,
$\langle i{-}1\,i\,k{-}1\,k\rangle$,
$\langle i\,i{+}1\,k{-}1\,k\rangle$, and
$\langle i{-}1\,i\,j{-}1\,j\rangle$.

Compatibility is local in the cyclic order, so it suffices to check the
minimal configuration in $\Gr(4,8)$, obtained by relabelling
$i{-}1,i,i{+}1,j{-}1,j,j{+}1,k{-}1,k=1,2,3,4,5,6,7,8$.  Then the
statement becomes that
$\langle 2345\rangle$, $\langle 1278\rangle$, $\langle 2378\rangle$,
and $\langle 1245\rangle$ are compatible with the quadratic cluster
variable $\langle 245 \mid 78 \mid 123\rangle$.  This is verified in the
$\Gr(4,8)$ cluster algebra; equivalently, these variables appear
together in a common cluster.

It remains to consider the non-leading-singularity terms.  In all cases
the relevant words are of the form $x^2_{ij}\otimes x^2_{i{l}}$.  After
passing to momentum twistors, these become
$\langle i{-}1\,i\,j{-}1\,j\rangle\otimes
\langle i{-}1\,i\,{l}{-}1\,{l}\rangle$.  The two four-brackets share
the pair $\{i{-}1,i\}$ and are weakly separated.  Hence they are
compatible by \cref{lem:basic-compatibility}.

This proves cluster adjacency for
$\operatorname{Exp}_{4B}\mathcal S[A^{(1)}_{n,k}]$.
\end{proof}

\begin{proof}[Proof of \cref{th:clust_adj_full_amplitude}]
By \cref{prop:exp_4B_adj}, the four-bracket part
$\operatorname{Exp}_{4B}\mathcal S[A^{(1)}_{n,k}]$ satisfies cluster
adjacency in $\Gr(4,n)$, hence also in $\Fl_{2,4;n}$.

By \cref{prop:MHV_adj}, the MHV symbol
$\mathcal S[A^{(1)}_{n,\mathrm{MHV}}]$ satisfies cluster adjacency in
$\Fl_{2,4;n}$.  Therefore its $I$-dependent part
$\operatorname{Exp}_{I}\mathcal S[A^{(1)}_{n,\mathrm{MHV}}]$ also
satisfies cluster adjacency.  Using \cref{eq:exp-I-from-MHV}, the
$I$-dependent part of the full one-loop amplitude is
$\operatorname{Exp}_{I}\mathcal S[A^{(1)}_{n,k}]
=
R^{(0)}_{n,k}\,
\operatorname{Exp}_{I}\mathcal S[A^{(1)}_{n,\mathrm{MHV}}]$, and hence is
cluster adjacent as well.

Combining the four-bracket part with the $I$-dependent part, we conclude that the full expanded symbol of
$A^{(1)}_{n,k}(Z,I)$ satisfies cluster adjacency in $\Fl_{2,4;n}$.

Finally, the compatibility checks above do not depend on where the
auxiliary consecutive columns $i,i+1$ spanning $I$ are inserted in the cyclic order.
The result is therefore independent of the chosen Grassmannian embedding
into $\Gr(4,n+2)$.
\end{proof}

\subsection{Cluster adjacency for algebraic letters}

The symbol of the $4$-mass box in momentum-twistor variables reads
\cite{Prlina:2017azl,He:2020vob}:
\begin{equation}\label{eq:4mbtw}
\mathcal{S}_2[F^{\rm 4m}_{ijkl}]
=
\frac{\langle i{-}1\,i\,j{-}1\,j\rangle
\langle k{-}1\,k\,l{-}1\,l\rangle}
{\langle i{-}1\,i\,k{-}1\,k\rangle
\langle j{-}1\,j\,l{-}1\,l\rangle}
\otimes
\frac{Q_{ijkl}-\sqrt{\Delta_{ijkl}}}
{Q_{ijkl}+\sqrt{\Delta_{ijkl}}}
+(ijkl\rightarrow jkli),
\end{equation}
where
\begin{align}
Q_{ijkl}
={}&
\langle i{-}1\,i\,k{-}1\,k\rangle
\langle j{-}1\,j\,l{-}1\,l\rangle
+
\langle j{-}1\,j\,k{-}1\,k\rangle
\langle i{-}1\,i\,l{-}1\,l\rangle
\nonumber\\
&-
\langle i{-}1\,i\,j{-}1\,j\rangle
\langle k{-}1\,k\,l{-}1\,l\rangle ,
\end{align}
and
\begin{align}
\Delta_{ijkl}
:={}&
\Big(
\langle i{-}1\,i\rangle_I
\langle j{-}1\,j\rangle_I
\langle k{-}1\,k\rangle_I
\langle l{-}1\,l\rangle_I
\Big)^2
\rho_{ijkl}
\nonumber\\
={}&
Q_{ijkl}^2
-
4\,
\langle i{-}1\,i\,k{-}1\,k\rangle
\langle j{-}1\,j\,l{-}1\,l\rangle
\langle j{-}1\,j\,k{-}1\,k\rangle
\langle i{-}1\,i\,l{-}1\,l\rangle .
\label{eq:Delta4m}
\end{align}
Here $\rho_{ijkl}$ is the dual-coordinate discriminant defined in
\cref{eq:jacobrho}, whereas $\Delta_{ijkl}$ is its homogeneous
momentum-twistor counterpart.

Since these algebraic letters contain square roots and are not
polynomials in the Pl\"ucker coordinates, they are not ordinary cluster
$\mathcal A$-coordinates of $\operatorname{Gr}(4,n)$
\cite{ArkaniHamedLamSpradlin2021,HeLiYang2022Constraints}.
Consequently, the standard notion of compatibility between cluster
variables does not directly apply to them.

The Sklyanin Poisson structure on the Grassmannian
\cite{Sklyanin1982,GekhtmanShapiroVainshtein2003} was proposed in
\cite{Golden:2019kks} as an efficient test for cluster
compatibility. For two cluster $\mathcal A$-coordinates $a_1,a_2$
belonging to a common cluster, one has
\begin{equation}
\{\log a_1,\log a_2\}_{\mathfrak S}
\in \frac12\mathbb Z .
\end{equation}
Conversely, half-integrality of the bracket is conjectured to
characterize compatibility. This converse has been verified by
enumeration in finite-type cases but remains conjectural for the
infinite-type cluster algebras $\operatorname{Gr}(4,n)$ with $n\geq8$
\cite{Golden:2019kks}. We say that a pair satisfying
the half-integrality condition passes the Sklyanin-bracket test.

Although the algebraic letters considered here are not cluster
$\mathcal A$-coordinates, the Sklyanin bracket extends to the quadratic
extension of the function field obtained by adjoining
$\sqrt{\Delta_{ijkl}}$. We therefore use the same half-integrality
condition as an operational extension of the Sklyanin-bracket test.
Explicitly, we find
\begin{equation}
\left\{
\log\!\left(
\langle i{-}1\,i\,j{-}1\,j\rangle
\langle k{-}1\,k\,l{-}1\,l\rangle
\right),
\log\!\left(Q_{ijkl}\pm\sqrt{\Delta_{ijkl}}\right)
\right\}_{\mathfrak S}
\in \frac{1}{2}\mathbb Z.
\end{equation}

Together with the corresponding relations obtained by
$(ijkl)\rightarrow(jkli)$, bilinearity of the bracket on logarithms
implies that every symbol word of the $4$-mass box passes this
generalized Sklyanin-bracket test. We propose this as an operational
notion of ``adjacency'' for these algebraic letters.

We note that $Q_{ijkl}\pm\sqrt{\Delta_{ijkl}}$ passes the test only with
specific products of pairs of brackets. For example, for $n=9$,
$Q_{1358}\pm\sqrt{\Delta_{1358}}$ passes the test with products
\begin{equation}
\langle a-1,a,b-1,b\rangle
\langle c-1,c,d-1,d\rangle ,
\end{equation}
where, denoting such a product by $\{ab,cd\}$, the possible pairs are
\begin{equation}
\{24,58\},\quad
\{13,58\},\quad
\{29,58\},\quad
\{15,38\},\quad
\{18,58\},\quad
\{35,58\}.
\end{equation}
In particular, $\{13,58\}$ and $\{15,38\}$ are the two products needed
for the adjacency of the words appearing in the $4$-mass box.
We also remark that, in general,
$Q_{ijkl}\pm\sqrt{\Delta_{ijkl}}$ does not pass the test with an
individual four-bracket. This suggests that the natural objects entering
this generalized notion of adjacency are products of brackets, or
dual-conformal-invariant cross-ratios of such products, such as the
first entries in \cref{eq:4mbtw}.

We have also checked this generalized criterion against a sample of
rational--algebraic consecutive pairs in the known two-loop NMHV
symbols at eight and nine points
\cite{Zhang:2019vnm,He:2020vob}, finding no violations. This provides
preliminary evidence that the proposed notion may persist beyond one
loop.

\subsection{Cluster adjacency for coefficients}
\label{subsec:LL-adjacency}

We now refine the cluster-adjacency statement by including the rational
coefficients of the symbol.  
Recall that the \emph{LL-cluster-adjacency} conjecture of
\cite{Gurdogan:2020tip} predicts that, for each branch of a maximal
cut, the cluster-coordinate factors of its Landau singularity can be
found in a common cluster with all poles of any Yangian invariant that
can appear in a representation of the corresponding leading
singularity.
 In the language of the present paper,
this gives a natural source of adjacency between the last entries of
LS-terms and the poles of their leading-singularity coefficients.

Here we consider a stronger property.  Let $Y$ be a Yangian
invariant, and let $\operatorname{Pol}(Y)$ denote the set of irreducible
cluster variables appearing as poles of $Y$.  We say that a term
$Y\,a\otimes b$ is \emph{totally coefficient-adjacent} if
$\{a,b\}\cup\operatorname{Pol}(Y)$ is contained in a common cluster.

This definition is representation-dependent: if a coefficient is written
as a linear combination $d=\sum_\nu \lambda_\nu Y_\nu$ of Yangian
invariants, we ask for the property term by term, i.e. for each summand
$Y_\nu\,a\otimes b$.  This is the appropriate notion for us, since
Yangian invariants satisfy linear identities and hence the set of poles
of a coefficient may depend on the chosen representation. Related adjacency relations between final symbol entries and the poles
of their Yangian-invariant coefficients were studied in
\cite{DrummondFosterGurdogan2019,
Mago:2020eua}.

We do not expect the full amplitude $A^{(1)}_{n,k}(Z,I)$ to admit a
representation satisfying this property for all terms.  Indeed, the
$I$-dependent part contains last entries such as
$\langle i{-}1\,i\rangle_I$, whose coefficients are proportional to the
tree amplitude.  Already at NMHV level, the tree amplitude cannot be
written as a sum of $R$-invariants whose poles are all compatible with a
fixed two-bracket $\langle i{-}1\,i\rangle_I$.  For this reason, in this
subsection we restrict to the four-bracket part
$\operatorname{Exp}_{4B}\mathcal S[A^{(1)}_{n,k}]$.

We prove the statement at NMHV
level, where all coefficients can be expressed explicitly in terms of
$R$-invariants \cite{Britto:2004nc,
Drummond:2008bq}, and where LL-adjacency was proven in \cite{Gurdogan:2020tip}.
\begin{theorem}
\label{th:NMHV-total-coefficient-adjacency}
The four-bracket expansion of the one-loop NMHV amplitude admits a
representation of the form
\begin{equation}
  \operatorname{Exp}_{4B}\mathcal S[A^{(1)}_{n,1}]=
  \sum_{\alpha,\nu}
  \lambda_{\alpha,\nu}\,
  R_{\alpha,\nu}\,
  a_\alpha\otimes b_\alpha,
\end{equation}
where each $R_{\alpha,\nu}$ is an NMHV $R$-invariant, and every term is
totally coefficient-adjacent.  Equivalently, for every summand
$R_{\alpha,\nu}\,a_\alpha\otimes b_\alpha$, the collection
$\{a_\alpha,b_\alpha\}\cup\operatorname{Pol}(R_{\alpha,\nu})$ is
contained in a common cluster of $\Gr(4,n)$.
\end{theorem}

\begin{proof}
We split the proof into LS and non-LS terms.

\smallskip
\noindent
\textbf{LS terms.}
For the LS part, last-entry coefficient adjacency follows from
LL-adjacency.  It remains to check compatibility with the first entry.
We do this using the scalar-box representation of the one-loop amplitude.

Recall that the poles of an NMHV $R$-invariant
$[a,b,c,d,e]$ are the five Pl\"ucker coordinates obtained by deleting one
of its five indices:
\begin{equation}
  \operatorname{Pol}([a,b,c,d,e])
  =
  \{
  \langle b c d e\rangle,
  \langle a c d e\rangle,
  \langle a b d e\rangle,
  \langle a b c e\rangle,
  \langle a b c d\rangle
  \}.
\end{equation}
We will repeatedly use weak separation to check compatibility of these
poles with the first entries of the corresponding symbol words.

\smallskip
\noindent
\emph{One-mass boxes.}
The relevant last entries are $x^2_{i,i+2}$ and $x^2_{i+1,i+3}$, which
become the frozen variables
$\langle i{-}1\,i\,i{+}1\,i{+}2\rangle$ and
$\langle i\,i{+}1\,i{+}2\,i{+}3\rangle$.  Hence compatibility with the
last entry is automatic.  The only non-frozen first entry is
$x^2_{i,i+3}$, which becomes
$\langle i{-}1\,i\,i{+}2\,i{+}3\rangle$.  The corresponding coefficient
is $c^{\rm 1m}_{i,i+1,i+2,i+3}=A^{(0)}_{\mathrm{NMHV}}([i{+}2,i])$.  The poles of the
$R$-invariants appearing in this amplitude are Pl\"ucker coordinates
supported on the cyclic interval $[i{+}2,i]$, and these are weakly
separated from $\langle i{-}1\,i\,i{+}2\,i{+}3\rangle$.  Thus the
one-mass LS terms are totally coefficient-adjacent.

\smallskip
\noindent
\emph{Two-mass-easy boxes.}
The nontrivial last entry is
$x^2_{i,j+1}x^2_{i+1,j}-x^2_{i,j}x^2_{i+1,j+1}$, which gives the product
$\langle i{-}1\,i\,i{+}1\,j\rangle
 \langle i\,j{-}1\,j\,j{+}1\rangle$ in momentum twistors.  The possible
first entries are
$\langle i{-}1\,i\,j\,j{+}1\rangle$,
$\langle i\,i{+}1\,j{-}1\,j\rangle$,
$\langle i{-}1\,i\,j{-}1\,j\rangle$, and
$\langle i\,i{+}1\,j\,j{+}1\rangle$.  The coefficient is
$c^{\rm 2me}_{i,i+1,j,j+1}=A^{(0)}_{\mathrm{NMHV}}([i,j])+A^{(0)}_{\mathrm{NMHV}}([j,i])$.  Every pole of an
$R$-invariant appearing in $A^{(0)}_{\mathrm{NMHV}}([i,j])$ or $A^{(0)}_{\mathrm{NMHV}}([j,i])$ is supported on one
of the two cyclic intervals $[i,j]$ or $[j,i]$.  By weak separation, such
poles are compatible with all the first and last entries listed above.
Therefore the two-mass-easy LS terms are totally coefficient-adjacent.

\smallskip
\noindent
\emph{Two-mass-hard boxes.}
The relevant last entries are $x^2_{i+1,k}$ and $x^2_{i,i+2}$,
becoming $\langle i\,i{+}1\,k{-}1\,k\rangle$ and the frozen variable
$\langle i{-}1\,i\,i{+}1\,i{+}2\rangle$.  The nontrivial first entries
are
$\langle i{-}1\,i\,k{-}1\,k\rangle$,
$\langle i{+}1\,i{+}2\,k{-}1\,k\rangle$, and
$\langle i\,i{+}1\,k{-}1\,k\rangle$.  The coefficient is
\begin{equation}
  c^{\rm 2mh}_{i,i+1,i+2,k}
  =
  [i{-}1,i,i{+}1,k{-}1,k]
  +
  [i,i{+}1,i{+}2,k{-}1,k].
\end{equation}
The poles of these two $R$-invariants are Pl\"ucker coordinates supported
on the five-element sets
$\{i{-}1,i,i{+}1,k{-}1,k\}$ and
$\{i,i{+}1,i{+}2,k{-}1,k\}$.  Each such pole is compatible with the
nontrivial first and last entries above.  This follows from weak
separation; for example, any pole sharing $k{-}1,k$ with
$\langle i\,i{+}1\,k{-}1\,k\rangle$ differs from it by replacing the
adjacent pair $\{i,i{+}1\}$, hence is weakly separated from it.  The
remaining poles are handled by the same adjacent-pair criterion.

\smallskip
\noindent
\emph{Three-mass boxes.}
The nontrivial last entry is the quadratic cluster variable
$\langle i\,j{-}1\,j \mid k{-}1\,k \mid i{-}1\,i\,i{+}1\rangle$.  The
possible first entries are
$\langle i\,i{+}1\,j{-}1\,j\rangle$,
$\langle i{-}1\,i\,k{-}1\,k\rangle$,
$\langle i\,i{+}1\,k{-}1\,k\rangle$, and
$\langle i{-}1\,i\,j{-}1\,j\rangle$.  The coefficient is
$c^{\rm 3m}_{i,i+1,j,k}=[i,j{-}1,j,k{-}1,k]$.

The Pl\"ucker poles of this $R$-invariant are all compatible with the
first entries by weak separation.  For compatibility with the quadratic
last entry, it is enough to check the minimal configuration in
$\Gr(4,8)$, obtained by relabelling
$i{-}1,i,i{+}1,j{-}1,j,j{+}1,k{-}1,k$ as $1,2,3,4,5,6,7,8$.  The needed
finite check is that
$\langle 2457\rangle$, $\langle 4578\rangle$,
$\langle 2578\rangle$, $\langle 2478\rangle$, and
$\langle 2458\rangle$ are compatible with
$\langle 245 \mid 78 \mid 123\rangle$ in the $\Gr(4,8)$ cluster algebra.
This proves total coefficient adjacency for the three-mass LS terms.

\smallskip
\noindent
\textbf{Non-LS terms.}
We now turn to the non-LS part.  At NMHV level the only non-LS words are
of the form $x^2_{ij}\otimes x^2_{i{l}}$.  In momentum twistors this is
$
  \langle i{-}1\,i\,j{-}1\,j\rangle
  \otimes
  \langle i{-}1\,i\,{l}{-}1\,{l}\rangle $.
  
Let $I_{ij{l}}=\{i{-}1,i,j{-}1,j,{l}{-}1,{l}\}$.  For
$s\in I_{ij{l}}$, write $R_s$ for the $R$-invariant whose five labels
are $I_{ij{l}}\setminus\{s\}$. By \cref{prop:3distinct},
 the three pairs
$\{i{-}1,i\}$, $\{j{-}1,j\}$, and $\{{l}{-}1,{l}\}$ do not overlap in the non-LS part. 
We first consider the generic case, where no two of the  pairs are adjacent.  The
coefficient of $x^2_{ij}\otimes x^2_{i{l}}$ is
\begin{align}
 \mathrm{Coeff}[x^2_{ij} \otimes x^2_{il}]
  &=
  \frac{1}{2}\Big(
  -c^{\rm 3m}_{i-1,i,j,l}
  +c^{\rm 3m}_{i,i+1,j,l}
  +c^{\rm 3m}_{i,j-1,j,l}
  -c^{\rm 3m}_{i,j,j+1,l}
  -c^{\rm 3m}_{i,j,l-1,l}
  +c^{\rm 3m}_{i,j,l,l+1}
  \Big) \notag \\
  &=
  \frac{1}{2}
  \left(
  -R_i+R_{i-1}+R_j-R_{j-1}-R_{l}+R_{{l}-1}
  \right)
  =
  R_j-R_{j-1}.
  \label{eq:generic-nonLS-coeff}
\end{align}
In the last step we used the six-term identity among the six
$R$-invariants supported on $I_{ij{l}}$.

The important point is that both $R_j$ and $R_{j-1}$ contain the last
entry $\langle i{-}1\,i\,l{-}1\,l\rangle$ as one of their poles.
Thus last-entry coefficient adjacency is immediate.  Moreover, the poles
of $R_j$ and $R_{j-1}$ are also compatible with the first entry
$\langle i{-}1\,i\,j{-}1\,j\rangle$ by weak separation.  For instance,
the pole $\langle i{-}1\,j{-}1\,l{-}1\,l\rangle$ of $R_j$ is
compatible with $\langle i{-}1\,i\,j{-}1\,j\rangle$, since the symmetric
difference consists of the adjacent pair $\{l{-}1,l\}$ and the
pair $\{i,j\}$.

It remains to check the boundary cases, where one of the pairs becomes
adjacent to another.  These are obtained from the generic expression by
replacing the corresponding pair of three-mass coefficients with a two-mass-hard
coefficient.  The only replacements needed are
\begin{alignat}{3}
&i=j-2:  \quad && c^{\rm 3m}_{i,i+1,j,l}
  +c^{\rm 3m}_{i,j-1,j,l}
  \quad\to \quad  && c^{\rm 2mh}_{i,i+1,j,{l}}
   =
  R_j+R_{i-1},
\\
&j=l-2: \quad   && c^{\rm 3m}_{i,j,j+1,l}
  +c^{\rm 3m}_{i,j,l-1,l}  \quad\to \quad
&&   c^{\rm 2mh}_{i,j,j+1,l}
   =
  R_l+R_{j-1}, \text{ and}\\
&l=i-2: \quad 
&&c^{\rm 3m}_{i,j,l,l+1} + c^{\rm 3m}_{i-1,i,j,l}
 \quad\to  \quad &&   c^{\rm 2mh}_{i,j,l,l+1}
  =
  R_i+R_{{l}-1}. 
\end{alignat}
The sums in the first two lines correspond to  combinations in \cref{eq:generic-nonLS-coeff}, leaving its second line unchanged, but the third line has a discrepancy of signs.
All boundary cases thus reduce to one of the following two forms:
\[
\begin{array}{c|c|c}
\text{case} &\mathrm{Coeff}[x^2_{ij} \otimes x^2_{il}] & \mathrm{Coeff}^{\mathrm{LS}}[x^2_{ij} \otimes x^2_{il}] \\
\hline
{l}=i{-}2 
&
R_j-R_{j-1}+R_i
&
R_i+R_{{l}-1}
\\
\text{all other cases} 
&
R_j-R_{j-1}
&
0
\end{array}
\]
In the second line, the argument above for cluster adjacency applies directly: the coefficient
is $R_j-R_{j-1}$, and the last entry is a pole of both summands.

In the first line, the summands $R_j$ and $R_{j-1}$ are treated in the
same way.  The additional summand $R_i$ is precisely one of the
$R$-invariants already appearing in the LS contribution to the same
letter pair.  Hence its poles are compatible with both entries by the
LS analysis above.  Thus the non-LS terms are also totally
coefficient-adjacent.

Combining the LS and non-LS analyses proves the theorem.
\end{proof}

We conjecture that \cref{th:NMHV-total-coefficient-adjacency} generalizes for any helicity.

\begin{conjecture}[Total coefficient adjacency]
\label{conj:total-coefficient-adjacency}
The four-bracket expansion of the one-loop N$^k$MHV amplitude admits a
representation of the form 
\begin{equation}
\operatorname{Exp}_{4B}\mathcal S[A^{(1)}_{n,k}]
=
\sum_{\alpha,\nu}
\lambda_{\alpha,\nu}\,
Y^{(k)}_{\alpha,\nu}\,
a_\alpha\otimes b_\alpha ,
\end{equation}
where the $Y^{(k)}_{\alpha,\nu}$ are rational Yangian invariants and
$
\{a_\alpha,b_\alpha\}
\cup
\operatorname{Pol}(Y^{(k)}_{\alpha,\nu})$
is contained in a common cluster of $\Gr(4,n)$.
\end{conjecture}

\section{Discussion}
\label{sec:discussion}

In this work, we have studied the weight-two symbol of the full one-loop amplitude in planar
$\mathcal N=4$ SYM, at arbitrary multiplicity and helicity and in dimensional regularization.
Within the scalar-box representation, the result separates into a leading-singularity part
controlled by nested cuts, a distinct four-mass algebraic sector, and residual non-leading-
singularity terms. The LS decomposition isolates the part of the
symbol that is directly accounted for by maximal cuts and their leading
singularities, while the residual terms make precise what remains beyond
this description. The two-mass triangle relations explain the cancellation of the
difference-type letters and simplify part of this residual sector. Their branchwise refinement
follows directly from a residue theorem on the two components of the triple cut and suggests
a geometric interpretation as two dissections of the same projected cut region. We have
also established cluster adjacency for the rational symbol in $\mathrm{Fl}_{2,4;n}$, including
the infinity-twistor letters associated with infrared divergences. The four-mass letters pass
a natural Sklyanin-bracket test, while the four-bracket sector of the NMHV amplitude admits
a totally coefficient-adjacent representation. 
It would be interesting to check if this property extends to all helicity sectors.
Our results identify
a concrete interface between generalized unitarity, Landau analysis,
and cluster structure, and make precise several of the questions that
remain in passing from one description to the other.

The most immediate higher-loop extension is to two loops. The scalar-box expansion is then
replaced by a richer space of master integrals. Double-box and pentabox sectors provide concrete settings in
which to test whether multivariate residue theorems generate relations among leading
singularities analogous to the two-mass triangle relations, and whether these relations have
an interpretation as alternative dissections of projected loop-cut geometries. More generally,
one may ask whether nested two-loop cuts and their Landau discriminants control the ordering
of the four symbol entries, whether there is a useful higher-loop LS/non-LS decomposition,
and whether the full dimensionally regulated two-loop amplitude remains cluster adjacent in
the appropriate flag variety. The four-mass sector also suggests
seeking an extension of the cluster or Poisson framework to the more general algebraic structures encountered at higher loops.

Several ingredients of our analysis are not intrinsically tied to maximal supersymmetry.
Unitarity cuts, Landau geometry, dimensional regularization and the partial-flag varieties
describing massless kinematics all make sense in theories with less supersymmetry, including
pure Yang--Mills theory and QCD. What is special to planar $\mathcal N=4$ SYM is instead
the box-only expansion, Yangian invariance and the amplituhedron interpretation. In a generic
gauge theory, triangle and bubble integrals, non-uniform transcendental weight and rational
terms must also be included, with the latter being invisible to the symbol and to ordinary
four-dimensional cuts. Nevertheless, recent results show that partial-flag cluster algebras
organize substantial parts of QCD symbol alphabets and that the maximal-weight parts of
planar two-loop MHV QCD hard functions have leading-singularity prefactors computable
from on-shell diagrams
\cite{PokrakaSpradlinVolovichWeng:2025,CarroloEtAlPrescriptiveQCD}.
It would be interesting to determine which parts of our Landau and cluster picture survive
for the complete dimensionally regulated amplitude, away from the maximal-weight and
infrared-subtracted sectors.

Finally, the coaction may provide a natural framework in which to combine these questions.
The diagrammatic coaction organizes master integrands together with independent cut contours
and thereby encodes discontinuities and differential equations
\cite{AbreuEtAlMultipleCuts,Abreu:2021vhb}.
This raises the possibility that the nested-cut hierarchy underlying the LS part has a
coaction-compatible formulation.

\paragraph{AI-assisted technology.}
During the preparation of this manuscript, the authors used OpenAI ChatGPT and Anthropic Claude for language editing, organizational feedback, reference checking, and consistency checks. All scientific statements, calculations, and conclusions were independently verified by the authors, who take full responsibility for the content.

\acknowledgments

We thank Samuel Abreu, James Drummond, \"Omer G\"urdo\u{g}an, and Tristan McLoughlin for helpful discussions.
Computations were performed with the assistance of the packages {\tt PolyLogTools} \cite{Duhr:2019tlz} and {\tt loop\_amplitudes} \cite{Bourjaily:2013mma}. 
We acknowledge the support of the Munich Institute
for {Astro-,} Particle and BioPhysics (MIAPbP) which is
funded by the Deutsche Forschungsgemeinschaft (DFG,
German Research Foundation) under Germany’s Excellence Strategy—EXC-2094—390783311; the support of the Galileo Galilei Institute for Theoretical Physics and the INFN; and the Erwin Schrödinger International Institute for Mathematics and Physics (ESI), University of Vienna (Austria).
This work is funded by the European Union (ERC, MaScAmp, 101167287). Views and opinions expressed are however those of the author(s) only and do not necessarily reflect those of the European Union or the European Research Council Executive Agency. Neither the European Union nor the granting authority can be held responsible for them.

\appendix
\label{sec:app}
\section{Background for amplitudes}
\label{app:formulas}
\subsection{Kinematic variables}
\label{app:kinvar}
\paragraph{Spinor-helicity variables.}
Consider a particle with momentum $p^\mu$.
We can construct a $2 \times 2$ matrix $p_{a\dot b}$ by contraction with the Pauli matrices, 
\begin{equation}
p_{a\dot b} := p_\mu (\sigma^\mu)_{a\dot b},
\label{eq:pauli-dictionary}
\end{equation}
where $\sigma^\mu=(\mathbb{I},\sigma^i)$ and $\sigma^i$ for $i=1,2,3$ are the usual Pauli matrices.
If $p^2=0$, it follows that $\det(p_{a\dot b})=0$ and therefore that it is possible to write 
\begin{equation}
    p_{a\dot b} = \lambda_a \tilde\lambda_{\dot b},
\end{equation}
where $\lambda$ and $\tilde\lambda$ are two-component Weyl spinors of positive and negative chirality, respectively.

Spinor indices are raised and lowered with the Levi-Civita symbols $\epsilon^{ab}$ and $\epsilon^{\dot a \dot b}$ and their inverses, $\epsilon_{ab}$ and $\epsilon_{\dot a \dot b}$. Angle and square spinor brackets are used to denote the contractions
\begin{equation}\label{eq:spinorbracket}
\langle ij \rangle:=\epsilon_{ a b}\lambda_i^{ a}\lambda_j^{ b},
\qquad
\lbrack ij\rbrack:=\epsilon^{\dot a \dot b}\tilde\lambda_{i\,\dot a}\tilde\lambda_{j\, \dot b},
\end{equation}
where we have used the shorthand notation $\vert i\rangle^{a} = \lambda^{a}_i$ and $\vert i\rbrack^{\dot a} = \tilde\lambda^{\dot a}_i$ for a particle with momentum $p_i^\mu$. A Mandelstam invariant $s_{ij}=(p_i+p_j)^2$ is then equal to $\langle ij\rangle \lbrack ji\rbrack$.  We denote the contraction of momentum $P_{a \dot b}$ with  $\lambda_i,\tilde\lambda_j$ by
\begin{equation}
    \langle i \vert P \vert j\rbrack := \lambda^{a}_i P_{ a \dot b} \tilde\lambda_j^{\dot b}.
\end{equation}
\paragraph{Momentum twistors.}
The relation of a dual momentum $x_i$, with $p_i=x_i-x_{i+1}$ implies the incidence relation 
\begin{equation}
    \lbrack \mu_i\vert^{\dot a}=\langle i|_{ a} x_i^{a \dot a}=\langle i|_{a} x_{i+1}^{ a \dot a}.
\end{equation}
A momentum twistor is then defined as \cite{Hodges:2009hk}
\begin{equation}\label{eq:momtwist}
    Z_i^{J}:=\left( \vert i \rangle^{a}, \lbrack \mu_i\vert^{\dot a}\right),
\end{equation}
where $J=(a,\dot a)$ is an $\mathrm{SU}(2,2)$ index. We note that $Z_i^J$ is associated to two points $x_i^\mu$ and $x_{i+1}^\mu$ in the dual space. Thus, a point in $x$-space corresponds to a line in the momentum twistor $Z$-space. We can form a dual conformal invariant four-bracket as
\begin{equation}
    \langle i j k l\rangle:= \epsilon_{IJKL}Z^I_i Z^J_jZ^K_kZ^L_l.
\end{equation}
We also choose a line $I=\langle A,B\rangle\subset \mathbb P^3$, called the
\emph{infinity twistor}. It determines spinor-helicity two-brackets by
$\langle a b\rangle_I:=\langle A B a b\rangle$, since under \cref{eq:momtwist} this reproduces the spinor-bracket of \cref{eq:spinorbracket}. When no confusion can
arise, we suppress the subscript $I$ and write simply $\langle a b\rangle$
for $\langle a b\rangle_I$.
We can relate dual variables to twistors as
\begin{equation}\label{eq:dualtotwist}
    x_{ij}^2 = \frac{\langle i-1\,i\,j-1\,j\rangle}{\langle i-1,i\rangle_I\langle j-1,j\rangle_I}.
\end{equation}
For lines and planes we also use the intersection bracket
\begin{equation}\label{eq:aintersection}
\langle abc\vert de\vert fgh\rangle
:=\langle abcd\rangle\langle efgh\rangle
-\langle abce\rangle\langle dfgh\rangle,
\end{equation}
which represents the contraction of the intersection point $(de)\cap(abc)$ with the
plane $(fgh)$. 

The above construction extends to ${\mathcal N}=4$ SYM superspace. In addition to the dual coordinates, we introduce dual fermionic coordinates $\theta_{i}^\alpha$, with $\alpha$ an $\mathrm{SU}(4)$ R-symmetry label, as
\begin{equation}
    \vert \theta_{i}^{\alpha}\rangle - \vert\theta_{i+1}^{\alpha}\rangle = \vert i\rangle\eta_{i}^{\alpha}.
    \label{eq:theta-dual}
\end{equation}
We can then define the fermionic analogue of incidence relations
\begin{equation}
    \eta^\alpha_i=\langle i \theta_{i}^{\alpha}\rangle =\langle i \theta_{i+1}^{\alpha}\rangle, 
\end{equation}
allowing us to extend the $\mathrm{SU}(2,2)$ momentum twistors to $\mathrm{SU}(2,2|4)$ momentum super-twistors $\mathcal{Z}_i^{A}$
\begin{equation}
\mathcal{Z}_i^{A}:=\left( \vert i \rangle^{a}, \lbrack \mu_i\vert^{\dot a} , \eta_{i}^{\alpha}\right), \qquad \text{where } A=(a,\dot a,\alpha).
\end{equation}
The super-twistors transform linearly under the dual superconformal group $\mathrm{SU}(2,2|4)$, so quantities built out of them are manifestly dual superconformal invariant.
\subsection{One-loop MHV amplitudes}
\label{app:onelmhv}
The MHV amplitude is built from one-mass and two-mass-easy box functions only, and the nonvanishing coefficients are the same and equal to the tree MHV amplitude \cite{Bern:1994zx}, so that
\begin{align}
    A^{(1)}_{n,\,\mathrm{MHV}}=&\sum_{i=1}^n c^{\mathrm{1m}}_{i,i+1,i+2,i+3} F^{\rm 1m}_{i,i+1,i+2,i+3}+\frac{1}{2}\sum_{i=1}^n \sum_{j=i+3}^{i+n-3}c^{\mathrm{2me}}_{i,i+1,j,j+1} F^{\rm 2me}_{i,i+1,j,j+1}\\=& A^{(0)}_{n,\,\mathrm{MHV}}\left[\sum_{i=1}^n  F^{\rm 1m}_{i,i+1,i+2,i+3}+\frac{1}{2}\sum_{i=1}^n \sum_{j=i+3}^{i+n-3} F^{\rm 2me}_{i,i+1,j,j+1}\right].
\end{align}
Using \cref{eq:1mbox-symbol,eq:2mebox-symbol}, we find that the weight-two symbol of the transcendental function defined by $V^{(1)}_{n,\,\mathrm{MHV}}=A^{(1)}_{n,\,\mathrm{MHV}}/A^{(0)}_{n,\,\mathrm{MHV}}$ is equal to
\begin{align}
    \mathcal{S}\left[  V^{(1)}_{n,\,\mathrm{MHV}}\right]=& \frac{1}{2}\sum_{i=1}^n \sum_{j=i+3}^{i+n-3}\frac{x^2_{{i+1},j}x^2_{i,j+1}}{x^2_{ij} x^2_{{i+1},{j+1}}} \otimes
   \left(x^2_{i,j+1} x^2_{{i+1},j}-x^2_{ij} x^2_{{i+1},{j+1}}\right)\nonumber\\&+\sum_{i=1}^n \left[\frac{x^2_{i,{i+3}}}{x^2_{i,{i+2}}x^2_{{i+1},{i+3}}}\otimes x^2_{i,{i+2}}x^2_{{i+1},{i+3}}+x^2_{i,{i+2}}\otimes x^2_{i,{i+2}}\right]\label{eq:amplmhvVn}.
\end{align}
\subsection{One-loop NMHV amplitudes}
\label{app:onelnmhv}
The basic dual superconformal invariants at NMHV level are the R-invariants \cite{Drummond:2008vq}, which in terms of momentum super-twistors can be written as
\begin{equation}
    R_{nij}=\frac{\delta^{(4)}\left(\langle i-1,i,j-1,j\rangle \chi_n +\text{cyclic} \right)}{\langle n,i-1,i,j-1\rangle \langle i-1,i,j-1,j\rangle \langle i,j-1,j,n\rangle \langle j-1,j,n,i-1\rangle \langle j,n,i-1,i\rangle}.
\end{equation}
Given that the above expression is cyclic in $\left(n,i-1,i,j-1,j \right)$, we can introduce the 5-bracket notation
\begin{align}
    R_{nij}=\left[n,i-1,i,j-1,j \right],
\end{align}
which is totally antisymmetric in its five arguments.

In terms of the R-invariants, the NMHV tree-level superamplitude is \cite{Drummond:2008vq}
\begin{equation}
    A^{(0)}_{n,\,\mathrm{NMHV}}
    = A^{(0)}_{n,\,\mathrm{MHV}}\sum_{i=2}^{n-3} \sum_{j=i+2}^{n-1} R_{nij}.
    \label{eq:NMHVtree}
\end{equation}

At one loop, each box coefficient of the NMHV superamplitude is a linear combination of R-invariants multiplying the tree MHV prefactor. Explicit expressions for all coefficients, together with a representation of the ratio function in terms of dual conformal cross-ratios, are given in Sec.~5 of \cite{Elvang:2009ya}, whose results we use. Since we work with the unnormalized amplitude in dimensional regularization, we restore the IR-divergent part through the universal formula of \cref{eq:generalir}.
\section{Coefficients of words}
\label{app:coeff}
In the following table we give explicitly the coefficients of the words $x_{ij}^2 \otimes x_{il}^2$. The table assumes an ordering $i<j<l$. The $i<l<j$ ordering can be obtained by reflection, replacing every index $r$ with $ n-r$.
\begin{sidewaystable}
\centering
\renewcommand{\arraystretch}{1.4}
\resizebox{\textheight}{!}{%
\begin{tabular}{||c|c|c|c|c|} 
\hline
\begin{tabular}[c]{@{}c@{}} Adjacency conditions \end{tabular}
&$x^2_{ij} \otimes x^2_{il}$ & $\mathrm{Coeff}[x^2_{ij} \otimes x^2_{il}]$ & $\mathrm{Coeff}^{\mathrm{LS}}[x^2_{ij} \otimes x^2_{il}]$& $d_{ijl}$ \\
\hline
\hline
\begin{tabular}[c]{@{}c@{}}$|i-j|,|i-l|\geq 3$\\ $|j-l| \geq 3$\end{tabular} &$x^2_{ij} \otimes x^2_{il}$ &\begin{tabular}[c]{@{}c@{}}
$\half\left[-c^{\rm 3m}_{i-1,i,j,l} + c^{\rm 3m}_{i,i,j,l}
+ c^{\rm 3m}_{i,j-1,j,l} \right.$\\
$\left.- c^{\rm 3m}_{i,j,j,l}-c^{\rm 3m}_{i,j,l-1,l} + c^{\rm 3m}_{i,j,l,l+1}\right]$
\end{tabular}&-&$\mathrm{Coeff}[x^2_{ij} \otimes x^2_{il}]$\\
\hline
\begin{tabular}[c]{@{}c@{}}$|i-j|, |i-l|\geq 3,$\\ $|j-l| =2$\end{tabular} & $x^2_{ij} \otimes x^2_{i,j+2}$  & 
\begin{tabular}[c]{@{}c@{}}
$\half\left[-  c^{\rm 3m}_{i-1,i,j,j+2} +  c^{\rm 3m}_{i,i,j,j+2}
   +  c^{\rm 3m}_{i,j-1,j,j+2} \right. $\\
$ \left.+  c^{\rm 3m}_{i,j,j+2,j+3} -c^{\rm 2mh}_{i,j,j,j+2}\right] $
\end{tabular}&-&$\mathrm{Coeff}[x^2_{ij} \otimes x^2_{il}]$\\
\hline
\begin{tabular}[c]{@{}c@{}}$|i-j|, |j-l|\geq 3,$\\ $|i-l| =2$\end{tabular} &$x^2_{ij} \otimes x^2_{i,i-2}$  & 
\begin{tabular}[c]{@{}c@{}}
$\half\left[c^{\rm 3m}_{i-2,i,i,j}
   +  c^{\rm 3m}_{i-2,i,j-1,j} -  c^{\rm 3m}_{i-2,i,j,j}\right]
    $\\
$\left. - c^{\rm 3m}_{i-3,i-2,i,j} +  c^{\rm 2mh}_{i-2,i-1,i,j} \right]$
\end{tabular}&$ c^{\rm 2mh}_{i-2,i-1,i,j}$&\begin{tabular}[c]{@{}c@{}}
$\half\left[ c^{\rm 3m}_{i-2,i,i,j}
   +  c^{\rm 3m}_{i-2,i,j-1,j} -  c^{\rm 3m}_{i-2,i,j,j}\right]
    $\\
$\left. - c^{\rm 3m}_{i-3,i-2,i,j} -  c^{\rm 2mh}_{i-2,i-1,i,j} \right]$
\end{tabular}\\
\hline
\begin{tabular}[c]{@{}c@{}}$|i-l|, |j-l|\geq 3,$\\ $|i-j| =2$\end{tabular} &$x^2_{i,i+2} \otimes x^2_{i,l}$  & 
\begin{tabular}[c]{@{}c@{}}
$\half\left[-c^{\rm 3m}_{i-1,i,i+2,l} - c^{\rm 3m}_{i,i+2,i+3,l} -c^{\rm 3m}_{i,i+2,l-1,l}\right. $\\
$\left.+ c^{\rm 3m}_{i,i+2,l,l+1} + c^{\rm 2mh}_{i,i,i+2,l}\right] $
\end{tabular}&-&$\mathrm{Coeff}[x^2_{ij} \otimes x^2_{il}]$\\
\hline
\begin{tabular}[c]{@{}c@{}}$|i-j|= |i-l|=2,$\\ $|j-l| \geq 3$\end{tabular} &$x^2_{i,i + 2} \otimes x^2_{i,i - 2}$  & 
\begin{tabular}[c]{@{}c@{}}
$ \half\left[- c^{\rm 3m}_{i-2,i,i+2,i+3}  
   -  c^{\rm 3m}_{i-3,i-2,i,i+2}\right.$\\ $ \left.+  c^{\rm 2mh}_{i-2,i,i,i+2}
   + c^{\rm 2mh}_{i-2,i-1,i,i+2} \right] $
\end{tabular}&$  c^{\rm 2mh}_{i-2,i-1,i,i+2} $&\begin{tabular}[c]{@{}c@{}}
$ \half\left[- c^{\rm 3m}_{i-2,i,i+2,i+3}  
   -  c^{\rm 3m}_{i-3,i-2,i,i+2}\right.$\\ $ \left.+  c^{\rm 2mh}_{i-2,i,i,i+2}
   - c^{\rm 2mh}_{i-2,i-1,i,i+2} \right] $
\end{tabular}\\
\hline
\begin{tabular}[c]{@{}c@{}}$|i-j| =|j-l|= 2,$\\ $|i-l| \geq 3$\end{tabular} &$x^2_{i,i + 2} \otimes x^2_{i,i + 4}$  & 
\begin{tabular}[c]{@{}c@{}}
$ \half\left[  c^{\rm 3m}_{i,i+2,i+4,i+5} - c^{\rm 3m}_{i-1,i,i+2,i+4} \right.$\\
 $ \left.  + c^{\rm 2mh}_{i,i,i+2,i+4}
   - c^{\rm 2mh}_{i,i+2,i+3,i+4}  \right]$
\end{tabular}&-&$\mathrm{Coeff}[x^2_{ij} \otimes x^2_{il}]$\\
\hline
\begin{tabular}[c]{@{}c@{}}$|j-l| =|i-l|= 2,$\\ $|i-j| \geq 3$\end{tabular} &$x^2_{i,i - 4} \otimes x^2_{i,i - 2}$  & 
\begin{tabular}[c]{@{}c@{}}
$\half\left[     c^{\rm 3m}_{i-4,i-2,i,i} 
+  c^{\rm 3m}_{i-5,i-4,i-2,i}\right.$\\$
  \left.  
   - c^{\rm 2mh}_{i-4,i-3,i-2,i} +  c^{\rm 2mh}_{i-4,i-2,i-1,i}  \right] $
\end{tabular}
&$  c^{\rm 2mh}_{i-4,i-2,i-1,i}$
&\begin{tabular}[c]{@{}c@{}}
$\half\left[     c^{\rm 3m}_{i-4,i-2,i,i} 
+  c^{\rm 3m}_{i-5,i-4,i-2,i}\right.$\\$
  \left.  
   - c^{\rm 2mh}_{i-4,i-3,i-2,i} -  c^{\rm 2mh}_{i-4,i-2,i-1,i}  \right] $
\end{tabular}
\\
\hline
\begin{tabular}[c]{@{}c@{}}$|i-j| =|i-l|=$\\ $=|j-l|=2$
\end{tabular} &$x^2_{13} \otimes x^2_{15}$  & 
\begin{tabular}[c]{@{}c@{}}
$  \half\left[  c^{\rm 2mh}_{1235} - c^{\rm 2mh}_{1345} +  c^{\rm 2mh}_{1356} \right]  $
\end{tabular}
&$  c^{\rm 2mh}_{1356}$& \begin{tabular}[c]{@{}c@{}}
$  \half\left[ c^{\rm 2mh}_{1235} - c^{\rm 2mh}_{1345} -  c^{\rm 2mh}_{1356} \right]  $
\end{tabular}
\\
\hline
\end{tabular}%
}
\caption{}
\end{sidewaystable}
\newpage
\section{Special cases of Proposition \ref{prop:repeated}} 
\label{ref:app_special}
\paragraph{Proof of Proposition \ref{prop:repeated}}
For $n\geq7$ and $j=4$, we consider the word $x^2_{14} \otimes x^2_{14}$, since its coefficient is equivalent to that of $x^2_{1,n-2} \otimes x^2_{1,n-2}$, up to relabeling. We have
\begin{align}
\mathrm{Coeff}[x^2_{14} \otimes x^2_{14}]=& 
-2 c^{\rm 2mh}_{1,2,4,n}
-2 c^{\rm 2mh}_{1,3,4,5}
+\half c^{\rm 2mh}_{1,4,5,6}
+\half c^{\rm 2mh}_{1,4,n-1,n}
\nonumber \\ &
+
c^{\rm 1m}_{1,2,3,4} + c^{\rm 2me}_{1,4,5,n} 
- c^{\rm 2me}_{1,2,4,5} 
- c^{\rm 2me}_{1,3,4,n}
 \nonumber \\ &
+ \sum_{j=7}^{n-1} \half c^{\rm 3m}_{1,4,5,j}
+ \sum_{j=6}^{n-2} \half c^{\rm 3m}_{1,4,j,n}
- \sum_{j=6}^{n-1} \half c^{\rm 3m}_{1,3,4,j}
- \sum_{j=6}^{n-1} \half c^{\rm 3m}_{1,2,4,j}\\
=&\frac{1}{2}\left(t_{1,2,4}-t_{3,4,1}+t_{4,5,1}-t_{n,1,4}\right)- c^{\rm 2mh}_{1,2,4,n}
- c^{\rm 2mh}_{1,3,4,5}.
\end{align}
The coefficient is then equal to
\begin{equation}
    \mathrm{Coeff}[x^2_{14} \otimes x^2_{14}]=- c^{\rm 2mh}_{1,2,4,n}
- c^{\rm 2mh}_{1,3,4,5}=\mathrm{Coeff}^{\mathrm{LS}}[x^2_{14} \otimes x^2_{14}].
\end{equation}

For $n=6$, the coefficient of $w=x^2_{14} \otimes x^2_{14}$ is
\begin{align}
\mathrm{Coeff}[x^2_{14} \otimes x^2_{14}]=&-2 c^{\rm 2mh}_{1,2,4,6}
-2 c^{\rm 2mh}_{1,3,4,5}
+ c^{\rm 1m}_{1,2,3,4} 
+ c^{\rm 1m}_{1,4,5,6}
- c^{\rm 2me}_{1,2,4,5} 
- c^{\rm 2me}_{1,3,4,6}
\nonumber \\
& = \frac{1}{2}\left(t_{1,2,4}-t_{3,4,1}+t_{4,5,1}-t_{6,1,4}\right)-c^{\rm 2mh}_{1,2,4,6}
-c^{\rm 2mh}_{1,3,4,5},
\end{align}
which reduces to
\begin{equation}
    \mathrm{Coeff}[x^2_{14} \otimes x^2_{14}]=- c^{\rm 2mh}_{1,2,4,6}
- c^{\rm 2mh}_{1,3,4,5}=\mathrm{Coeff}^{\mathrm{LS}}[x^2_{14} \otimes x^2_{14}].
\end{equation}

Finally, for $n=6$ and $x^2_{13} \otimes x^2_{13}$, we have
\begin{align}
\mathrm{Coeff}[x^2_{13} \otimes x^2_{13}]=& 
- c^{\rm 1m}_{1,2,3,4}
- c^{\rm 1m}_{1,2,3,6}
- \half c^{\rm 2mh}_{1,2,3,5}
+\half c^{\rm 2mh}_{1,3,4,5}
+\half c^{\rm 2mh}_{1,3,5,6}
+ c^{\rm 2me}_{1,3,4,6} 
\\ &= \frac{1}{2}\left(t_{3,4,1}-t_{6,1,3}- c^{\rm 1m}_{1,2,3,6}- c^{\rm 1m}_{1,2,3,4}-c^{\rm 2mh}_{1,2,3,5}\right).
\end{align}
This reduces to
\begin{equation}
    \mathrm{Coeff}[x^2_{13} \otimes x^2_{13}]=\frac{1}{2}\left(- c^{\rm 1m}_{1,2,3,6}- c^{\rm 1m}_{1,2,3,4}-c^{\rm 2mh}_{1,2,3,5}\right)=\half\mathrm{Coeff}^{\mathrm{LS}}[x^2_{13} \otimes x^2_{13}].
\end{equation}
 \qed
\newpage
\section{Proof of Two-Mass Triangle Relations} \label{ref:app_2mt_relations}
We give here the full proof of \cref{prop:2mtrls}, which was sketched in \cref{sec:sketchproof}.
\paragraph{Parametrizing the triple cut of the two-mass triangle.} Our starting point is the triple cut of the two-mass triangle shown in \cref{fig:triplecut}. 
  We place the three chosen propagators on shell with delta functions, and integrate over the Grassmann variables associated with each cut propagator in order to carry out the state super-sum \cite{Bern:2009xq}, giving the integral
\begin{align}
{\mathcal C}_{n1i}=\int \prod_{r=1,i,n} \mathrm{d}^4 \eta_{\ell_r}\mathrm{d}^4\ell& \delta\left(\ell^2\right)\delta\left(\left(\ell+x_{1 i}\right)^2\right)\delta\left(\left(\ell-p_n\right)^2\right) A_{1,i}({\ell}_1,p_1,\dots,p_{i-1},-\ell_i)\nonumber\\&
A_{i,n}(\ell_i,p_{i},\dots,p_{n-1},-\ell_n)
A_{n,1}(\ell_n,p_n,-{\ell}_1).
\end{align}
The solutions to the three on-shell conditions can be parametrized in terms of spinors (see \cref{app:kinvar}) as follows.
We expand the loop momentum $\ell:=\ell_1$ in terms of a basis of null momenta whose coefficients are parameters $a,b,c,d$,
\begin{equation}
\ell^{a\dot{a}}_1=\ell^{a\dot{a}}=a\lambda_A^a \tilde{\lambda}_A^{\dot{a}}+b\lambda_B^a \tilde{\lambda}_B^{\dot{a}}+c\lambda_A^a \tilde{\lambda}_B^{\dot{a}}+d\lambda_B^a \tilde{\lambda}_A^{\dot{a}}.
\end{equation}
We further choose  
\begin{equation}
 \lambda_A = \lambda_n, \quad \tilde{\lambda}_A=\tilde{\lambda}_n, 
\quad \lambda_B = x_{1 i}\cdot\tilde{\lambda}_n, \quad \tilde\lambda_B = x_{1 i}\cdot{\lambda}_n.
\end{equation}
Notice that $b,c,d$ are not dimensionless.
The three cut conditions 
$0=\ell_1^2$, $0=\ell_i^2=(\ell_1+x_{1 i})^2$, $0=\ell_n^2=(\ell_1-p_n)^2$
respectively impose the constraints
\begin{equation}\label{eq:abvalues}
    ab+cd=0, \quad a=\frac{x_{1 i}^2}{x_{1 i}^2-x_{i n}^2}, \quad b=0,
\end{equation} 
where we have used the fact that $\langle n \vert x_{1 i} \vert n \rbrack = 2 p_n \cdot x_{1 i} = x_{i n}^2-x_{1 i}^2$.
The integration measure along with the delta functions can be written as
\begin{equation}
\int \mathrm{d}^4\ell \,\delta\left(\ell^2\right)\delta\left(\left(\ell+x_{1 i}\right)^2\right)\delta\left(\left(\ell-p_n\right)^2\right)
=J\int \delta \left(ab+cd\right)\delta(b)\delta\left(a+\frac{x_{1 i}^2}{x_{i n}^2-x_{1 i}^2}\right)\mathrm{d}a\,\mathrm{d}b\, \mathrm{d}c\,\mathrm{d}d,
\end{equation}
where $J$ is the Jacobian factor, given by 
\begin{equation} J
={\langle \lambda_A \lambda_B \rangle}^2 {\lbrack \tilde{\lambda}_A\tilde{\lambda}_B \rbrack}^2=(x_{i n}^2-x_{1 i}^2)^2.
\end{equation}
At this point, we see that $a$ and $b$ are localized by the last two delta functions, and the first can be rewritten as $\delta(cd)$, giving a discrete choice of $c=0$ or $d=0$.
 These two cases are related by reversing the chirality of the spinors, so we will focus on the case $c=0$. 
 
 Thus we have $\delta(cd)=\delta(c)/d$, and we are left with a one-dimensional integral for the triple cut. We finally rewrite our remaining parameter as 
 \begin{equation}
    d=z
\end{equation}
 to emphasize that it is a complex variable to which we will apply the residue theorem.
 
The loop momentum satisfying the triple-cut conditions is denoted by $\hat{\ell}(z)$ as a function of $z$, with
\begin{equation}
    \hat{\ell}^{a\dot{a}}(z)=a\lambda_A^a \tilde{\lambda}_A^{\dot{a}}+z\lambda_B^a \tilde{\lambda}_A^{\dot{a}}, 
\label{eq:ellofz}
\end{equation}
and the value of $a$ as in \cref{eq:abvalues}.

We have arrived at the following form of the triple cut integral:
\begin{align}\label{eq:triple}
{\mathcal C}_{n1i}
=J\int \frac{\mathrm{d}z}{z} \prod_{r=1,i,n} \mathrm{d}^4 \eta_{\hat{\ell}_r}A_{1,i}(z) A_{i,n}(z) A_{n,1}(z).
\end{align}

\paragraph{Residue theorem for the on-shell form.}
We are working in four dimensions, so there is still one degree of freedom left for the loop momentum $\hat\ell(z)$, now parametrized by the variable $z$. 
We will apply the residue theorem to the triple-cut amplitude as a function of $z$. The residues from poles within the amplitudes will be identified as quadruple cuts, and we additionally need to consider the behaviour at $z=0$ with the apparent pole in \cref{eq:triple}, and as $z \to \infty$.

\paragraph{Three-point MHV amplitude for \texorpdfstring{$c=0$}{}.}

\begin{align}
\vert \hat{\ell}_1 \rangle &=a\vert n\rangle+z \vert x_{1 i} \vert n \rbrack, \quad &\lbrack \hat{\ell}_1\vert&=\lbrack n \vert,\label{eq:mhvl1}\\
\vert \hat\ell_n \rangle&=(a-1)\vert n\rangle+z \vert x_{1 i} \vert n \rbrack,\quad &\lbrack \hat\ell_n\vert&=\lbrack n \vert. \label{eq:mhvl3}
\end{align}
 The above relations mean that the negative-chirality spinors appearing in the amplitude $A_{n,1}(p_n,\ell_1,\ell_n)$ are all equal, so that the three-point amplitude must be MHV. 

Similarly, in the case $d=0$, the three-point amplitude is constrained to be \textoverline{MHV}.
\paragraph{Behaviour at zero.}
In order to assess whether there is truly a singularity at $z=0$, we evaluate the three-point amplitude explicitly. For the case $c=0$, the three-point MHV amplitude can be written as 
\begin{align}
A_{n,1}^{\text{MHV}}(\hat\ell_n,p_{n},-\hat\ell_1)&=i\frac{\delta^{\left(8\right)}\left(\vert \hat\ell_n\rangle \eta_{\hat\ell_n}+\vert n\rangle \eta_{n}-\vert \hat\ell_1\rangle \eta_{\hat\ell_1}\right)}{{\langle \hat\ell_n\,n\rangle \langle n\,\hat\ell_1\rangle \langle \hat\ell_1\,\hat\ell_n\rangle}} 
\end{align}
Using the expressions of \cref{eq:mhvl1,eq:mhvl3}, we have
\begin{align}
A_{n,1}^{\text{MHV}}(\hat\ell_n,p_{n},-\hat\ell_1) &=-i\frac{\delta^{\left(8\right)}\left((a-1)\vert n\rangle \eta_{\hat\ell_n}+z\vert x_{1 i}|n\rbrack \eta_{\hat\ell_n}+\vert n\rangle \eta_{n}-a\vert n\rangle \eta_{\hat\ell_1}-z\vert x_{1 i}|n\rbrack \eta_{\hat\ell_1}\right)}{z^3
(x_{i n}^2-x_{1 i}^2)^3
} \\
&= -
i z (x_{i n}^2-x_{1 i}^2)
\delta^{(4)}\left( \eta_{\hat\ell_n}- \eta_{\hat\ell_1}\right)
\delta^{(4)}\left((a-1) \eta_{\hat\ell_n}+ \eta_{n}-a \eta_{\hat\ell_1}\right),
\label{eq:mhv3-grassmann}
\end{align}
where in the last line, the 8-fold delta function has been factorized by projecting the spinors through contraction with the basis of $\langle n \vert$ and $\lbrack n \vert x_{1 i}$.\footnote{See Eq.~(4.53) of \cite{bible}.}
Thus we find
\begin{align}
\eta_{\hat\ell_1}&=\frac{a-1}{a}\eta_{\hat\ell_n}+\frac{1}{a}\eta_{n}=\eta_{n},\label{eq:mhvetal1}\\
\eta_{\hat\ell_n}&=\eta_{n}\label{eq:mhvetal3},
\end{align}
and a visible factor of $z$ has emerged in \cref{eq:mhv3-grassmann} to cancel the potential singularity.
Performing the Grassmann integrations  over $\eta_{\hat\ell_1},\eta_{\hat\ell_n}$ gives
\begin{equation}\int d^4 \eta_{\hat\ell_1}d^4 \eta_{\hat\ell_n}\frac{A_{n,1}^{\text{MHV}}(\hat\ell_n,p_{n},-\hat\ell_1)}{z}
=-i (x_{i n}^2-x_{1 i}^2).\label{eq:3pmhv}
\end{equation}
Therefore, the above part of the integral has a finite contribution at $z=0$.
We will see below that the poles of $A_{1,i}$ and $A_{i,n}$ arise from their internal propagators and are distinct from $z=0$ for generic kinematics, so there is indeed no singularity at $z=0$.
\paragraph{Residue at infinity.}
Next, we show that the residue at infinity vanishes. Let us denote the
supercharge by
$Q^{\alpha\dot a}:=\sum_{i=1}^{n}\lbrack i\vert^{\dot a}\,
\frac{\partial}{\partial\eta_{i\alpha}}$,
with $\alpha$ the R-symmetry index in the fundamental representation of
$\mathrm{SU}(4)$ and $\dot a$ the spinor index. Under a finite $Q$
transformation with an anticommuting parameter $\zeta_{\alpha \dot a}$,
common to all legs, the on-shell coherent state of a leg with momentum
$p_i^{a\dot a}=\lambda_i^a\tilde\lambda_i^{\dot a}$ is translated as
\begin{equation}
e^{Q^{\alpha\dot a}\zeta_{\alpha\dot a}}\big\vert\eta_i\big\rangle
=\big\vert\eta_i+\lbrack i\zeta\rbrack\big\rangle,
\end{equation}
where $\lbrack i\zeta\rbrack=\tilde\lambda_{i}^{\dot a}\zeta_{\dot a}$ and R-symmetry indices are suppressed on $\eta_i$ and $\zeta$. This is a symmetry of the amplitude: the shift leaves the Grassmann delta
function invariant, since
\begin{equation}
\sum_{i=1}^{n}\lambda_i^{a}\lbrack i\zeta\rbrack
=\Big(\sum_{i=1}^{n}\lambda_i^{a}\tilde\lambda_i^{\dot a}\Big)\zeta_{\dot a}
=0
\end{equation}
by momentum conservation, and hence
$\delta^{(8)}\big(\sum_i\lambda_i\eta_i\big)$ is unchanged.
This means that if we label all external and cut legs by $\eta$ coherent
states,\footnote{Each leg may equivalently be labelled in the $\bar\eta$
basis, built on the negative-helicity gluon; the two are related by a
Grassmann Fourier transform. In a mixed basis, the action of $Q$ produces
exponential phases on the $\bar\eta$ legs, so \cref{eq:Qinv} would hold
only up to such factors.} built on the positive-helicity gluon, then the
$Q$ supersymmetries act purely as translations and the amplitude is
invariant,
\begin{equation}\label{eq:Qinv}
A(\eta_i)=A\big(\eta_i+\lbrack i\zeta\rbrack\big).
\end{equation}
The two spinor degrees of freedom of $\zeta_{\alpha\dot a}$ imply that we
can always choose its components appropriately to translate two $\eta$'s
to zero.

Consider now the $z$ dependence of the two massive corners. By \cref{eq:mhvl1,eq:mhvl3}, the legs adjacent to the three-point vertex have $\lbrack\hat\ell_1\vert=\lbrack\hat\ell_n\vert$, while at $A_{1,i}$ the $z$-dependence enters through $\vert\hat\ell_1\rangle$, and at $A_{i,n}$
through $\vert\hat\ell_n\rangle$. In each case the other cut leg, $\hat\ell_i$, is deformed so that all three legs remain on shell and momentum is conserved. This is a BCFW deformation of the two cut legs at
each corner. Of the four helicity assignments of
the two shifted legs, only the case in which the $\lambda$-shifted leg is
negative and the $\tilde\lambda$-shifted leg is positive can fail to give
a $1/z$ falloff \cite{Arkani-Hamed:2008owk,Britto:2004ap,Britto:2005fq}. Setting $\eta_{\hat\ell_1}=\eta_{\hat\ell_n}=0$ makes the
$\lambda$-shifted leg at each corner a positive-helicity gluon, so we are
in one of the three remaining cases and obtain the $1/z$ falloff at each
corner.

Let us show this explicitly for the $c=0$ case. From \cref{eq:3pmhv}, we
see that the three-point amplitude is finite when $z\to\infty$. As
established above, it suffices to make the $\lambda$-shifted leg at each
corner a positive-helicity gluon, which we ensure by sending
$\eta_{\hat\ell_1},\eta_{\hat\ell_n}\to 0$. By \cref{eq:mhvl1,eq:mhvl3}, 
$\lbrack\hat\ell_1\rvert=\lbrack\hat\ell_n\rvert=\lbrack n\rvert$, so
under the translation of \cref{eq:Qinv} both are shifted by the same
amount. Considering
$\zeta=c_1\tilde\lambda_n+c_2\,x_{1i}\cdot\lambda_n$, the Grassmann
variables are translated as
\begin{align}
\eta_{\hat\ell_1}&\to\eta_{\hat\ell_1}+\lbrack n\,\zeta\rbrack
=\eta_n+c_2\,(x_{in}^2-x_{1i}^2),\\
\eta_{\hat\ell_n}&\to\eta_{\hat\ell_n}+\lbrack n\,\zeta\rbrack
=\eta_n+c_2\,(x_{in}^2-x_{1i}^2),
\end{align}
where we have used \cref{eq:mhvl1,eq:mhvl3} and
\cref{eq:mhvetal1,eq:mhvetal3}, together with
$\langle n\vert x_{1i}\vert n\rbrack=x_{in}^2-x_{1i}^2$. Since the
three-point vertex has already localized
$\eta_{\hat\ell_1}=\eta_{\hat\ell_n}=\eta_n$, the two conditions coincide
and reduce to the single equation $c_2(x_{in}^2-x_{1i}^2)=-\eta_n$,
giving $c_2=-\eta_n/(x_{in}^2-x_{1i}^2)$, with $c_1$ left free since
$\lbrack n n\rbrack=0$. This sets both
$\eta_{\hat\ell_1}$ and $\eta_{\hat\ell_n}$ to zero, so each corner has a
positive-helicity $\lambda$-shifted leg. Since the integrand of
\cref{eq:triple} carries an explicit $dz/z$, the residue at infinity
vanishes provided the remaining factors vanish as $z\to\infty$. Here they
do so as $1/z^{2}$, and the Grassmann integrations introduce no positive
powers of $z$.
\paragraph{Quadruple cuts of boxes compatible with the triple cut of the two-mass triangle.} 
There are simple poles in $z$ where any of the remaining propagators go on-shell.
Thus
we consider boxes all compatible with the triple cut ${\mathcal C}_{n1i}$ by partitioning the external legs in either $A_{1,i}$ or $A_{i,n}$, as shown in \cref{fig:quadruplecuts}.
Setting a fourth propagator on-shell fully specifies the loop momentum in each box at a different point in the complex plane according to the maximal cut solution. The two solutions of the quadruple cut correspond to the two choices of helicity for the three-point vertex.

\paragraph{Poles in $A_{1,i}$.} We first look at the case of setting $\hat\ell_j^2(z)=0$, where $1<j<i$.  With the parametrization above, in the case that $c=0$, 
we have
\begin{equation}\label{eq:jprop}
    \hat\ell_j^2(z) = (\hat\ell(z)+x_{1 j})^2= \lbrack n |  x_{1 j} x_{1 i} |n \rbrack(z-z_j),
\end{equation}
where the pole is located at
\begin{align}\label{eq:jpole}
    z_j 
     &= \frac{x^2_{1 j} x_{n i}^2-x_{1 i}^2 x_{n j}^2}{(x^2_{n i}-x^2_{1 i}) \lbrack n | x_{1 i} x_{1 j}|n \rbrack}.
\end{align}

\paragraph{Poles in $A_{i,n}$.} We repeat the above for the 
fourth cut at $\hat\ell_k^2(z)=0$, where 
$\hat\ell_k(z)=\hat\ell(z)+x_{1 k}$ and $i<k<n$. We can simply replace $j$ by $k$ in the previous calculation.  Thus we have
\begin{equation}\label{eq:kprop}
    \hat\ell_k^2(z) = (\hat\ell(z)+x_{1 k})^2= \lbrack n |  x_{1 k} x_{1 i} |n \rbrack(z-z_k),
\end{equation}
where the pole is located at
\begin{align}\label{eq:kpole}
    z_k 
     &= \frac{x^2_{1 k} x_{n i}^2-x_{1 i}^2 x_{n k}^2}{(x^2_{n i}-x^2_{1 i}) \lbrack n | x_{1 i} x_{1 k}|n \rbrack}.
\end{align}

\paragraph{Sum of residues.}
Summing all residues from the poles found above, we find
\begin{alignat}{3}
0 
&=\sum_{j=2}^{i-1}\Res_{z=z_j}\ \Biggl[ \frac{J}{z} \int \prod_{r=1,j,i,n}\mathrm{d}^4 \eta_{\hat{\ell}_r} &&
\frac{A_{1,j}(z) A_{j,i}(z)}{\hat{\ell}^2_j(z)}
A_{i,n}(z)A_{n,1}(z)\Biggr] \nonumber\\&+\sum_{k=i+1}^{n-1}\Res_{z=z_k}  \Biggl[\frac{J}{z} \int  \prod_{r=1,i,k,n} \mathrm{d}^4 \eta_{\hat{\ell}_r} &&A_{1,i}(z) \frac{A_{i,k}(z)A_{k,n}(z)}{\hat{\ell}^2_k(z)} A_{n,1}(z)\Biggr]
\end{alignat}
We have included a Grassmann integration for the fourth propagator. 
The combination of a hat and a bar is used for momenta fully constrained by the quadruple cut, whose value will depend on the specific solution we found before, through $z_j$ or $z_k$. 
Substituting the expressions found in \cref{eq:jprop,eq:jpole,eq:kprop,eq:kpole} and evaluating the residues, we arrive at
\begin{align}
0 
=&\sum_{j=2}^{i-1}
\frac{1}{x^2_{1 j} x_{n i}^2-x_{1 i}^2 x_{n j}^2} \int  \prod_{r=1,j,i,n} \mathrm{d}^4 \eta_{\hat{\ell}_r} A_{1,j}(z_j) A_{j,i}(z_j) A_{i,n}(z_j)A_{n,1}(z_j) \nonumber\\&+\sum_{k=i+1}^{n-1} \frac{1}{x^2_{1 k} x_{n i}^2-x_{1 i}^2 x_{n k}^2} \int  \prod_{r=1,i,k,n} \mathrm{d}^4 \eta_{\hat{\ell}_r}  A_{1,i}(z_k) A_{i,k}(z_k)A_{k,n}(z_k) A_{n,1}(z_k),
\end{align}
where now only the Grassmann integral remains. We now recognize the individual terms as leading singularities associated to the quadruple cuts. However, notice that the prefactors represent {\em opposite} choices of the branch of the square root when we take the cyclic ordering of indices into account: according to the conventions listed in \cref{eq:sqrt-rho}, we have $x^2_{1 j} x_{n i}^2-x_{1 i}^2 x_{n j}^2 = -\sqrt{\rho_{1jin}}$, while $x^2_{1 k} x_{n i}^2-x_{1 i}^2 x_{n k}^2 = \sqrt{\rho_{1ikn}}$.
Moreover, since we have made the choice $c=0$, the terms of the expression correspond to one type of leading singularity of the boxes appearing in the two-mass triangle equations, namely the ones with an MHV three-point subamplitude as defined in \cref{eq:maxcutsusy}. Our result is therefore
\begin{equation}
    \sum_{j=2}^{i-1} {\rm LeS}_{1, j, i, n}^{+}[A^{(1)}_n]= \sum_{k=i+1}^{n-1}{\rm LeS}^{+}_{1, i, k, n}[A^{(1)}_n].  
\end{equation}
Repeating the entire analysis for the case with $d=0$, where the three-point amplitude is \textoverline{MHV}, and taking $c$ as the complex variable, gives the corresponding relations for ${\rm LeS}^-$. Taken together, we have proved \cref{prop:2mtrls}.
Adding the two equations with a factor of $1/2$ each, we recover the definition of the box coefficients given in \cref{eq:boxcoef}, thus proving the two-mass triangle relations of \cref{eq:2mrelation}.	

\bibliographystyle{JHEP}
\bibliography{biblio.bib}

\end{document}